\documentclass[aps,preprintnumbers,nofootinbib,onecolumn, superscriptaddress]{revtex4}

\usepackage{graphicx}
\usepackage{caption}
\usepackage{pdfpages}
\usepackage[shortlabels]{enumitem}
\usepackage{mathtools}
\usepackage{fancyhdr}
\usepackage{nicefrac}
\usepackage{bbm}
\makeatletter

\usepackage{upgreek}
\newcommand{\changefont}{%
    \fontsize{9}{11}\selectfont
}
 \usepackage{makecell}
 \usepackage{placeins}
  \usepackage{afterpage}
\usepackage{mathrsfs}

\usepackage{graphics}
\usepackage{amsfonts,amssymb,amsmath}            
\usepackage{ifthen}
\usepackage{amsthm, thm-restate}
\usepackage{dsfont}
\usepackage{float}
\usepackage{MnSymbol}

\usepackage[caption=false]{subfig}
\PassOptionsToPackage{hyphens}{url}
\usepackage[hyperfootnotes=false]{hyperref}
\usepackage{xcolor}
\definecolor{mylinkcolor}{rgb}{0,0,0.7} 
\hypersetup{unicode=true,%
  bookmarksnumbered=false,bookmarksopen=false,bookmarksopenlevel=1, %
  breaklinks=true,pdfborder={0 0 0},colorlinks=true}%
\hypersetup{%
  anchorcolor=mylinkcolor,citecolor=mylinkcolor, %
  filecolor=mylinkcolor,linkcolor=mylinkcolor, %
  menucolor=mylinkcolor,runcolor=mylinkcolor, %
  urlcolor=mylinkcolor} 

\usepackage{tikz}
\usetikzlibrary{arrows}
\usetikzlibrary{arrows.meta}
\usepackage{rotating, xcolor}
\usetikzlibrary{shapes}
\usetikzlibrary{hobby}
\usetikzlibrary{decorations.markings}
\usetikzlibrary{arrows.meta}
\usetikzlibrary{snakes} 
\usepackage{tikz-cd}
\tikzset{degil/.style={
            decoration={markings,
            mark= at position 0.5 with {
                  \node[transform shape] (tempnode) {$\setminus$};
                  }
              },
              postaction={decorate}
}
}

\newcommand\longrsquigarrow{
\begin{tikzpicture}
\draw [decorate, decoration={zigzag, segment length=+6pt, amplitude=+.95pt,post length=+2pt}, arrows={-Classical TikZ Rightarrow}]  (0,0.1) -- (0.6,0.1); \draw[draw=none] (0,0)--(0.6,0);
\end{tikzpicture}
}

\newcommand\longlsquigarrow{
\begin{tikzpicture}
\draw [decorate, decoration={zigzag, segment length=+6pt, amplitude=+.95pt,post length=+2pt}, arrows={-Classical TikZ Rightarrow},rotate around={180:(0.3,0.1)}]  (0,0.1) -- (0.6,0.1); \draw[draw=none] (0,0)--(0.6,0);
\end{tikzpicture}
}

\makeatletter
\newcommand*{\da@rightarrow}{\mathchar"0\hexnumber@\symAMSa 4B }
\newcommand*{\da@leftarrow}{\mathchar"0\hexnumber@\symAMSa 4C }
\newcommand*{\xdashrightarrow}[2][]{%
  \mathrel{%
    \mathpalette{\da@xarrow{#1}{#2}{}\da@rightarrow{\,}{}}{}%
  }%
}
\newcommand{\xdashleftarrow}[2][]{%
  \mathrel{%
    \mathpalette{\da@xarrow{#1}{#2}\da@leftarrow{}{}{\,}}{}%
  }%
}
\newcommand*{\da@xarrow}[7]{%
  \sbox0{$\ifx#7\scriptstyle\scriptscriptstyle\else\scriptstyle\fi#5#1#6\m@th$}%
  \sbox2{$\ifx#7\scriptstyle\scriptscriptstyle\else\scriptstyle\fi#5#2#6\m@th$}%
  \sbox4{$#7\dabar@\m@th$}%
  \dimen@=\wd0 %
  \ifdim\wd2 >\dimen@
    \dimen@=\wd2 %
  \fi
  \count@=2 %
  \def\da@bars{\dabar@\dabar@}%
  \@whiledim\count@\wd4<\dimen@\do{%
    \advance\count@\@ne
    \expandafter\def\expandafter\da@bars\expandafter{%
      \da@bars
      \dabar@ 
    }%
  }%
  \mathrel{#3}%
  \mathrel{%
    \mathop{\da@bars}\limits
    \ifx\\#1\\%
    \else
      _{\copy0}%
    \fi
    \ifx\\#2\\%
    \else
      ^{\copy2}%
    \fi
  }%
  \mathrel{#4}%
}
\makeatother

\newcommand\mlnode[1]{\fbox{\begin{tabular}{@{}c@{}}#1\end{tabular}}}

\usepackage[nodayofweek]{datetime}
\newcommand{\comment}[1]{}
\newcommand{\ket}[1]{| #1 \rangle}

\newcommand{\tr}{{\rm tr}}
\newcommand{\cA}{\mathcal{A}}
\newcommand{\cB}{\mathcal{B}}
\newcommand{\cC}{\mathcal{C}}
\newcommand{\cD}{\mathcal{D}}

\newcommand{\cF}{\mathcal{F}}
\newcommand{\cG}{\mathcal{G}}
\newcommand{\cH}{\mathscr{H}}

\newcommand{\cR}{\mathcal{R}}
\newcommand{\cS}{\mathcal{S}}
\newcommand{\cT}{\mathcal{T}}
\newcommand{\cX}{\mathcal{X}}
\newcommand{\cY}{\mathcal{Y}}
\newcommand{\cZ}{\mathcal{Z}}

\newcommand{\rC}{\mathrm{C}}
\newcommand{\rG}{\mathrm{GPT}}
\newcommand{\rQ}{\mathrm{Q}}

\newcommand{\indep}{\upmodels}
\newcommand{\nindep}{\not\upmodels}

\theoremstyle{plain}

\newtheorem{lemma}{Lemma}[section]
\newtheorem{corollary}{Corollary}[section]
\newtheorem{proposition}{Proposition}[section]

\theoremstyle{definition}
\newtheorem{definition}{Definition}[section]
\newtheorem{remark}{Remark}[section]
\newtheorem{example}{Example}[section]

\begin{document}

\title{A general framework for cyclic and fine-tuned causal models and their compatibility with space-time}

\author{V. Vilasini}
\email{vilasini@phys.ethz.ch}
\affiliation{Institute for Theoretical Physics, ETH Zurich, 8093 Z\"{u}rich, Switzerland}
\affiliation{Department of Mathematics, University of York, Heslington, York YO10 5DD, United Kingdom}
\author{Roger Colbeck}
\email{roger.colbeck@york.ac.uk}
\affiliation{Department of Mathematics, University of York, Heslington, York YO10 5DD, United Kingdom}

\date{\today}

\begin{abstract}
  Causal modelling is a tool for generating causal explanations of observed correlations and has led to a deeper understanding of correlations in quantum networks. Existing frameworks for quantum causality tend to focus on acyclic causal structures that are not fine-tuned i.e., where causal connections between variables necessarily create correlations between them. However, fine-tuned causal models (which permit causation without correlation) play a crucial role in cryptography, and cyclic causal models can be used to model physical processes involving feedback and may also be relevant in exotic solutions of general relativity. Here we develop a causal modelling framework capable of dealing with these general scenarios. The key feature of our framework is that it allows operational and relativistic notions of causality to be independently defined and for connections between them to be established in a theory-independent manner.  The framework first gives an operational way to study causation that allows for cyclic, fine-tuned and non-classical causal influences.  We then consider how a causal model can be embedded in a space-time structure (modelled as a partial order) and propose a \emph{compatibility} condition for ensuring that the embedded causal model does not allow signalling outside the space-time future.  We identify several distinct classes of causal loops that can arise in our framework, showing that compatibility with a space-time can rule out only some of them. We discuss conditions for preventing superluminal signalling within arbitrary (and possibly cyclic) causal structures and consider models of causation in post-quantum theories admitting so-called jamming correlations. 
Finally, this work introduces the concept of a ``higher-order affects relation'', which is useful for causal discovery in fined-tuned causal models.
\end{abstract}

\maketitle

\tableofcontents

\section{Introduction}

The process of identifying cause-effect relationships underlying our observations is central to science. The causal modelling paradigm~\cite{Pearl2009,Spirtes2001} provides mathematical tools for relating correlation and causation in scenarios described by classical variables, and have found applications in wide ranging disciplines including medical testing~\cite{Kleinberg2011, Raita2021}, economic predictions~\cite{Spirtes2005, Pearl2009} and machine learning~\cite{Maya2014, Kusumawardani2020, Liu2021}. A consequence of Bell's theorem~\cite{Bell} is that in certain scenarios, classical causal models fail to explain quantum correlations~\cite{Wood2015}. This has led to a significant progress in the development of quantum causal models~\cite{Tucci_1995,Leifer_2006,Laskey2007,Leifer_2008,Leifer2013,Henson2014,Wood2015,Pienaar2015,Ried_2015,Costa2016,Fritz_2015,Allen2017,Barrett2020A,Pienaar_2020} that have deepened our fundamental understanding of quantum causality and quantum correlations, as well as in practical information processing tasks such as quantum cryptography, communication, quantum computation.

Previous work on quantum causality has focused on acyclic causal structures and on causal models without fine-tuned parameters, where causation and signalling become equivalent notions. While it may be considered undesirable for a physical theory of nature to allude to fine-tuned causal explanations~\cite{Wood2015}, the security of cryptographic protocols such as the one-time pad rely on fine-tuning. Here, fine-tuning is required to ensure that the cipher text gives no information about the original message without the key, even though the cipher text was generated from the original message and thus causally depends on it. Cyclic causal models have been developed and widely studied in the classical causal modelling literature for describing physical scenarios with feedback~\cite{Forre2017, Bongers2021}, for instance, where variables such as demand and price causally influence each other. In the quantum literature, cyclic causation has been considered in the context of more exotic phenomena such as closed timelike curves or processes with indefinite causal order~\cite{Araujo2017, Barrett2020}, which may be useful in approaches to quantum gravity without a definite space-time structure. The causal modelling approach enables an operational formulation of causality that is independent of space-time structure~\cite{Pearl2009, Spirtes2001}. Whether a cyclic causal model describes a physical scenario with feedback or a closed timelike curve depends on how the causal model is combined with space-time information (see also \cite{VilasiniRenner2022}). Thus, from a purely operational standpoint, the most general class of causal models we would like to consider include those that are cyclic, fine-tuned and also allow for non-classical causal influences. To make a connection to physical experiments, it is also desirable to characterise how this general class of causal models can be embedded in a space-time structure, such as Minkowski space-time and to characterise when they prevent violations of relativistic causality principles such as no signalling outside the future in the space-time. 

In the case of acyclic causal models without fine-tuning, the condition for ensuring no superluminal signalling in a space-time is straightforward: whenever $A$ is a cause of $B$ in the causal model, we can interpret $B$ as being in the future of $A$ with respect to a space-time such as Minkowski space-time. This ensures that all causal influences and therefore all signals propagate from past to future in the space-time. Operationally, interventions allow us to verify causation and define a notion of signalling: if intervening on $A$ leads to different correlations on $B$ (compared to without the intervention), then we can say that $A$ signals to $B$ and use this to infer that $A$ is a cause of $B$. In the absence of fine-tuning, every causal relationship can be verified using interventions, and in such models, causation implies the ability to signal with an intervention. In the presence of fine-tuning, it is possible to have causation without signalling and in this case, demanding that there is no signalling outside the space-time future does not guarantee that all causal influences propagate from past to future in the space-time. The connection between superluminal signalling and causation has been previously studied by analysing correlations in Bell-type experiments in Minkowski space-time (see for instance \cite{Grunhaus1996, Horodecki2019}). However to find conditions for ensuring no signalling outside the space-time future in arbitrary scenarios, correlations alone do not suffice; to ascertain causation we must also consider interventions. Furthermore, allowing for cyclic causal influences while considering a partially ordered space-time such as Minkowski space-time allows for an investigation of the relationships between causal loops and superluminal signalling. A mathematical framework for causally modelling these general scenarios and establishing their connection to relativistic causality principles in a space-time is currently lacking.

In this work, we develop such a framework by defining causation and space-time structure as separate notions, and then characterising their compatibility.  We keep the causal part of the framework general by allowing for causation without signalling (i.e., fine-tuned causal influences), cyclic causation as well as quantum and post-quantum causes. We describe this through a causal modelling approach, but under minimal theory-independent assumptions, and while taking into account correlations as well as arbitrary interventions.  We then connect this to physics by considering the embedding of the observed variables involved in the causal model into a space-time structure, such as Minkowski space, and we characterise when such embeddings do not allow superluminal signalling. The framework proposed here has two main advantages.
On the one hand, keeping causation and space-time structure separate is a useful feature for considering more general formulations of physics without a fixed background space-time structure (e.g., in a theory of quantum gravity~\cite{Oreshkov2012, Zych2019}), while keeping a notion of processing and communicating physical information available. On the other hand, characterising the compatibility between operational causation and space-time structure can give insights into which of these scenarios is physically realisable in a space-time.

The framework introduced in this work allows a characterisation of causality in a class of post-quantum theories (producing so-called \emph{jamming} non-local correlations) previously proposed in the literature~\cite{Grunhaus1996, Horodecki2019}, clarifies the relationships between several concepts, and enables us to address a number of open questions. Even within causality conditions related to space-time, there can be several distinct notions. For example, physical principles such as ``no superluminal signalling'' and ``no causal loops/closed time-like curves'' are both associated with relativistic causality and implied by the mathematical framework of special relativity. However, these can be distinct concepts in a more general mathematical framework where the causal structure is not fully specified by the space-time structure, but only constrained by requirements such as no superluminal signalling once embedded in a space-time. Within our framework, we distinguish these concepts. In an associated Letter~\cite{VilasiniColbeckPRL}, we apply this framework to show the mathematical possibility of causal loops between Minkowski space-time events, the existence of which can be operationally detected without leading to superluminal signalling.\footnote{Here, by Minkowski space-time, we only mean the partial order corresponding to the light cone structure of Minkowski space-time.} 
Our framework also suggests further conditions that could be used to rule out certain types of causal loops.

When we refer to operationally detectable, we mean detectable using inferences from the observed correlations and those under intervention. Some properties of an underlying causal structure can be operationally found from the observed correlations. For example, a violation of Bell inequalities within the Bell causal structure certifies the non-classicality of the underlying common cause from the observed correlations.  To distinguish causation and correlation we need to consider interventions, which allow more general inferences about the causal structure~\cite{Pearl2009}. Recently, it has been experimentally demonstrated~\cite{Agresti2021} that the non-classicality of a causal structure can be operationally certified from causation measures based on interventions even when no such certification is possible using correlation measures alone. 

 Apart from these foundational implications, several features of our framework are useful from a more practical perspective. For instance, 
 security of relativistic cryptographic protocols~\cite{Kent_RBC, CK1} combines both relativistic notions of causality (such as the impossibility of signalling outside the future light cone) and information-theoretic concepts. Operational information about the causal structure (which encodes the structure of communication channels between agents), the embedding of the causal structure in a space-time structure, and the compatibility between the two are all relevant for cryptography.

To operationally model causation, we adopt a causal modelling approach similar to that of~\cite{Pearl2009, Spirtes2001}, in which causal structures are represented using directed graphs. These indicate how information flows through a network of physical systems (classical, quantum or possibly those of a post-quantum probabilistic theory), and the directed graph is in principle independent of any consideration of space-time. One can however consider embedding the systems represented in the causal structures within a space-time, and relativistic causality would then impose constraints on the embedding such that the causal model cannot be used to signal outside the space-time future, in which case we say that the causal model is \emph{compatible} with the space-time structure. For example, if an active intervention on a variable $A$ produces a change in probability distribution over another variable $B$, then one would say that $A$ affects $B$ (or $A$ signals to $B$), which implies that $A$ is a cause of $B$. Assigning space-time locations to the variables and requiring the effect $B$ to always be embedded in the future light cone of the cause $A$ makes this causal relationship compatible with the partial order of space-time. In some situations (such as for jamming \cite{Grunhaus1996}) we wish to allow a variable to jointly affect a set of variables without affecting individual variables in the set, and, more generally, we may consider more complicated affects/non-affects relations between arbitrary sets of variables. Such scenarios correspond to causal models where the correlations are \emph{fine-tuned} to hide certain causal influences from direct observation such that there is causation without correlation or signalling. In the presence of fine-tuning, characterising when a causal model can be compatibly embedded in a space-time structure is more complicated. In our work we provide a method to do so by developing a general framework and introducing causal modelling tools that have applications for analysing causality in a previously proposed class of post-quantum scenarios as well more practical problems related to causal discovery, as we explain below.

Previously, minimal conditions for preventing superluminal signalling have been considered in Bell-type scenarios. This led to the introduction of a general class of post-quantum correlations that can be defined in a tripartite Bell experiment (see Figure~\ref{fig: trins}) that were dubbed jamming non-local correlations~\cite{Grunhaus1996}. In later work, the constraints defining this class of correlations were claimed to be necessary and sufficient for ruling out superluminal signalling and causal loops~\cite{Horodecki2019}, under certain assumptions on the space-time configuration. Previous works analysing post-quantum theories admitting jamming correlations only consider the observed correlations produced in such Bell-type scenarios. However, to rigorously analyse causation and signalling possibilities in such theories, correlations alone do not suffice (since correlation does not imply causation), and interventions must also be taken into account. A defining feature of jamming correlations is that they allow the measurement setting of one party to jointly signal to the measurement outcomes of two other parties, without signalling to them individually (this can only happen with fine-tuning). In the space-time configuration considered in~\cite{Grunhaus1996, Horodecki2019}, this leads to superluminal causal influences without superluminal signalling. Since we allow fine-tuning, more generally, we can consider whether it is possible to have causal loops in a causal structure that do not lead to superluminal signalling when the systems in the causal structure are embedded in Minkowski space-time. Therefore for a clear understanding of the general validity of such claims for ruling out causal loops, a rigorous causal modelling framework is required. A general framework for modelling causality and its compatibility with space-time, as described in the above paragraphs will also enable us to consider conditions for preventing signalling outside the future lightcone and causal loops in arbitrary scenarios (not just those associated with Bell experiments). To our knowledge, such a mathematical framework is lacking in the previous literature.

A framework allowing for cyclic quantum causal models was proposed in~\cite{Barrett2020}. There the focus was on indefinite causal order processes and the authors adopt a fully quantum approach where all nodes are associated with quantum systems. 
To model post-quantum theories admitting jamming correlations~\cite{Grunhaus1996, Horodecki2019} and analyse the signalling possibilities therein, we distinguish between classical nodes corresponding to measurement settings and outcomes, and non-classical nodes (which may be quantum, or more generally post-quantum systems modelled by a generalised probabilistic theory).
This is similar to the approach of~\cite{Henson2014} but, in contrast to~\cite{Henson2014}, we allow for cyclic causal models, fine-tuning and also consider space-time embeddings.


Finally we note some implications for the problem of causal discovery (inferring causation from empirical data), which is ubiquitous in science. Causal discovery algorithms are often based on the assumption of ``no fine-tuning'' or faithfulness (see~\cite{Pearl2009, Spirtes2001}). Allowing fine-tuning significantly complicates causal discovery by allowing for causal influences that are not immediately reflected in certain types of empirical data. The framework, and results presented here make explicit several new aspects of fine-tuned causal models and elucidate relationships between several concepts relating to causal models that are equivalent in the absence of fine-tuning, but that become inequivalent in the presence of fine-tuning. This suggests new methods for exploring the problem of causal discovery in the presence of fine-tuning, a problem that is of interest to the scientific community beyond the foundations of quantum physics.

\bigskip\par
\textbf{Summary of contributions.} We first review the necessary preliminaries of the causal modelling approach in Section~\ref{sec: prelim} and discuss the jamming scenario along with other motivating examples in Section~\ref{sec:motiveg}. In the rest of the paper, we present several results that address the open questions outlined above, which are summarised below.
\begin{itemize}
    \item In Sections~\ref{sec: causmod} and~\ref{sec:space-time} we develop an operational framework for analysing cyclic and fine-tuned causal models in the presence of latent non-classical causes, and characterising their compatibility with a space-time structure. In particular, this provides a mathematical framework for causally modelling post-quantum theories admitting jamming non-local correlations~\cite{Grunhaus1996} (referred to as relativistic causal correlations in~\cite{Horodecki2019}). The framework consists of two parts---the first concerns causal models and the second characterises the embedding of these causal models in a space-time structure. 
    \item In the causality part of the framework (Section~\ref{sec: causmod}), we extend a number of results previously established in the classical causal modelling literature, typically used for acyclic and faithful causal models, to the more general scenarios considered here, such as Pearl's rules of do-calculus~\cite{Pearl2009}. We also introduce several causal modelling concepts, such as ``higher-order affects relations'' which only become relevant in fine-tuned causal models. We derive relationships between the many distinct properties of such causal models, highlighting the deviation from the standard case of faithful causal models. These technical results have applications for the problem of causal discovery in fine-tuned causal models, which is of independent interest. 
    
    \item  In the second part of the framework (Section~\ref{sec:space-time}), we use higher-order affects relations to define when a causal model can be said to be compatible with an embedding in a space-time structure, which is intended to capture that the model does not allow signalling outside the space-time future. We also consider alternative compatibility conditions (in Section~\ref{sec: necsuff}), and discussing the relationships between them and their physical intuition.
    \item In Section~\ref{sec: loops}, we define several distinct classes of causal loops and consider theories that are consistent with the principle that signalling outside the space-time future is not possible. We show that such theories are necessarily free of certain types of causal loops, and, in an associated letter~\cite{VilasiniColbeckPRL}, we apply our framework to construct a causal model for an operationally detectable causal loop that can be embedded in Minkowski space-time without leading to superluminal signalling. We discuss this example and illustrate in Appendix~\ref{appendix: MoreLoops} that such theories (which allow for causal loops without signalling outside the future of a partially ordered space-time) can involve further distinct classes of causal loops beyond those defined in the main text.
    \item The above results illustrate the counter-intuitive possibilities allowed by fine-tuned causal models---it is logically possible to have superluminal causal influences without superluminal signalling (as in non-local hidden variable theories~\cite{Bohm1952} or the jamming scenario of~\cite{Grunhaus1996, Horodecki2019}), as well as causal loops that do not lead to superluminal signalling. These results have consequences for the claim of~\cite{Horodecki2019} that certain conditions on correlations in a tripartite Bell scenario are necessary and sufficient for ruling out all causal loops. 
   This claim does not hold in our framework without further assumptions (see Section~\ref{ssec: missing_assump}).
\end{itemize}

 In upcoming work~\cite{Jamming2} we apply the results of the present paper to analyse the post-quantum jamming scenario of~\cite{Grunhaus1996, Horodecki2019} in detail, where we identify an explicit protocol that leads to superluminal signalling in this setting (contrary to previous claims), as well as new properties of post-quantum theories that admit such correlations.

 A reader who is more interested in the physical implications of the framework rather than causal modelling, may choose to skip the latter parts of Section~\ref{sec: causmod} on causal modelling, and directly move on to the space-time part of our framework in Section~\ref{sec:space-time}. In particular, while Sections~\ref{ssec: causmodA} and~\ref{ssec: interventions} are important for what follows, Examples~\ref{example: HOaffects1},~\ref{example: HOaffects2} and~\ref{example: HOaffects3} of Section~\ref{ssec: HOaffects} already give the main intuition behind the new concept of higher-order affects relations, and how it can be applied to define compatibility with a space-time in Section~\ref{sec:space-time}. The reader may therefore choose to skip the remaining technical details of Section~\ref{ssec: HOaffects}, as well as the subtleties of Section~\ref{ssec: relations} in their first reading.

\section{Preliminaries: Acyclic and faithful causal models}
\label{sec: prelim}
We first briefly review the literature on classical and non-classical causal models, where cause-effect relationships are typically taken to be acyclic and assumed not to be fine-tuned, before developing a model where these assumptions are relaxed.

\par
A causal structure can be represented as a directed
graph over several nodes, some of which are labelled observed
and some unobserved, typically this is taken to be a Directed Acyclic Graph (DAG). Each observed node corresponds to a classical
random variable\footnote{These may represent settings or outcomes of an experiment for example.}, while each unobserved node is
associated with a classical, quantum or post-quantum system. The causal structures we consider in this paper always feature observed random variables and we will denote the union $S_1\cup S_2$ of any two sets $S_1$ and $S_2$ of random variables by $S_1S_2$. A causal structure is called \emph{classical} (denoted
$\mathcal{G}^\rC$), \emph{quantum} (denoted $\mathcal{G}^\rQ$) or
\emph{GPT} (denoted $\mathcal{G}^\rG$) depending on the nature of the
unobserved nodes, where GPT stands for generalised probabilistic theory~\cite{Barrett07}. Edges of causal graphs will be denoted using $\longrsquigarrow$, and it will be useful to later classify these edges as solid $\longrightarrow$ or dashed $\xdashrightarrow{}$ based on certain operational conditions for detecting causation.  The following definition of \emph{cause} is implicit in the meaning of such a causal structure.

\begin{definition}[Cause]
\label{def: cause}
Given a causal structure represented by a directed graph $\cG$, possibly containing observed as well as unobserved nodes, we say that a node $N_i$ is a \emph{cause} of another node $N_j$ if there is a directed path $N_i\longrsquigarrow\ldots\longrsquigarrow N_j$ from $N_i$ to $N_j$ in $\cG$. More generally, we say that a set of nodes $S_1$ is a cause of a disjoint set of nodes $S_2$ if there exist nodes $N_i\in S_1$ and $N_j\in S_2$ such that $N_i$ is a cause of $N_j$.
\end{definition}

For an acyclic causal structure $\cG^C$ over the $n$ random variables $\{X_1,\ldots,X_n\}$ (i.e., having those variables as nodes), a distribution $P(X_1,\ldots,X_n)$ is said to be \emph{compatible with} $\mathcal{G}^\rC$ if it satisfies the \emph{causal Markov condition} i.e., the joint distribution decomposes as
\begin{equation}
\label{eq: markov}
    P(X_1,\ldots,X_n)=\prod\limits_{i=1}^{n} P(X_i|\text{par}(X_i)),
\end{equation}
where $\text{par}(X_i)$ denotes the set of all parent nodes of the node $X_i$ in the DAG $\mathcal{G}^\rC$. [We later discuss a notion of compatibility for more general (possibly cyclic) causal structures (Definition~\ref{definition: compatdist}), that is weaker but recovers the present definition in the classical acyclic case where all nodes are observed.] The Markov condition of Equation~(\ref{eq: markov}) is equivalent to the conditional
independence $X_i \indep \text{nd}(X_i)|\text{par}(X_i)$ of $X_i$ from its non-descendants, denoted
$\text{nd}(X_i)$ given its parents $\text{par}(X_i)$ in
$\mathcal{G}$ i.e., $\forall i\in \{1,\ldots,n\}$,
$P(X_i
  \text{nd}(X_i)|\text{par}(X_i))=P(X_i|\text{par}(X_i))P(\text{nd}(X_i)|\text{par}(X_i))$~\cite{Pearl2009}. 
In the case of classical causal structures with unobserved nodes, the set of compatible observed distributions for the causal structure are obtained by marginalisation of a total distribution (over all nodes) that satisfies Equation~\eqref{eq: markov}. 

In non-classical causal structures, this compatibility condition no longer applies since a node (e.g., a measurement outcome) and its parents (e.g., the quantum states that were measured to produce that outcome) in the causal structure may not coexist. Here, we can only assign a joint distribution over all the observed nodes, and this cannot in general be seen as a marginal of a joint distribution over all nodes, as in the classical case. Instead, the observed distribution in a non-classical causal structure is obtained using the states, transformations and measurements of the theory under consideration (which we will call the \emph{causal mechanisms}), in the order specified by the causal structure and in accordance with the probability rule specified by the theory. For example, in quantum theory, this would be the Born rule. Compatibility with non-classical causal structures can be formulated in terms of a generalised Markov condition~\cite{Henson2014} that requires the non-classical causal mechanisms (e.g., the quantum channels) to factorise in a manner analogous the classical Markov condition~\eqref{eq: markov}, but the exact form of this will not be relevant here. There are several frameworks for describing quantum and post-quantum causal structures which typically differ in how the nodes and edges are associated with the causal mechanisms of the theory. However, the details of these different frameworks do not change the operational predictions that can be made from the causal structure, such as the possible observed correlations realizable in the causal structure, and will not be needed in the rest of this paper. As an illustration, the following example describes the sets of compatible observed correlations in the classical and quantum version of the well-known bipartite Bell causal structure $\cG_B$ of Figure~\ref{fig: Bell1}. In the following, $\mathscr{P}_n$ denotes the set of all probability distributions over $n$ random variables and $\mathscr{S}(\mathscr{H})$ denotes the set of positive semi-definite and trace one operators on a Hilbert space $\mathscr{H}$.

\begin{example}[Sets of compatible correlations in the bipartite Bell causal structure $\cG_B$]
\label{example: Bellsets}
In the classical causal structure $\cG_B^C$, the set of compatible (observed) distributions is obtained by assuming a joint distribution $P(\Lambda XYAB)\in \mathscr{P}_5$ over all nodes, that satisfies the Markov condition~\eqref{eq: markov} and  marginalising over the unobserved node $\Lambda$,
\begin{equation}
    \mathscr{P}(\cG_B^C):=\{P(XYAB)\in\mathscr{P}_4\ |\ P(XYAB)=\sum_\Lambda P(\Lambda)P(A)P(B)P(X|A\Lambda)P(Y|B\Lambda)\}.
\end{equation}
If $\Lambda$ is a continuous random variable, the sum is replaced by an integral over $\Lambda$. This compatibility condition for the classical causal structure $\cG_B^C$ is identical to the \emph{local causality condition} used in the derivation of Bell inequalities (see~\cite{Brunner2014} for a comprehensive review). In the quantum causal structure $\cG_B^Q$, the unobserved node $\Lambda$ corresponds to a bipartite quantum state $\rho_{\Lambda}\in\mathscr{S}(\cH_{\Lambda})=\mathscr{S}(\cH_{\Lambda_X}\otimes\cH_{\Lambda_Y})$, and the observed nodes $X$ and $Y$ are associated with the POVMs, $\{E^X_A\}_X$ and $\{F^Y_B\}_Y$, that act on the subsystems $\cH_{\Lambda_X}$ and $\cH_{\Lambda_Y}$, depending on the inputs $A$ and $B$ respectively to generate the output distribution.
\begin{equation}
    \mathscr{P}(\cG_B^Q):=\{P(XYAB)\in\mathscr{P}_4\ |\ P(XYAB)=\tr\Big((E^X_A\otimes F^Y_B) \rho_{\Lambda}\Big)P(A)P(B)\}.
\end{equation}
\end{example}

In classical and non-classical causal structures alike, conditional independences play an important role. For instance, in the Bell causal structure, irrespective of the nature of $\Lambda$, we have $X\indep B|A$ and $Y\indep A|B$. Expressed in terms of probabilities these are the no-signalling constraints. The concept of \emph{d-separation} developed in \cite{Geiger1987, Pearl1988, Geiger1990, Verma1988} provides a method to read off implied conditional independence relations from the graph, both in classical and non-classical causal structures. It is defined as follows.

\begin{definition}[Blocked paths]
Let $\mathcal{G}$ be a DAG in which $X$ and $Y\neq X$ are nodes and $Z$ be a
set of nodes not containing $X$ or $Y$.  A path from $X$ to $Y$ is
said to be \emph{blocked} by $Z$ if it contains either $A\longrsquigarrow W\longrsquigarrow B$
with $W\in Z$, $A\longlsquigarrow W\longrsquigarrow B$
with $W\in Z$ or $A\longrsquigarrow W\longlsquigarrow B$ such that neither $W$ nor any descendant of $W$ belongs to $Z$, where $A$ and $B$ are arbitrary nodes in the path between $X$ and $Y$.
\end{definition}

\begin{definition}[d-separation]
\label{definition: d-sep}
Let $\mathcal{G}$ be a DAG in which $X$, $Y$ and $Z$ are disjoint
sets of nodes.  $X$ and $Y$ are \emph{d-separated} by $Z$ in
$\mathcal{G}$, denoted as $(X\perp^d Y|Z)_{\cG}$ (or simply $X\perp^d Y|Z$ if $\cG$ is obvious from the context) if every path from a variable in $X$ to a variable in
$Y$ is \emph{blocked} by $Z$, otherwise, $X$ is said to be \emph{d-connected} with $Y$ given $Z$.
\end{definition}

In classical acyclic causal structures (where the Markov condition of Equation~\eqref{eq: markov} holds), it has been shown that every d-separation relation $X \perp^d Y |Z$ between pairwise disjoint subsets of nodes implies that the conditional independence $X \indep Y |Z$ holds in the corresponding probability distribution~\cite{Geiger1990,Verma1988}. In non-classical acyclic causal structures, the same has been shown for d-separation relations between arbitrary disjoint sets of the observed nodes~\cite{Henson2014}. In our example of the Bell causal structure, we have the d-separation relations $X\perp^d B|A$ and $Y\perp^d A|B$, which imply the conditional independences $X\indep B|A$ and $Y\indep A|B$ characterising the no-signalling constraints.

Furthermore, in both cases, given a causal structure $\cG$ and a distribution $P$ compatible with it, the pair $(\cG,P)$ constitute a \emph{faithful causal model} if every conditional independence $X\indep Y|Z$ in $P$ corresponds to a d-separation relation $X\perp^d Y|Z$ in $\cG$. In the non-classical case, $P$ corresponds to the distribution over the observed nodes and cannot be seen as a marginal of a joint distribution over all nodes. Hence conditional independence in the sense of $P(XY|Z)=P(X|Z)P(Y|Z)$ can only be defined when $X$, $Y$ and $Z$ are pairwise disjoint subsets of the observed nodes. In the classical case, conditional independence in this form can also be defined for unobserved nodes and in a faithful, classical causal model, all such conditional independences imply a corresponding d-separation. Note that it is possible to define a notion of conditional independence between quantum nodes in terms of conditional quantum states (instead of conditional probability distributions)~\cite{Allen2017}, but in this paper, we will only consider conditional independence relations involving sets of classical variables, which could be the observed nodes of non-classical causal structures or any node of a classical causal structure. Then an \emph{unfaithful} or \emph{fine-tuned} causal model is one where there exists a conditional independence $X\indep Y|Z$ in the distribution $P$ even though $X$ and $Y$ are \emph{d-connected} in $\cG$. For example, Figure~\ref{fig: Bell2} provides an extension of the Bell causal structure, where there are additional causal influences from each party's input to the other party's output and it is known that any distribution realisable in the original causal structure is realisable in the classical version of this modified causal structure~\cite{Wood2015} (see Appendix~\ref{appendix: doMechanisms} for further details). Note however that the d-separation relation $Y\perp^d A|B$ no longer holds here, and hence any no-signalling distribution would be fine-tuned or unfaithful with respect to this causal structure but not with respect to the original one of Figure~\ref{fig: Bell1}. In other words, the first causal structure faithfully explains no-signalling correlations using non-classical causal mechanisms while the second provides an unfaithful explanation of such correlations using classical causal mechanisms. 

\begin{figure}[t!]
    \centering
\subfloat[]{\includegraphics[]{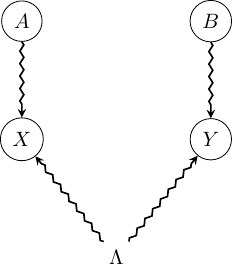}\label{fig: Bell1}}\qquad\qquad\qquad\qquad\subfloat[]{\includegraphics[]{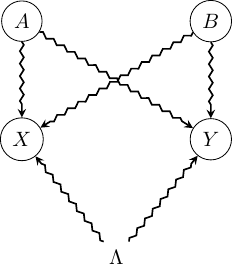}\label{fig: Bell2}}
    \caption{(a) The bipartite Bell causal structure: $\Lambda$ represents a bipartite state (classical, quantum or that of a generalised probabilistic theory) shared by two non-communicating parties Alice and Bob who measure their subsystems locally using classical measurement settings $A$ and $B$ to obtain classical outcomes $X$ and $Y$. (b) A variation of (a) in which the settings $A$ and $B$ are both causes of both outcomes.}
\label{fig: Bell}
\end{figure}

\section{Motivation for analysing fine-tuned and cyclic causal models}
\label{sec:motiveg}
One of the most common assumptions made in the analysis of causal models is that of \emph{faithfulness} or \emph{no fine-tuning}. Fine-tuning complicates causal inference because it involves independences that disappear with small amounts of noise, and fine-tuning is often avoided in the literature (also on the grounds that fine-tuned causal models constitute a set of measure-zero). Even in the Bell scenario explained above, a faithful explanation of the correlations using non-classical causal models is often preferred over the unfaithful explanation using classical causal models. However, there are a number of examples, as we will see below, that necessitate a fine-tuned explanation irrespective of whether the causal structure is classical and non-classical. These include certain everyday scenarios, cryptographic protocols as well as more exotic cases that arise in certain post-quantum theories that allow for superluminal influences without superluminal signalling, which we discuss in Sections~\ref{ssec:friedman} and~\ref{ssec:intro_jamming}.

Another common assumption in the causality literature is that the causal structure is acyclic. Allowing fine-tuned causal influences makes possible cyclic causal structures that are compatible with minimal notions of relativistic causality, such as the impossibility of signalling superluminally at the observed level. Cyclic causal models have also found applications in the classical literature for describing systems with feedback loops~\cite{Pearl2013, Forre2017}. Developing a framework for cyclic and fine-tuned causal models in non-classical theories therefore has both foundational and practical relevance,  enabling us to better understand the operational relationships between causality and signalling with respect to a space-time structure, and their implications for information processing. We now present some concrete examples that necessitate such causal models.

\subsection{Friedman's thermostat and the one-time pad}
\label{ssec:friedman}

Consider a house with an ideal thermostat. Such a thermostat would maintain a constant inside temperature $T_I$ throughout the year, despite variations in the outside temperature $T_O$, by adjusting its energy consumption $E$ accordingly. Going by the correlations alone, one might incorrectly conclude that the inside temperature $T_I$ is causally independent of everything else as it has no correlations with any other variables. However a closer look at the internal workings of the thermostat would reveal that the correct causal explanation is the one shown in Figure~\ref{fig: thermostat}, where $T_I$ is causally influenced by both $T_O$ and $E$, and the feedback loop between $T_I$ and $E$ ensures that the indoor temperature remains constant, by suitably adjusting $E$. The causal model in this case is fine-tuned since the independence of $T_I$ from $T_O$ and $E$ does not correspond to a d-separation relation in the causal structure (Figure~\ref{fig: thermostat}). This thermostat analogy which is attributed to Milton Friedman~\cite{Friedman2003}, can be extended to a number of other scenarios such as the effect of fiscal and monetary policies on economic growth~\cite{Rowe2009}, or physical systems where several forces exactly balance out. 

In cryptographic settings, examples that necessitate fine-tuning include the one-time pad or the ``traitorous lieutenant problem''~\cite{Brul2017}. Consider a general who wishes to relay an important secret message $M$ to an ally and has two lieutenants available as messengers, but one of them is a traitor who might leak the message to enemies. Consider for simplicity that $M$ is a single bit. The general could then adopt the following strategy: Depending on $M=0$ or $M=1$, generate two bits $M_1$ and $M_2$ such that $M_1=M_2$ or $M_1\neq M_2$ and with both uniformly distributed. Give $M_1$ to the first and $M_2$ to the second lieutenant to relay to the ally. Then the ally would receive $M_1$ and $M_2$ and can simply use modulo-2 addition $\oplus$ to obtain $M^*$ which is indeed the original message $M^*=M=M_1\oplus M_2$ (Figure~\ref{fig: traitor}). More importantly, the individual messages $M_1$ and $M_2$ contain no information about $M$ and hence neither lieutenant has any information about the secret message. A similar protocol underlies the one-time pad where a message $M$ is encrypted using a secret key $K$ (both binary for this example) to produce an encrypted message $M_E=M\oplus K$ which can be sent through a public channel as it will carry no information about the original message $M$ if the key $K$ is uniformly distributed and is kept private. Only a receiver of $M_E$ who knows the key $K$ can decrypt the message $M=M_E\oplus K$ (Figure~\ref{fig: otp}). Hence fine-tuning of causal influences i.e., causation in the absence of correlation, is crucial for the security. 

Further, cyclic causal models have been analysed in the classical literature~\cite{Pearl2013, Forre2017} for the purpose of describing complex systems involving feedback loops, analogous to the thermostat example. Note that the cyclic dependencies here do not correspond to closed time-like curves since the variables under question are considered over a period of time--- e.g., a demand at time $t_1$ influences the price at time $t_2>t_1$, which in turn influences the demand at time $t_3>t_2$. Within our framework we would use separate random variables for each of the times, which in some cases would remove the cyclicity. To characterise genuine closed time-like curves one must consider not only the pattern of causal influences, but also how the relevant variables are assigned space-time locations.
\begin{figure}[t!]
 \centering\subfloat[\label{fig: thermostat}]{\includegraphics[]{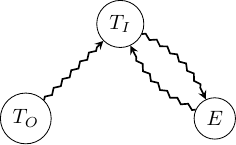}}\qquad
 \subfloat[\label{fig: traitor}]{\includegraphics[]{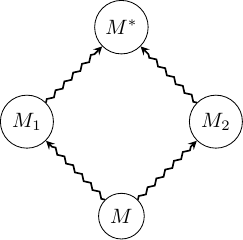}}\qquad \subfloat[\label{fig: otp}]{\includegraphics[]{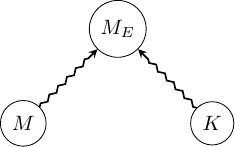}}
    \caption[Motivating examples for cyclic and fine-tuned causal models]{\textbf{Causal structures for the motivating examples described in the main text: } \textbf{(a)} Friedman's thermostat \textbf{(b)} Traitorous Lieutenant \textbf{(c)} One-time pad. Note that there may be additional causal influences. For example, in (b), we will later see that an additional common cause between $M_1$ and $M_2$ will be required to fully explain the correlations (cf.\ Figure~\ref{fig: jamming}).}
    \label{fig: motiv_eg}
\end{figure}

\subsection{Jamming non-local correlations}
\label{ssec:intro_jamming}
Another example that involves fine-tuning, even though it has not been motivated or discussed in this context, is that of jamming non-local correlations introduced in~\cite{Grunhaus1996}. The work~\cite{Grunhaus1996} outlines the possibility of post-quantum theories beyond the standard no-signalling probabilistic theories (such as box-world) that are still compatible with the impossibility of superluminal signalling. A better understanding of such theories would shed light on the principles of causality (beyond no superluminal signalling) that distinguish quantum and GPTs from these more general post-quantum theories. However, a mathematical framework for analysing causality in such theories is lacking, and the main purpose of this paper is to develop a general framework for modelling the relationships between causation and space-time structure, that can in particular be applied to jamming theories. In upcoming work~\cite{Jamming2}, we apply our framework to the jamming scenario in more detail identifying new aspects of theories that admit such scenarios. We proceed by reviewing the jamming scenario.

Consider three space-like separated parties, Alice, Bob and Charlie sharing a tripartite system $\Lambda$ which they measure using measurement settings $A$, $B$ and $C$, producing outcomes $X$, $Y$ and $Z$ respectively. Suppose that their space-time locations are such that Bob's future light cone entirely contains the joint future of Alice and Charlie, as shown in Figure~\ref{fig: trins}. The standard no-signalling conditions forbid the input of each party from being correlated with the outputs of any subset of the remaining parties, for instance, the joint distribution $P(XYZ|ABC)$ satisfies $P(XZ|ABC)=P(XZ|AC)$. In~\cite{Grunhaus1996} it is argued that a violation of this requirement does not lead to superluminal signalling in the space-time configuration of Figure~\ref{fig: trins}, as long as $P(X|ABC)=P(X|A)$ and $P(Z|ABC)=P(Z|C)$. This is because any influence that $B$ exerts jointly (but not individually) on $X$ and $Z$ can only be checked when $X$ and $Z$ are brought together to evaluate the correlations $P(XZ|ABC)$. This is only possible in their joint future, which is by construction contained in the future of $B$. Bob is said to \emph{jam} the correlations between Alice and Charlie non-locally.

In~\cite{Horodecki2019} the causal structure for such an experiment is represented by introducing a new random variable $C_{XZ}$ associated with the set $XZ$ that encodes the correlations between its elements. Then $B$ is seen as a cause of $C_{XZ}$ but not as a cause of either $X$ or $Z$. In general scenarios, this representation would require adding a new variable for every non-empty subset of the observed nodes, which can become intractable.\footnote{In general, this representation would include up to $2^n-1$ observed variables whenever the original set of observed variables has $n$ elements.} In fact, given the assumptions that $B$ is freely chosen and is hence a parentless node, and that for non-trivial jamming, it must be correlated with $XZ$, any causal structure where $B$ is not a cause of at least one of $X$ and $Z$ (the causal structure proposed in~\cite{Horodecki2019} being such an example) would not lead to a sensible causal model satisfying the d-separation property (Definition~\ref{definition: compatdist}), which is a basic property satisfied by classical and non-classical causal models alike~\cite{Pearl2009, Henson2014}. This is because such a causal structure would have a d-separation between $B$ and $XZ$ which would require these sets to be uncorrelated, and hence disallow any non-trivial jamming. Further, this representation does not always correspond to what is physically going on--- for instance, in the example of the traitorous lieutenant, this would introduce a new variable $C_{M_1M_2}$ that is observably influenced by the general's original message $M$, while $M$ would no longer be seen as a cause of $M_1$ or $M_2$. However, we know that we physically generated $M_1$ and $M_2$ using $M$\footnote{And possibly some additional information to explain the distribution over the individual variables. As we will see later in Figure~\ref{fig: jamming}, a common cause $\Lambda$ between $M_1$ and $M_2$ would also be required in such examples.}, hence it is indeed a cause of at least one of them.  Therefore, we aim to develop a new approach to causal modelling in a general class of fine-tuned and cyclic scenarios, using only the original variables/systems. The following proposition illustrates that the jamming scenario considered in~\cite{Grunhaus1996, Horodecki2019} necessarily corresponds to a fine-tuned causal model over the original variables. Here, jamming is considered in the context of multipartite Bell scenarios where the jamming variable is a freely chosen input of one of the parties. In the causal model approach adopted here, we will take free choice of a variable to correspond to the exogeneity of that variable in the causal structure.\footnote{This is a standard way of modelling free choices in a causal model, although note that it is not equivalent to other definitions of free choice~\cite{CR_ext,CR2013,Horodecki2019}.} 

\begin{figure}[t!]
\centering
	\subfloat[\label{fig: trins1}]{	\centering
\includegraphics[]{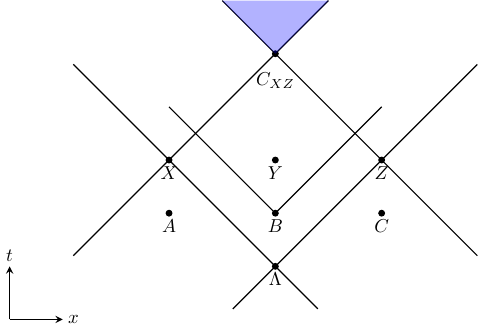}}\qquad\quad
\subfloat[\label{fig: trins2}]{	\centering
\includegraphics[]{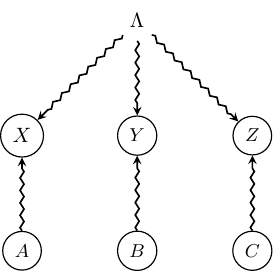}}

	\caption{\textbf{Jamming correlations in the tripartite Bell scenario:} Three parties Alice, Bob and Charlie share a tripartite system $\Lambda$, they measure their subsystem using the freely chosen measurement settings $A$, $B$ and $C$, producing the outcomes $X$, $Y$ and $Z$ respectively, without communicating. \textbf{(a)} Space-time configuration for the jamming scenario~\cite{Grunhaus1996, Horodecki2019}: the measurement of the three parties are pairwise space-like separated with the future of Bob's input $B$ containing the joint future of Alice's and Charlie's outputs $X$ and $Z$ (blue region). Here, it is argued that allowing $B$ to signal to $X$ and $Z$ jointly but not individually is consistent with the principle of ``no signalling outside the future lightcone'', since the joint signalling can only be verified in the blue region which is in the future of $B$. Such correlations form a larger set as compared to the standard tripartite no-signalling correlations, which forbid individual as well as joint signalling from the inputs of any set of parties to the outputs of a complementary set of parties~\cite{Salazar2020}. To model the joint signalling through jamming, a new variable $C_{XZ}$ was introduced in~\cite{Horodecki2019}, located at the earliest point in the joint future of $X$ and $Z$ and representing the correlations between $X$ and $Z$. \textbf{(b)} Causal structure for the usual tripartite Bell experiment.}
	\label{fig: trins}
\end{figure}

\begin{proposition}
\label{prop:finetune}
Consider a tripartite Bell experiment where three parties Alice, Bob and Charlie share a system $\Lambda$ which they measure using the setting choices $A$, $B$ and $C$, producing the measurement outcomes $X$, $Y$ and $Z$ respectively. Let $\cG$ be any causal structure with only $\{A,B,C,X,Y,Z\}$ as the observed nodes where $A$, $B$ and $C$ are exogenous. Then any conditional distribution $P(XYZ|ABC)$ corresponding to the jamming correlations of~\cite{Grunhaus1996, Horodecki2019} defines a fine-tuned causal model over $\cG$, irrespective of the nature (classical, quantum or GPT) of $\Lambda$. 
\end{proposition}
\begin{proof}
Jamming allows Bob's input $B$ to be correlated jointly with $X$ and $Z$ but not individually with $X$ or $Z$. Hence jamming correlations in the tripartite Bell experiment of~\cite{Grunhaus1996, Horodecki2019} are characterised by the conditions $B\indep X$ and $B\indep Z$ while $B\not\indep XZ$. Since $B$ is exogenous (i.e., has no incoming arrows), the only way to explain the correlation between $B$ and $XZ$ is through an outgoing arrow or a directed path from $B$ to the set $XZ$ i.e., either an arrow from $B$ to $X$, or from $B$ to $Z$ or both.\footnote{If this were not the case, $B$ would be d-separated from $XZ$ and therefore cannot be correlated with it.} Since we require both independences $B\indep X$ and $B\indep Z$ to hold, at least one of these will not be a consequence of d-separation and hence the causal model must be fine-tuned in order to produce these correlations in the causal structure $\cG$.
\end{proof}

The simplest example of jamming is where $B=X\oplus Z$ and all variables are binary and uniformly distributed (the remaining variables are irrelevant here), and we will revisit this example several times in this paper. These are the same correlations as the traitorous lieutenant example.  However in the jamming case, the three variables involved are taken to be pairwise space-like separated and since $B$ is exogenous, this corresponds to a situation where $B$ superluminally influences the correlations between $X$ and $Z$. The jamming scenario involves superluminal causal influences that need not lead to observable superluminal signalling. Generalising from this idea, one can consider whether such influences can be used to create causal loops that do not lead to any signalling to the past, or even outside the space-time future. In the interest of generality and of understanding the relationships between the principles of ``no superluminal signalling'' and ``no causal loops'', one must consider fine-tuned causal influences along with cyclic causal influences, and characterise when these influences may or may not lead to signalling outside the future with respect to a space-time structure, even in the presence of latent non-classical causes. 

\section{The framework, Part 1: Causality}
\label{sec: causmod}  
This section is devoted to outlining our causal modelling framework. Section~\ref{ssec: causmodA} provides a minimal definition (Definition~\ref{definition: compatdist}) of a causal model, allowing cyclic, fine-tuned and non-classical causal influences, including when an observed distribution is compatible with a causal structure. In Section~\ref{ssec: interventions}, we describe the use of interventions within such causal models. This enables us to show that Pearl's rules of do-calculus \cite{Pearl2009} hold in the more general causal models defined here (Theorem~\ref{theorem:dorules}). Interventions give rise to \emph{affects relations} which capture the notion of signalling in a causal model (Definition~\ref{definition: affects}). Using these we classify the causal arrows in terms of whether or not they enable signalling.
For some of our results we find it useful to extend these affects relations to \emph{conditional and higher-order (HO) affects relations} (Section~\ref{ssec: HOaffects}), which capture the most general way of signalling in our framework, through joint interventions on multiple nodes.  Corollary~\ref{corollary:HOaffectsCause2} gives a main implication of conditional HO affects relations on the underlying causal structure. Section~\ref{ssec: relations} summarises the relations between the various concepts and illustrates them with several examples.

\subsection{Cyclic and fine-tuned causal models}
\label{ssec: causmodA}
Following the motivation set out in the previous sections, we wish to relax the assumptions of acyclicity and faithfulness and extend causal modelling methods to cyclic and fine-tuned causal structures with latent quantum and post-quantum causes. While quantum cyclic causal models have been previously studied~\cite{Barrett2020}, these have been analysed in the faithful case and are based on the split-node causal modelling approach of~\cite{Allen2017}. This approach is not equivalent to the standard causal modelling approach such as~\cite{Henson2014} in the cyclic case, for example the former forbids faithful 2 node cyclic causal structures~\cite{Barrett2020} but the latter does not, and the former admits a Markov factorisation (analogous to Equation~\eqref{eq: markov}) while the latter does not in general (as explained in the next paragraph). To the best of our knowledge, there is no prior framework for causally modelling cyclic and unfaithful causal structures in the presence of quantum and post-quantum latent nodes, the lack of a Markov factorisation posing a particular challenge. Here, we propose a framework for achieving this.\footnote{Note that there may be other, inequivalent ways to do the same, based on a different condition for compatibility of a distribution with a causal structure, for example.} We will define causal models in terms of minimal conditions that they must satisfy at the level of the observed ( and hence classical) nodes.

\paragraph{Observed distribution: }In classical acyclic causal models, the causal Markov condition~\eqref{eq: markov} is used for defining the compatibility of the observed distribution with the causal structure~\cite{Pearl2009}. In the non-classical case, an analogous generalised Markov condition of~\cite{Henson2014} constraining the non-classical causal mechanisms (states, transformations and measurements) provides a compatibility condition. However, in cyclic causal models, demanding such a factorisation will be too restrictive even in the classical case. For example, consider the simplest cyclic causal structure, the 2-cycle where $X\longrsquigarrow Y$ and $Y\longrsquigarrow X$, with $X$ and $Y$ observed and $X=Y$. Used na\"ively, the Markov condition would imply that $P(XY)=P(X|Y)P(Y|X)$. Since $X=Y$, the right hand side is a product of deterministic distributions, which forces $P(XY)$ to also be deterministic in order to be a valid distribution. In order to not restrict directed cycles to only consist of deterministic variables, we instead use a weaker compatibility condition in terms of d-separation between observed nodes. As we have previously noted, this is a concept that also applies to non-classical causal structures. The condition captures the intuition that certain graph separation properties in the causal structure must imply (conditional) independences in the correlations it gives rise to. Based on this, we define compatibility of the observed distribution with a cyclic causal structure as follows within our framework.

\begin{definition}[Compatibility of observed distribution with a causal structure] 
\label{definition: compatdist}
Let $\{X_1,\ldots,X_n\}$ be a set of random variables denoting the observed nodes of a directed graph $\mathcal{G}$ (which may also have unobserved nodes), and $P(X_1,\ldots,X_n)$ be a joint probability distribution over them. Then $P$ is said to be \emph{compatible with} $\mathcal{G}$ (or to \emph{satisfy the d-separation property}) if for all disjoint subsets $X$, $Y$ and $Z$ of $\{X_1,\ldots,X_n\}$,\footnote{Note that we only need to consider d-separation between observed sets of variables in this definition, however the paths being considered may involve unobserved nodes. For example, if the observed variables $X$ and $Y$ have an unobserved common cause $\Lambda$, then $X$ and $Y$ are not d-separated by the empty set since there is an unblocked path between $X$ and $Y$ through the unobserved common cause, and naturally we do not expect $X$ and $Y$ to be independent in this case.}
\begin{equation*}
    X\perp^d Y|Z \quad\Rightarrow\quad X\indep Y|Z \quad \text{ i.e., $P(XY|Z)=P(X|Z)P(Y|Z)$}.
\end{equation*}

\end{definition}

 In the previous literature, causal models are typically defined in terms of a causal structure and causal mechanisms (which are then used to derive the observed distribution). When doing so it known that Definition~\ref{definition: compatdist} is satisfied by classical as well as non-classical causal models in the acyclic case~\cite{Pearl2009, Henson2014}. The compatibility property holds in several classical cyclic causal models~\cite{Pearl2013, Forre2017}. For classical acyclic models, it is equivalent to the causal Markov condition~\eqref{eq: markov}~\cite{Pearl1995}. In Appendix~\ref{appendix: doMechanisms}, we provide an example of a quantum cyclic causal model (with causal mechanisms) where this holds. However there also exist cyclic causal models producing observed distributions that do not satisfy Definition~\ref{definition: compatdist}, we discuss this further in the Appendix as well. There, we also present further motivation for the compatibility condition of Definition~\ref{definition: compatdist} in terms of the properties of the underlying causal mechanisms (e.g., functional dependences in the classical case or completely positive maps in the quantum case) and outline possible methods for identifying when this condition might hold for non-classical cyclic causal models. Even in the classical case, several inequivalent definitions of compatibility are possible (which become equivalent in the acyclic case) and~\cite{Forre2017} presents a detailed analysis of these conditions and the relationships between them. Such an analysis for the non-classical case is beyond the scope of the present work. For the rest of this paper, we will only consider causal models that satisfy the compatibility condition~\ref{definition: compatdist}.
\par
We will work with the following minimal definition of a causal model in this paper which is in terms of the graph and observed distribution only. Further details about the causal mechanisms such as the functional relationships between classical variables, choice of quantum states or transformations, or generalised tests~\cite{Henson2014} can also be included in the full specification of the causal model. These constitute the \emph{causal mechanisms} of the model. Developing a complete and formal specification of these mechanisms and deriving the conditions for their compatibility with cyclic, fine-tuned and non-classical causal models is a tricky problem, we outline possible ideas for this in Appendix~\ref{appendix: doMechanisms} and leave the full problem for future work. The results of this paper hold without such a specification which if added would be a way to generalise them. Interestingly, we find that even with this minimal definition, we can derive several new results for a general class of causal models and also reproduce results that were originally derived for acyclic classical causal models.

\begin{definition}[Causal model]
\label{def: causalmodel}
A causal model over a set of observed random variables $\{X_1,\ldots,X_n\}$ consists of a directed graph $\mathcal{G}$ over them (possibly involving classical, quantum or GPT unobserved systems) and a joint distribution $P_{\cG}(X_1,\ldots,X_n)$ that is compatible with the graph $\cG$ according to Definition~\ref{definition: compatdist}. 
\end{definition}
Note that other definitions of causal model are used in the literature, in particular, sometimes the definition requires that $P_{X_i|\text{par}(X_i)}$ (or more generally, a possibly non-classical channel from $\text{par}(X_i)$ to $X_i$) is given for each node $X_i$, see e.g.~\cite{Pearl2009, Henson2014}.

 Definition~\ref{definition: compatdist} allows for fine-tuned distributions to be compatible with the causal structure since it only requires that d-separation implies conditional independence and not the converse.  Fine-tuned causal models may in general have an arbitrary number of additional conditional independences that are not implied by the d-separation relations in the corresponding causal graph. The following lemma shows that some additional conditional independences that are not directly implied by d-separation can be derived using d-separation and other independences (not implied by d-separation) that may be provided.

\begin{restatable}{lemma}{CI}
\label{lemma: CI} 
Let $S_1$, $S_2$ and $S_3$ be three disjoint sets of RVs such that $S_1\indep S_2|S_3$. If $S$ is a set of RVs that is d-separated from these sets in a directed graph $\mathcal{G}$ containing all the members of $S_1$, $S_2$, $S_3$ and $S$ as nodes i.e., $S\perp^d S_i$ $\forall i\in\{1,2,3\}$, then any distribution $P$ that is compatible with $\mathcal{G}$ also satisfies the following conditional independences,
$$S_1 S\indep S_2|S_3,\quad S_1\indep S_2 S|S_3 \text{ and} \quad S_1\indep S_2|S_3 S.$$
\end{restatable}
A proof can be found in Appendix~\ref{appendix: proofs1}.
Note that this lemma is trivial in the case of faithful causal models. This is because, the independence $S_1\indep S_2|S_3$ implies the d-separation $S_1\perp^d S_2|S_3$ for a faithful causal model.  Then, combined with $S\perp^d S_i$, we get the d-separations $S_1 S\perp^d S_2|S_3$, $S_1\perp^d S_2 S|S_3$ and $S_1\perp^d S_2|S_3 S$, which in turn imply the corresponding independences. This property is not so straighforward for fine-tuned causal models but nevertheless holds. Specific examples of this property for fine-tuned causal models are discussed in Appendix~\ref{appendix: examples}.

\subsection{Interventions and affects relations}
\label{ssec: interventions}

So far, we have only discussed the possible correlations that can be compatible with a causal structure. However, it is not possible to infer an underlying causal structure from correlations alone: correlations are symmetric while causal relationships are directional. For example, if two variables $X$ and $Y$ are correlated, Reichenbach's principle~\cite{Reichenbach} asserts that either $X$ must be a cause of $Y$, $Y$ must be a cause of $X$, $X$ and $Y$ share a common cause or any combination thereof. These causal explanations cannot be distinguished on the basis of observed correlations alone. However, intuitively, we can argue that if ``doing'' something only to $X$ produces a change in the distribution over $Y$, then $X$ is a cause of $Y$. This intuition is formalised in terms of interventions and do-conditionals~\cite{Pearl2009}, and we will adopt the augmented graph approach~\cite{Pearl2009} for defining these. 
\bigskip
\paragraph{Pre-intervention, augmented and post-intervention causal strutures:} Consider a causal model associated with a causal structure $\cG$ over a set $S=\{X_1,\ldots,X_n\}$ of observed nodes. External intervention on a node $X\in S$ can be described using an augmented graph $\cG_{I_{X}}$ which is obtained from the original graph $\cG$ by adding a node $I_{X}$ and an edge $I_{X}\longrsquigarrow X$ (with everything else unchanged). The intervention variable $I_{X}$ can take values in the set $\{\mathrm{idle}, \{\mathrm{do}(x)\}_{x\in X}\}$, where $I_{X}=\mathrm{idle}$ corresponds to the case where no intervention is performed (i.e., the situation described by the original causal model) and $I_{X}=\mathrm{do}(x)$ forces $X$ to take the value $x$ by cutting off its dependence on all other parents. From this, we see that whenever $I_{X}\neq \mathrm{idle}$, $X$ no longer depends on its original parents $\text{par}_{\cG}(X)$. Therefore, conditioned on $I_{X}\neq \mathrm{idle}$, it is illustrative to consider a new graph which we denote by $\cG_{\mathrm{do}(X)}$ that represents the post-intervention causal structure after a non-trivial intervention has been performed. The causal graph $\cG_{\mathrm{do}(X)}$  is obtained by cutting off all incoming arrows to $X$ except the one from $I_{X}$ in the causal graph $\cG_{I_{X}}$, with everything else unchanged. An example of the graphs $\cG$, $\cG_{I_{X}}$ and $\cG_{\mathrm{do}(X)}$ is given in Figure~\ref{fig:augCS}. The above procedure also applies to interventions on subsets of the nodes, for example, if $X$ is a subset of the observed nodes that is being intervened on, an exogenous intervention variable $I_{X_i}$ will be introduced for each element $X_i$ of $X$, along with the corresponding edge $I_{X_i}\longrsquigarrow X_i$. Then, $I_X\longrsquigarrow X$ will be used as a short hand to denote that each element of $I_X=\{I_{X_i}\}_i$ has a direct causal arrow to the corresponding $X_i$. Note that requiring each $I_{X_i}$ to be exogenous ensures that the intervention to be performed on each node is chosen independently (in principle, one could consider correlated interventions as well but we do not do so here).

\bigskip
\paragraph{Defining the post intervention causal model: }The effect of an intervention on the node $X$ setting $X=x$, i.e., performing $\mathrm{do}(x)$ is to transform the original probability distribution $P_{\mathcal{G}}(X_1,\ldots,X_n)$ into a new probability distribution $P_{\mathcal{G}_{\mathrm{do}(X)}}(X_1,\ldots,X_n,I_{X})$. 
These distributions are compatible with the original (i.e., pre-intervention) and the post-intervention graphs, $\cG$ and $\cG_{\mathrm{do}(X)}$ respectively and the following defining rules tell us some of the relationships between these distributions. Here the distribution $P_{\mathcal{G}_{I_X}}(X_1,\ldots,X_n,I_{X})$ compatible with the augmented graph $\cG_{I_X}$ mediates the relationships between the pre and post intervention scenarios. Note that the set of intervention variables $I_X$ is additionally introduced in going from $\cG$ to $\cG_{I_X}$ or $\cG_{\mathrm{do}(X)}$. In the corresponding causal models, the distribution over $I_X$ can be arbitrary and all of the following definitions and results hold for any choice of $P_{I_X}$. Then, for any two disjoint subsets $X$ and $Y$ of the observed nodes, the following defining equations hold.
\begin{subequations}
\begin{align}
\label{eq: intervention1}
P_{\cG_{I_X}}(Y|I_X=\mathrm{idle})&=P_{\cG}(Y) \\
\label{eq: intervention2}
P_{\cG_{I_X}}(Y|I_X=\mathrm{do}(x))&=P_{\cG_{\mathrm{do}(X)}}(Y|I_X=\mathrm{do}(x))=P_{\cG_{\mathrm{do}(X)}}(Y|X=x) \quad \forall  x\\
\label{eq: intervention3}
P_{\cG_{I_X}}(Y|I_X=\mathrm{do}(x), X=x)&=P_{\cG_{I_X}}(Y|I_X=\mathrm{do}(x)) \quad \forall  x\\
\label{eq: intervention4}
P_{\cG_{I_X}}(I_X=\mathrm{do}(x), X=x')&=0 \quad \forall x, x' \text{ such that } x\neq x'
\end{align}
\end{subequations}

Intuitively, the first equation tells us that when all the intervention variables are ``idle'', this corresponds to the original causal model, as no intervention is performed. The remaining three equations capture the fact that when a non-trivial intervention is performed, each intervention variable $I_{X_i}\in I_X$ is perfectly correlated with the corresponding intervened variable $X_i\in X$. 
The conditional probability distribution $P_{\cG_{\mathrm{do}(X)}}(Y|X=x)$ of Equation~\eqref{eq: intervention2} is often denoted simply as $P(Y|\mathrm{do}(x))$ and commonly referred to as the \emph{do-conditional}. Note that $P(y|\mathrm{do}(x)):=P_{\cG_{\mathrm{do}(X)}}(y|x)\neq P(y|x):=P_{\cG}(y|x)$ in general.  At first sight, it might appear that these defining equations do not tell us how the pre and post intervention distributions $P_{\cG}$ and $P_{\cG_{\mathrm{do}(X)}}$ are related since $P_{\cG}$ is related to $P_{\cG_{I_X}}$ only when $I_X=$ idle (Equation~\eqref{eq: intervention1}) and $P_{\cG_{I_X}}$ is related to $P_{\cG_{\mathrm{do}(X)}}$ only when $I_X\neq$ idle. However, as we will see in subsequent sections, these defining rules along with compatibility condition of Definition~\ref{definition: compatdist} allow us to derive further useful rules that explicitly connect the pre and post intervention distributions. The intuition for this is that the augmented and post-intervention graphs are constructed from the pre-intervention graph and certain d-separations in the pre-intervention graph imply corresponding d-separations in the augmented and post-intervention graphs, and therefore certain independences in the associated distributions.

\begin{figure}[t!]
    \centering
		\subfloat[]{\includegraphics[]{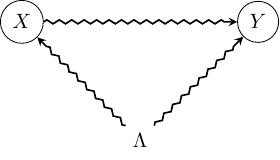}
	}\quad	\subfloat[]{\includegraphics[]{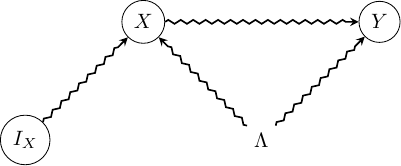}
	}\quad
	\subfloat[]{\includegraphics[]{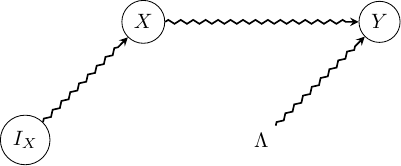}
	}
    \caption[Pre-intervention, augmented and post-intervention causal structures]{\textbf{Pre-intervention, augmented and post-intervention causal structures: } Taking the original, pre-intervention causal structure, $\cG$, to be that of (a), parts (b) and (c) of this figure illustrate the augmented causal structure, $\cG_{I_X}$, and post-intervention causal structure, $\cG_{\mathrm{do}(X)}$, for intervention on $X$. In $\cG_{I_X}$, the variable $I_X$ can take values in the set $\{\mathrm{idle},\{\mathrm{do}(x)\}_{x\in X}\}$ while in $\cG_{\mathrm{do}(X)}$, it can only take the values $\{\mathrm{do}(x)\}_{x\in X}$ corresponding to an active intervention. Conditioned on $I_X=\mathrm{idle}$, we effectively obtain the original causal model (a) which corresponds to no intervention being performed, as specified by Equation~\eqref{eq: intervention1}.}
    \label{fig:augCS}
\end{figure}

\bigskip
\paragraph{The physical picture: }At the level of the causal mechanisms (if these are also given), the causal mechanisms of $\cG_{\mathrm{do}(X)}$ can be obtained from those of $\cG$ simply by updating the causal mechanisms for each node $X_i$ in $X$ as $X_i=x_i$ iff $I_{X_i}=\mathrm{do}(x_i)$ (while leaving the causal mechanisms for all other nodes unchanged) i.e., $P_{\cG_{\mathrm{do}(X)}}(X)$ is fully determined by the original causal model, the causal mechanisms and $P_{\cG_{\mathrm{do}(X)}}(I_X)$ which can be chosen arbitrarily for the exogenous set $I_X$. Physically, the post-intervention distribution (or the do-conditional) corresponds to additional empirical data that are collected in an experiment, which can, in general, be different from the experiment generating the original, pre-intervention data. For example, when the original experiment involves passive observation of correlations between the smoking tendencies and presence of cancer in a group of individuals, an intervention model may involve forcing certain individuals to take up smoking and then studying their chances of developing cancer. In repeated trials, the proportion of individuals who are passively observed and those that are actively intervened upon may be chosen as desired. The latter type of experiments may not necessarily be ethical but are nevertheless a physical possibility. In certain cases, it may be possible to fully deduce the post-intervention statistics counterfactually from the pre-intervention data (passive observation) alone, and the latter experiment (active intervention) need not be actually performed, sparing us some ethical dilemmas. For example, in a causal structure where all nodes are observed, this is always possible~\cite{Pearl2009}. However, even in simple classical causal structures with unobserved nodes, the post-intervention distribution cannot be completely determined using the pre-intervention distribution alone~\cite{Pearl2009}.

\bigskip
\paragraph{Further relationships between the pre and post intervention causal models: }As explained above, determining the post-intervention distribution from the pre-intervention distribution alone is not possible in the general settings considered here. However, the compatibility condition of Definition~\ref{definition: compatdist} along with the defining rules of Equations~\eqref{eq: intervention1}-\eqref{eq: intervention4} allows us to derive further useful relationships between these distributions, in particular the three rules of Pearl's do-calculus~\cite{Pearl1995, Pearl2009}. These rules have been originally derived in faithful classical causal models satisfying the causal Markov property~\eqref{eq: markov} which does not hold in the general scenarios considered here. Here, we extend these rules to a large class of unfaithful and cyclic non-classical causal models, by noting that the derivation of these rules do not require the Markov property but only the weaker d-separation condition of Definition~\ref{definition: compatdist} along with the defining rules~\eqref{eq: intervention1}-\eqref{eq: intervention4}. This is captured in the following theorem and we present a proof in Appendix~\ref{appendix: proofs1} for completeness (this is similar to the original proof of~\cite{Pearl1995} but more explicit). In the following, $\cG_{\overline{X}}$ denotes the graph obtained by deleting all incoming edges to $X$ and $\cG_{\underline{X}}$ denotes the graph obtained by deleting all outgoing edges from $X$ in a graph $\cG$, where $X$ is some subset of the observed nodes. 

\begin{restatable}{theorem}{DoRules}
\label{theorem:dorules}
Given a causal model over a set $S$ of observed nodes, an associated causal graph $\cG$ and a distribution $P_S$ compatible with $\cG$ according to Definition~\ref{definition: compatdist}, the following 3 rules of do-calculus~\cite{Pearl2009} hold for interventions on this causal model.
\begin{itemize}
    \item \textbf{Rule 1: Ignoring observations} 
    \begin{equation}
    \label{eq: rule1}
        P_{\cG_{\mathrm{do}(X)}}(y|x,z,w)=P_{\cG_{\mathrm{do}(X)}}(y|x,w) \qquad \text{if } (Y\perp^d Z|XW)_{\cG_{\overline{X}}}
    \end{equation}
     \item \textbf{Rule 2: Action/observation exchange} 
    \begin{equation}
     \label{eq: rule2}
        P_{\cG_{\mathrm{do}(X Z)}}(y|x,z,w)=P_{\cG_{\mathrm{do}(X)}}(y|x,z,w) \qquad \text{if } (Y\perp^d Z|XW)_{\cG_{\overline{X}\underline{Z}}}
    \end{equation}
    \item \textbf{Rule 3: Ignoring actions/interventions} 
     \begin{equation}
      \label{eq: rule3}
        P_{\cG_{\mathrm{do}(X Z)}}(y|x,z,w)=P_{\cG_{\mathrm{do}(X)}}(y|x,w) \qquad \text{if } (Y\perp^d Z|XW)_{\cG_{\overline{XZ(W)}}},
    \end{equation}
\end{itemize}
where $X$, $Y$, $Z$ and $W$ are disjoint subsets of the observed nodes, $Z(W)$ denotes the set of nodes in $Z$ that are not ancestors of $W$, and the above hold for all values $x$, $y$, $z$ and $w$ of $X$, $Y$, $Z$ and $W$. 
\end{restatable}

While the observed distribution in the post-intervention causal model may not be completely specified by the pre-intervention observed distribution alone, considering the underlying causal mechanisms e.g., the states, transformations and measurements involved in the original causal model should allow for the complete specification of the post-intervention distribution. To the best of our knowledge, this problem has not been studied in non-classical and cyclic causal models, we discuss this point in further detail in Appendix~\ref{appendix: doMechanisms}, providing examples of non-classical cyclic causal models where the post-intervention distribution can be calculated from the causal mechanisms. The full solution to this problem will not be relevant to the results of the main paper. Using these concepts, we now define the \emph{affects relation} that is central to the results of this paper.

\begin{definition}[Affects relation]
\label{definition: affects}
Consider a causal model associated with a causal graph $\cG$ over a set $S$ of observed nodes and an observed distribution $P$ and let $X$ and $Y$ be disjoint subsets of $S$. If there exists a value $x$ of $X$ such that
\begin{equation*}
   P_{\cG_{\mathrm{do}(X)}}(Y|X=x)\neq P_{\cG}(Y),
\end{equation*}
then we say that $X$ \emph{affects} $Y$.
\end{definition}
With this definition, we are ready to state two useful corollaries of Theorem~\ref{theorem:dorules}.

\begin{corollary}
\label{corollary:exogenous}
If $X$ is a subset of observed exogenous nodes of a causal graph $\cG$, then for any subset $Y$ of nodes disjoint to $X$ the do-conditional and the regular conditional with respect to $X$ coincide i.e.,
$$P_{\cG_{\mathrm{do}(X)}}(Y|X)=P_{\cG}(Y|X).$$
In other words, for any subset $X$ of the observed exogenous nodes, correlation between $X$ and a disjoint set of observed nodes $Y$ in $\cG$ guarantees that $X$ affects $Y$. 
\end{corollary}
\begin{proof}
Since $X$ consists only of exogenous nodes, it can only be d-connected to other nodes through outgoing arrows. Then in the graph $\cG_{\underline{X}}$ (where all outgoing arrows from $X$ are cut off), $X$ becomes d-separated from all other nodes. This d-separation, $(Y\perp^d X)_{\cG_{\underline{X}}}$ implies, by Rule 2 of Theorem~\ref{theorem:dorules} that $P_{\cG_{\mathrm{do}(X)}}(Y|X=x)=P_{\cG}(Y|X=x) \quad \forall x$. Further if $X$ and $Y$ are correlated in $\cG$, i.e., $\exists x, y$ such that $P_{\cG}(y|x)\neq P_{\cG}(y)$, the equation previously established along with Definition~\ref{definition: affects} implies that $X$ affects $Y$.
\end{proof}

\begin{corollary}
\label{corollary:dsep-affects}
If $X$ and $Y$ are two disjoint subsets of the observed nodes such that $(X\perp^d Y)_{\cG_{\mathrm{do}(X)}}$, then $X$ \emph{does not affect} $Y$ and $P_{\cG_{\mathrm{do}(X)}}(Y)=P_{\cG}(Y)$.
\end{corollary}
\begin{proof}
The d-separation $(X\perp^d Y)_{\cG_{\mathrm{do}(X)}}$ implies the d-separation $(X\perp^d Y)_{\cG_{\overline{X}}}$ since $\cG_{\mathrm{do}(X)}$ and $\cG_{\overline{X}}$ only differ by the inclusion of the intervention nodes $I_{X_i}$ and the corresponding edges $I_{X_i}\longrightarrow X_i$ for each $X_i\in X$. Then by Rule 3 of Theorem~\ref{theorem:dorules} we have 
$$P_{\cG_{\mathrm{do}(X)}}(Y|X)=P_{\cG}(Y)$$
which by Definition~\ref{definition: affects} implies that $X$ does not affect $Y$. Further, the d-separation implies the conditional independence $(X\indep Y)_{\cG_{\mathrm{do}(X)}}$ i.e., $$P_{\cG_{\mathrm{do}(X)}}(Y|X)=P_{\cG_{\mathrm{do}(X)}}(Y)$$ which along with the result that $X$ does not affect $Y$ yields 
\begin{equation*}
    P_{\cG_{\mathrm{do}(X)}}(Y)=P_{\cG}(Y).\qedhere
\end{equation*}
\end{proof}

Note that $X$ \emph{affects} $Y$ implies that there must be a directed path from $X$ to $Y$ in $\mathcal{G}$ (which is equivalent to $X$ being a cause of $Y$, cf.\ Definition~\ref{def: cause}). This follows from the contrapositive statement of Corollary~\ref{corollary:dsep-affects}--- $X$ affects $Y$ implies that $X$ and $Y$ are not d-separated in $\cG_{\mathrm{do}(X)}$ and since this graph has no incoming arrows to $X$ (except those from the intervention nodes in $I_X$), the only way for $X$ and $Y$ to be d-connected in $\cG_{\mathrm{do}(X)}$ is through a directed path from $X$ to $Y$. 
However, the converse is not true. A directed path from $X$ to $Y$ in $\cG$ does not imply that $X$ affects $Y$ in the presence of fine-tuning (as illustrated in the examples of Appendix~\ref{appendix: examples}), even though it does imply d-connection between $X$ and $Y$ in $\cG_{\mathrm{do}(X)}$ by construction of this graph. This motivates the following classification of the causal arrows $\longrsquigarrow$ between observed nodes. The arrows  $\longrsquigarrow$ emanating from or pointing to an unobserved node cannot be operationally probed and hence need not be classified. 

\begin{definition}[Solid and dashed arrows]
\label{definition: solidasharrows}
Given a causal graph $\mathcal{G}$, if two observed nodes $X$ and $Y$ in $\mathcal{G}$ sharing a directed edge $X\longrsquigarrow Y$ are such that $X$ affects $Y$, then the causal arrow $\longrsquigarrow$ between those nodes is called a \emph{solid arrow}, denoted $X\longrightarrow Y$. Further, all arrows $\longrsquigarrow$ between observed nodes in $\mathcal{G}$ that are \emph{not} solid arrows are called \emph{dashed arrows}, denoted $\xdashrightarrow{}$. In other words, $X\xdashrightarrow{} Y$ for any two RVs $X$ and $Y$ in $\cG$ implies that the $X$ does not affect $Y$.
\end{definition}

\begin{remark}[Exogenous nodes]
Note that if $X$ is an exogenous node that is a direct cause of another node $Y$ in a causal graph $\cG$  i.e., $X\longrsquigarrow Y$, and $X$ and $Y$ are correlated in the corresponding causal model, then by Corollary~\ref{corollary:exogenous} and Definition~\ref{definition: solidasharrows} this would imply that the arrow from $X$ to $Y$ must be a solid one. Applying this to the graphs $\cG_{I_X}$ and $\cG_{\mathrm{do}(X)}$, where $I_X$ is exogenous and correlated with $X$ by construction (Equations~\eqref{eq: intervention1}-\eqref{eq: intervention4}), we can conclude that  the arrow from every intervention variable to the corresponding intervened variable must be a solid arrow, i.e., $I_X\longrightarrow X$.
\end{remark}

A  noteworthy implication that follows from the defining rules is encapsulated in the following lemma.
 \begin{lemma}
 \label{lemma: correl-affects}
 Given a causal graph $\cG$ and two disjoint subsets $X$ and $Y$ of observed nodes therein,
 $$(X\nindep Y)_{\cG_{\mathrm{do}(X)}}\Rightarrow X \text{ affects } Y.$$
 \end{lemma}
 \begin{proof}
 Suppose that $X$ does not affect $Y$. By Definition~\ref{definition: affects}, this implies that $P_{\cG_{\mathrm{do}(X)}}(y|x)=P_{\cG}(y)$ $\forall x,y$. Further suppose also that $(X\nindep Y)_{\cG_{\mathrm{do}(X)}}$. This means that there exist two distinct values $x$ and $x'$ of $X$ and some value $y$ of $Y$ such that $P_{\cG_{\mathrm{do}(X)}}(y|x)\neq P_{\cG_{\mathrm{do}(X)}}(y|x')$, which contradicts $P_{\cG_{\mathrm{do}(X)}}(y|x)=P_{\cG}(y)$ $\forall x,y$. Therefore $(X\nindep Y)_{\cG_{\mathrm{do}(X)}}$ must imply $X \text{ affects } Y$.
 \end{proof}
 
The converse of the above lemma does not hold as illustrated by Example~\ref{example: eg3}.
 Further, we note that the affects relation is not transitive in fine-tuned causal models, as illustrated by the following example.
 \begin{example}
 \label{example:affects_transitive}
 Consider the causal structure of Figure~\ref{fig:affects_transitive} where all RVs are binary and related by $X=\Lambda$, $Y=W=X\oplus\Lambda$, $Z=Y\oplus W$ with $\Lambda$ uniformly distributed. Here, both $P_{\mathcal{G}}(Y)$ and $P_{\mathcal{G}}(Z)$ are deterministic distributions. In the graph $\mathcal{G}_{\mathrm{do}(X)}$ obtained by intervening on $X$, we have $Y=W=X\oplus\Lambda$, $Z=Y\oplus W$ and $\Lambda$ uniform. Here, since $X$ is not always equal to $\Lambda$, $P_{\mathcal{G}_{\mathrm{do}(X)}}(Y|X)$ is no longer deterministic and we have $X$ affects $Y$, but $P_{\mathcal{G}_{\mathrm{do}(X)}}(Z|X)$ is still the same deterministic distribution irrespective of the value of $X$ since $Y=W$ which implies that $X$ does not affect $Z$. However, in the graph $\mathcal{G}_{\mathrm{do}(Y)}$, we no longer have $Y=W$ and $P_{\mathcal{G}_{\mathrm{do}(Y)}}(Z|Y)$ is not deterministic, which gives $Y$ affects $Z$. Therefore affects relations are in general non-transitive in fine-tuned causal models.
 \end{example}

\begin{figure}
    \centering
\includegraphics[]{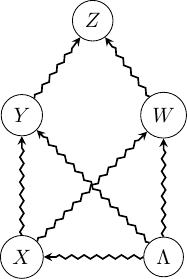}
    \caption{Causal structure of Example~\ref{example:affects_transitive}}
    \label{fig:affects_transitive}
\end{figure}

\subsection{Conditional and higher-order affects relations}
\label{ssec: HOaffects}

The affects relation defined in Definition~\ref{definition: affects} allows us to consider joint interventions on a subset of the observed nodes $S$. However certain affects relations where a subset $X\subset S$ that is not a single RV, affects another subset $Y$, may be ``trivial'' in the sense that they convey the same information as an affects relation $s_X$ affects $Y$, where $s_X$ is a proper subset of $X$, i.e., they can be ``reduced'' to the latter affects relation. On the other hand, in unfaithful causal models, certain affects relations of the same form can be ``non-trivial'' in the sense that the information that they convey is not the same as any affects relation from a proper subset of $X$ to $Y$. To capture this distinction, we introduce higher-order affects relations where we consider whether a set $X$ of RVs affects another disjoint set $Y$ conditioned on an active intervention performed on a third, mutually disjoint subset $Z$ of the RVs. Intuitively these relations are useful because additional interventional information can help us better detect fine-tuned causal influences. More generally, we can also condition on non-interventional information, which leads to the concept of conditional higher-order affects relations. As we will see later in the paper when we bring space-time into the picture, these higher-order affects relations have operational meaning in terms of signalling using joint interventions on space-time random variables, and the conditional higher-order affects relations capture the most general way that agents may signal to each other in our framework. Before we formalise these concepts, some examples would be illustrative. 

\begin{example}
\label{example: HOaffects1}
Consider a causal model where the only nodes are the observed binary variables $X$, $Y$ and $Z$, and the causal graph (Figure~\ref{fig: HOaffects1}) is simply $Z\longrightarrow Y$ and $X$ has no incoming or outgoing arrows. By Definition~\ref{definition: solidasharrows} of the solid arrow, $Z$ affects $Y$ and by Corollary~\ref{corollary:dsep-affects}, $X$ does not affect $Y$. We also have $XZ$ affects $Y$. This is because $P_{\cG_{\mathrm{do}(X Z)}}(Y|XZ)=P_{\cG}(Y|XZ)$ and $P_{\cG_{\mathrm{do}(Z)}}(Y|Z)=P_{\cG}(Y|Z)$ (by exogeneity of $X$ and $Z$), and using the d-separation condition~\ref{definition: compatdist}) we have $P_{\cG}(Y|XZ)=P_{\cG}(Y|Z)$. Then $Z$ affects $Y$ implies $P_{\cG}(Y|XZ)=P_{\cG}(Y|Z)\neq P_{\cG}(Y)$ i.e., $XZ$ affects $Y$. In this example, the node $X$ is entirely superficial as it neither causes nor is a cause of anything else and is therefore completely independent and the affects relation $XZ$ affects $Y$ follows ``trivially'' from $Z$ affects $Y$. 
\end{example}

\begin{example}
\label{example: HOaffects2}
Consider another causal model over the same nodes as the previous example, where the causal graph is a collider from $X$ and $Z$ to $Y$ i.e., $X\longrsquigarrow Y \longlsquigarrow Z$. Furthermore, suppose that $Z$ is uniformly distributed, $X$ is not uniformly distributed and $Y=X\oplus Z$ (where $\oplus$ denotes modulo-2 addition). One can then easily check that the same affects relations as the previous example hold i.e., $Z$ affects $Y$, $X$ does not affect $Y$ and $XZ$ affects $Y$, which allows us to classify the causal arrows as in Figure~\ref{fig: HOaffects2}. In this case, $Z$ gives partial information about $Y$ since $X$ is non-uniform, however $X$ and $Z$ taken together give full information about $Y$. This is in contrast to the previous example where $Z$ as well as $XZ$ gave the same information about $Y$. More explicitly, the distinguishing condition here is whether or not $P_{\cG_{\mathrm{do}(X Z)}}(Y|XZ)=P_{\cG_{\mathrm{do}(Z)}}(Y|Z)$; in the previous example this holds, while in the current one it does not. 
\end{example}

In general $X$, $Y$ and $Z$ from the above example may be pairwise disjoint subsets of the observed nodes, and we may condition not only on the set $Z$ (which has been intervened upon), but also on an additional disjoint set of nodes $W$, upon which an intervention has not been performed. We then have the following definition.

\begin{definition}[Conditional higher-order affects relation]
\label{definition:HOaffects}
Consider a causal model associated with a causal graph $\cG$ over a set $S$ of observed nodes and an observed distribution $P$. For four pairwise disjoint subsets $X$, $Y$, $Z$ and $W$ of $S$, we say that $X$ affects $Y$ given $\{\mathrm{do}(Z),W\}$ if there exists values $x$ of $X$, $z$ of $Z$ and $w$ of $W$ such that
\begin{equation}
\label{eq: HOaffects}
   P_{\cG_{\mathrm{do}(X Z)}}(Y|X=x,Z=z, W=w)\neq P_{\cG_{\mathrm{do}(Z)}}(Y|Z=z, W=w).
\end{equation}
An affects relation $X$ affects $Y$ given $\{\mathrm{do}(Z),W\}$ is a \emph{conditional affects relation} if $W\neq \emptyset$ and an \emph{unconditional affects relation} otherwise. When $Z\neq \emptyset$, it is a \emph{higher-order affects relation}, and a \emph{zeroth-order affects relation} otherwise. Definition~\ref{definition: affects} then refers to unconditional zeroth-order affects relations. In general, all of these will be simply called affects relations, unless they need to be explicitly distinguished.
\end{definition}
\begin{figure}[t]
    \centering
  \subfloat[\label{fig: HOaffects1}]{ \includegraphics[]{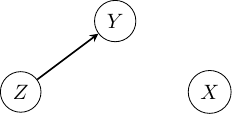}}\qquad\qquad\qquad  \subfloat[\label{fig: HOaffects2}]{ \includegraphics[]{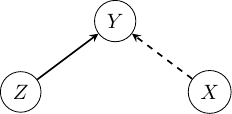}}\qquad\qquad\qquad\subfloat[\label{fig: HOaffects3}]{ \includegraphics[]{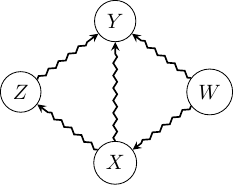}}
    \caption{Causal structures for Examples~\ref{example: HOaffects1}, \ref{example: HOaffects2} and \ref{example: HOaffects3} respectively.}
\end{figure}
The next lemma (proven in Appendix~\ref{appendix: proofs2}) establishes the implication of such affects relations for the underlying causal structure.
\begin{restatable}{lemma}{HOaffectsCause}
  \label{lemma:HOaffectsCause}
  For a causal model over a set $S$ of RVs where $X$, $Y$, $Z$ and $W$ are any pairwise disjoint subsets of $S$,
  \begin{enumerate}
      \item $X$ affects $Y$ given do$(Z)$
 $\Rightarrow$ $X$ is a cause of $Y$ (cf.\ Definition~\ref{def: cause}).
 \item $X$ affects $Y$ given $\{\mathrm{do}(Z),W\}$
 $\Rightarrow$ $X$ is a cause of $Y$ or $X$ is a cause of $W$.
 \end{enumerate}
\end{restatable}
It is possible for $X$ not to be a cause of $Y$ and yet satisfy $X$ affects $Y$ given $\{\mathrm{do}(Z),W\}$. A simple example is a 3 node collider causal structure $X\longrightarrow W\longleftarrow Y$ with $W=X.Y$, it is easy to check that $X$ affects $Y$ given $W$ even though $X$ and $Y$ are d-separated. This captures the well known fact that conditioning on a collider can introduce correlations between independent, exogenous variables. Note however that $X$ is a cause of $W$ as implied by the above lemma.

  The following lemmas provide useful connections between conditional higher-order and conditional zeroth-order affects relations, their proofs can be found in Appendix~\ref{appendix: proofs2}. We will often abbreviate higher-order to HO in the following.
  \begin{restatable}{lemma}{HOaffectsA}
  \label{lemma:HOaffects1}
  For a causal model over a set $S$ of RVs where $X$, $Y$, $Z$ and $W$ are pairwise disjoint subsets of $S$,
  $$X \text{ affects } Y \text{ given } \{\mathrm{do}(Z), W\} \quad\Rightarrow\quad Z \text{ affects } Y \text{ given } W \quad \text{or}\quad XZ \text{ affects } Y \text{ given } W.$$
  \end{restatable}

  \begin{restatable}{lemma}{HOaffectsB}
  \label{lemma:HOaffects2}
 For a causal model over a set $S$ of RVs where $X$, $Y$, $Z$ and $W$ are pairwise disjoint subsets of $S$ and $X$ consists only of exogenous nodes, $$X \text{ affects } Y \text{ given }\{\mathrm{do}(Z),W\} \quad\Rightarrow\quad XZ \text{ affects } Y \text{ given } W.$$
 \end{restatable}

The converse of Lemma~\ref{lemma:HOaffects2} is not true, we can have $X$ does not affect $Y$ given $\{\mathrm{do}(Z),W\}$ even when $XZ$ affects $Y$ given $W$, as we have seen for $W=\emptyset$ in Example~\ref{example: HOaffects1} where $X$ was superficial to the causal model, and the affects relation $XZ$ affects $Y$ trivially followed from the affects relation $Z$ affects $Y$. Note also that the implication of the above lemma does not hold in general when $X$ is not exogenous. This is because in fine-tuned causal models (rather counter-intuitively), $Z$ affects $Y$ does not imply that any set of RVs containing $Z$ also affects $Y$, which was a step required in the above proof. The following example illustrates this.

\begin{example}
\label{example: HOaffects3}
Consider the causal structure of Figure~\ref{fig: HOaffects3}. Suppose that the exogenous $W$ is uniformly distributed and the variables are related as $Y=X\oplus Z\oplus W$, $Z=X$, $X=W$. This gives $Y=X=Z=W$ and hence $P_{\cG}(Y)=P_{\cG}(W)$ is uniform. In the graph $\cG_{\mathrm{do}(Z)}$, we have $Y=X\oplus Z\oplus W$, and $X=W$ which gives $Y=Z$ and hence $P_{\cG_{\mathrm{do}(Z)}}(Y|Z)$ is deterministic. This gives $Z$ affects $Y$. In the graph $\cG_{\mathrm{do}(X Z)}$, we only have the relation $Y=X\oplus Z\oplus W$ which implies that $P_{\cG_{\mathrm{do}(X Z)}}(Y|XZ)$ is uniform and hence that $XZ$ does not affect $Y$. Note that we also have $X$ affects $Y$ given do$(Z)$.
\end{example}

Definition~\ref{definition:HOaffects}  does not yet fully capture the notion of ``reducibility'' or ``triviality'' of certain affects relations. consider Example~\ref{example: HOaffects2} again and add a superficial observed node $V$ with no incoming or outgoing arrows. Then we have both the higher-order affects relations $X$ affects $Y$ given do$(Z)$ and $XV$ affects $Y$ given do$(Z)$. However, the addition of $V$ adds no information to the original affects relation since $P_{\cG_{\mathrm{do}(X Z V)}}(Y|XZV)=P_{\cG_{\mathrm{do}(X Z)}}(Y|XZ)$ (i.e., $V$ does not affect $Y$ given do$(X Z)$). In other words, the affects relation $XV$ affects $Y$ given do$(Z)$ is reducible to the affects relation $X$ affects $Y$ given do$(Z)$. Based on this idea, we propose the following criterion for distinguishing between reducible and irreducible affects relations.

\begin{definition}[Reducible and irreducible affects relations]
\label{definition: ReduceAffects}
For a causal model defined over a set $S$ of observed nodes, the affects relation $X$ affects $Y$ given $\{\mathrm{do}(Z),W\}$ between pairwise disjoint subsets $X$, $Y$, $Z$ and $W$ of $S$ is said to be \emph{reducible} if there exists a proper subset $s_X$ of $X$ such that $s_X$ does not affect $Y$ given $\{\mathrm{do}(Z \tilde{s}_X),W\}$, where $\tilde{s}_X:=X\backslash s_X$.
Conversely, if for all proper subsets $s_X$ of $X$, $s_X$ affects $Y$ given $\{\mathrm{do}(Z \tilde{s}_X),W\}$, the affects relation $X$ affects $Y$ given $\{\mathrm{do}(Z),W\}$ is said to be \emph{irreducible}.
\end{definition}
Then we have the following lemmas, which make clear why the above definition captures a notion of ``reduction'' of the affects relation. Proofs of these lemmas can be found in Appendix~\ref{appendix: proofs2}.
\begin{restatable}{lemma}{Reduce}
\label{lemma:reduce}
For every reducible affects relation $X$ affects $Y$ given $\{\mathrm{do}(Z),W\}$, there exists a proper subset $\tilde{s}_X$ of $X$ such that $\tilde{s}_X$ affects $Y$ given $\{\mathrm{do}(Z),W\}$. 
\end{restatable}

\begin{restatable}{lemma}{ReduceB}
\label{lemma:reduce2}
For a causal model over a set $S$ of RVs of which $X_1$, $X_2$, $Y$, $Z$ and $W$ are pairwise disjoint subsets,
\begin{center}
   $X_1$ affects $Y$ given $\{\mathrm{do}(Z),W\}$ and $X_2$ does not affect $Y$ given $\{\mathrm{do}(Z X_1),W\}$
   \begin{center}
       $\Downarrow$ 
   \end{center} $X_1 X_2$ affects $Y$ given $\{\mathrm{do}(Z),W\}$.
\end{center}
\end{restatable}

Definition~\ref{definition: ReduceAffects} classifies the relation $XZ$ affects $Y$ as reducible in Example~\ref{example: HOaffects1} (Fig.~\ref{fig: HOaffects1}), and irreducible in Example~\ref{example: HOaffects2} (Fig.~\ref{fig: HOaffects2}). Note that checking for the (ir)reducibility of an affects relation involves considering an affects relation of a greater order than the original one, where the order of $X$ affects $Y$ given $\{\mathrm{do}(Z),W\}$ is measured by the cardinality $|Z|$ of $Z$.

The following lemma (proven in Appendix~\ref{appendix: proofs2}) relates conditional affects relations to unconditional affects relations such that the irreducibility of the former implies the irreducibility of the latter. As we will later see, this will allow us to restrict to unconditional affects relations without loss of generality when considering their space-time embeddings (cf.\ Remark~\ref{remark: HOaff_complete}).

\begin{restatable}{lemma}{HOaffComplete}
\label{lemma: HOaff_complete}
 For a causal model over a set $S$ of RVs where $X$, $Y$, $Z$ and $W$ are pairwise disjoint subsets of $S$,
 \begin{enumerate}
     \item  $X$ affects $Y$ given $\{\mathrm{do}(Z),W\}$ $\Rightarrow$ $X$ affects $YW$ given do$(Z)$.
      \item  $X$ affects $Y$ given $\{\mathrm{do}(Z),W\}$ is irreducible $\Rightarrow$ $X$ affects $YW$ given do$(Z)$ is irreducible.
     \item $X$ affects $YW$ given do$(Z)$ $\Leftrightarrow$ $X$ affects $Y$ given $\{\mathrm{do}(Z),W\}$  \emph{or} $X$ affects $W$ given do$(Z)$.
 \end{enumerate}
\end{restatable}
The converse does not hold for the first two statements of this lemma, as illustrated by the following counter-examples. For part 1, consider again Example~\ref{example: HOaffects2} with the superficial observed node $V$ having no incoming or outgoing arrows. Here we have $Z$ affects $Y$ and $Z$ affects $VY$ and yet $Z$ does not affect $V$ given $Y$ ($Z$, $V$ and $Y$ play the role of $X$, $Y$ and $Z$ in the above lemma with $W=\emptyset$). For part 2, consider the causal structure $X_1\xdashrightarrow{} W\xdashleftarrow{} X_2\longrightarrow Y$ with all variables binary, $W=X_1\oplus X_2$, $Y=X_2$, $X_1$ and $X_2$ uniformly distributed. Taking $X=X_1 X_2$, it is easy to verify that we have $X$ affects $Y$ given $W$, $X$ affects $YW$ and it is irreducible, while $X$ affects $Y$ given $W$ is reducible to $X_2$ affects $Y$ given $W$.

Using this, we obtain a stronger version of Lemma~\ref{lemma:HOaffectsCause} as a corollary of Lemmas~\ref{lemma:HOaffectsCause} and~\ref{lemma: HOaff_complete} (see Appendix~\ref{appendix: proofs2} for a proof).
\begin{restatable}{corollary}{HOcauseB}
\label{corollary:HOaffectsCause2}
For a causal model over a set $S$ of RVs where $X$, $Y$, $Z$ and $W$ are any pairwise disjoint subsets of $S$,
  \begin{enumerate}
      \item $X$ affects $Y$ given do$(Z)$ is irreducible
 $\Rightarrow$ for each element $e_X\in X$ there exists an element $e_Y\in Y$ such that $e_X$ is a cause of $e_Y$.
 \item $X$ affects $Y$ given $\{\mathrm{do}(Z),W\}$ is irreducible
 $\Rightarrow$ for each element $e_X\in X$ there exists an element $e_{YW}\in Y W$ such that $e_X$ is a cause of $e_{YW}$.
 \end{enumerate}
\end{restatable}

\begin{remark}
\label{remark: Pearl_HO}
Note that in the language of conditional HO affects relations, Pearl's 3rd rule of do-calculus (Theorem~\ref{theorem:dorules}) can be written in the equivalent form
\begin{center}
    $(Y\perp^d Z|XW)_{\cG_{\mathrm{do}(X Z(W))}}$ $\quad\Rightarrow\quad$ $X$ does not affect $Y$ given $\{\mathrm{do}(Z),W\}$,
\end{center}
where $Z(W)$ is the set of nodes in $Z$ that are not ancestors of $W$.
\end{remark}

\begin{remark}
\label{remark:HOaffects}
We have seen in Definition~\ref{definition: solidasharrows} that a dashed arrow from $X$ to $Y$ corresponds to causation in the absence of the corresponding zeroth-order affects relation $X$ affects $Y$. A natural question to ask is whether all dashed arrows in a causal model can be detected using higher-order affects relations. If we consider causal models with no latent nodes, then this is the case. Such a model is entirely classical and the causal mechanisms consist of functional equations i.e., for each node $Y$, a function $f_Y$ taking as input the parent variables par$(Y)$ and an independent, exogenous error variable $E_Y$ that completely determines $Y$ as $Y=f_Y(\text{par}(Y),E_Y)$. The meaning of saying that $X$ is a parent of $Y$ is that $f_Y$ has a nontrivial dependence on the input $X$, i.e., there exists a fixed value of all other inputs of $f_Y$ such that changing the value of $X$ produces a change in the function value. This is precisely captured by the higher-order affects relation $X$ affects $Y$ given do$(\text{par}(Y)\setminus X, E_Y)$. Therefore, given any unfaithful causal model where all nodes, including the error nodes are observed and can be intervened upon, full causal discovery is possible i.e., whether there exists a causal link $X\longrsquigarrow Y$ between any two nodes $X$ and $Y$ in the model, and whether this is a dashed or solid arrow can be determined by interventions in this case. While requiring all the nodes to be observable might be quite a strong assumption, we are not aware of a method for full causal discovery of arbitrary unfaithful causal models in previous literature even under this assumption. By introducing the new concept of higher-order affects relations, our framework suggests an advantage for the classical causal discovery problem for unfaithful causal models. The further exploration of the connections between our framework and the general causal discovery problem is left to future work.
\end{remark}

\subsection{Relationships between concepts}
\label{ssec: relations}
Due to the presence of fine-tuning and the introduction of the 2 types of causal arrows (solid and dashed), a number of concepts that are equivalent in faithful causal models are not equivalent for the causal models described in our framework. We summarise some of the relationships between the concepts arising in our causal modelling framework, before bringing space-time structure into the picture. This subsection can be skipped at the first reading.

The relationships are illustrated in Figure~\ref{fig: relations}. The reason for every implication is explained in the figure caption, and for every implication that fails, we provide a counter-example below. There are 14 implications in Figure~\ref{fig: relations} that do not hold. Some of these can be explained by the same counter-example or are immediately evident from the definitions. Therefore we first group these 14 cases based on the corresponding counter-example or argument needed for explaining them, in the end we will only need a few distinct counter-examples to cover all these cases. Note that if we restrict to faithful and/or acyclic causal models, not all of these non-implications would hold. For instance, in the case of faithful and acyclic causal models commonly considered in the literature, non-implications 1, 2, 3, 4, 5, 9 and 12 will become implications. This section does not cover all implication or non-implications found in this paper, since some of these also involve the newly introduced conditional HO affects relations. For this, we refer the reader to the previous sections. Here we consider relationships between certain basic notions such as correlation vs causation vs affects relations (unconditional zeroth-order ones), to illustrate how these differ in the fine-tuned case. 

\begin{enumerate}
 \item \textbf{Non-implication 1:} In unfaithful causal models, $X$ and $Y$ can be independent even when they are d-connected, as we have seen in the examples of Figure~\ref{fig: motiv_eg}. 
  \item \textbf{Non-implications 2, 11, 18:} These are covered by Example~\ref{example: eg3}.
 
    \item \textbf{Non-implications 3, 6, 8, 13:} These are covered by Example~\ref{example: jamming}.
     \item \textbf{Non-implications 4, 5:} $X$ is a cause of $Y$ does not imply that it is a direct cause of $Y$, it can be an indirect cause. Further $X$ can affect $Y$ even when it is an indirect cause, for example $X\longrightarrow Z\longrightarrow Y$.
    \item \textbf{Non-implication 7:} This is covered by Example~\ref{example: eg2}.
    \item \textbf{Non-implication 9:} It is evident that ``$X$ is a direct cause of $Y$'' does not imply $X\xdashrightarrow{} Y$, since it can also be a cause through a solid arrow.
    \item \textbf{Non-implications 10, 12:} These are just a consequence of the fact that correlation does not imply causation. Correlation between $X$ and $Y$ can arise when they share a common cause, without being a cause (direct or indirect) of each other.
     \item \textbf{Non-implications 14, 17:} In a simple common cause scenario, i.e., $Z\longrightarrow X$ and $Z\longrightarrow Y$ with $X=Y=Z$, $X$ does not affect $Y$ however $X$ is correlated with $Y$ and there is no dashed arrow from $X$ to $Y$.
    \item \textbf{Non-implication 15:} It is evident that independence of $X$ and $Y$ does not imply that there is a dashed arrow between them, they can also be d-separated.
     \item \textbf{Non-implication 16:} This is covered by Example~\ref{example: eg4}. 
\end{enumerate}
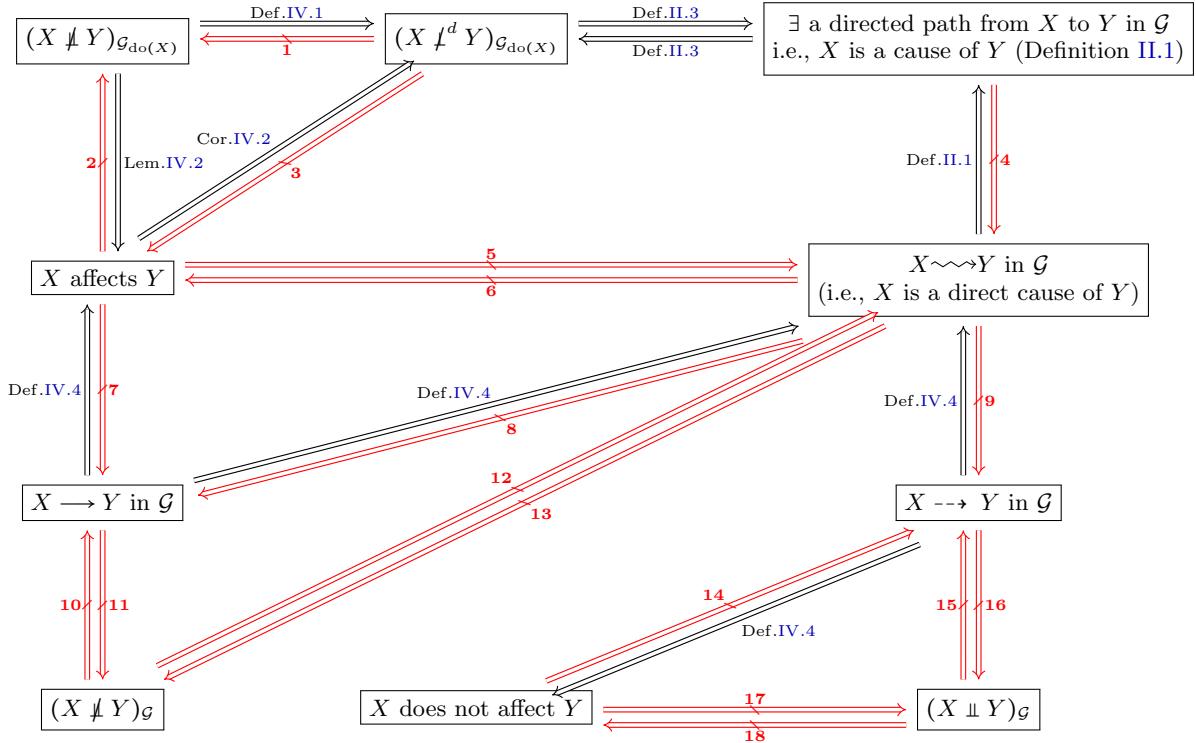
\begin{figure}[t!]
    \centering
    \begin{tikzcd}[arrows=Rightarrow, row sep=2cm, column sep=2cm]
\fbox{$(X\nindep Y)_{\mathcal{G}_{\mathrm{do}(X)}}$}\arrow[shift left=1.5ex]{d}{\mathrm{Lem.\ref{lemma: correl-affects}}}\arrow[shift left=1.5ex]{r}{\mathrm{Def.\ref{definition: compatdist}}}&\fbox{$(X\not\perp^d Y)_{\mathcal{G}_{\mathrm{do}(X)}}$}\arrow[red,degil]{dl}{\textcolor{red}{\textbf{3}}}\arrow[red,degil]{l}{\textcolor{red}{\textbf{1}}}\arrow[shift left=1.5ex]{r}{\mathrm{Def.\ref{definition: d-sep}}}&\mlnode{$\exists$ a directed path from $X$ to $Y$ in $\mathcal{G}$\\i.e., $X$ is a cause of $Y$ (Definition~\ref{def: cause})}\arrow{l}{\mathrm{Def.\ref{definition: d-sep}}}\arrow[red,degil,shift left=1.5ex]{d}{\textcolor{red}{\textbf{4}}}\\
   \fbox{$X$ affects $Y$} \arrow[shift left=1.5ex]{ur}{\mathrm{Cor.\ref{corollary:dsep-affects}}} \arrow[red,degil,shift left=1.5ex]{rr}{\textcolor{red}{\textbf{5}}}\arrow[red,degil]{d}{\textcolor{red}{\textbf{7}}}\arrow[red,degil]{u}{\textcolor{red}{\textbf{2}}} &    &\mlnode{$X\longrsquigarrow Y$ in $\cG$\\(i.e., $X$ is a direct cause of $Y$)}\arrow[red,degil]{ddll}{\textcolor{red}{\textbf{13}}}\arrow{u}{\mathrm{Def.\ref{def: cause}}} \arrow[ red,degil]{ll}{\textcolor{red}{\textbf{6}}} {}\arrow[red,degil]{d}{\textcolor{red}{\textbf{9}}}\arrow[red,degil, shift left=1.5ex]{dll}{\textcolor{red}{\textbf{8}}} \\
      \fbox{$X\longrightarrow Y$ in $\mathcal{G}$}\arrow{urr}{\mathrm{Def.\ref{definition: solidasharrows}}}  \arrow[shift left=1.5ex]{u}{\mathrm{Def.\ref{definition: solidasharrows}}}\arrow[red,degil]{d}{\textcolor{red}{\textbf{11}}} &  &  \fbox{$X\xdashrightarrow{} Y$ in $\mathcal{G}$}\arrow[shift left=1.5ex]{dl}{\mathrm{Def.\ref{definition: solidasharrows}}}\arrow[shift left=1.5ex]{u}{\mathrm{Def.\ref{definition: solidasharrows}}}\arrow[red,degil]{d}{\textcolor{red}{\textbf{16}}}\\
      \fbox{$(X\nindep Y)_{\mathcal{G}}$}\arrow[red,degil, shift left=1.5ex]{uurr}{\textcolor{red}{\textbf{12}}}\arrow[red,degil, shift left=1.5ex]{u}{\textcolor{red}{\textbf{10}}}&\fbox{$X$ does not affect $Y$}\arrow[red,degil]{ur}{\textcolor{red}{\textbf{14}}}\arrow[red,degil]{r}{\textcolor{red}{\textbf{17}}} &\fbox{$(X\indep Y)_{\mathcal{G}}$}\arrow[red,degil, shift left=1.5ex]{u}{\textcolor{red}{\textbf{15}}}\arrow[red,degil, shift left=1.5ex]{l}{\textcolor{red}{\textbf{18}}}
\end{tikzcd}
  \caption[Relationships between concepts associated with causal models]{\textbf{Relationships between concepts relating to causal models: }  The black arrows denote implications while red (crossed out) arrows denote non-implications. The numbers label the counter-examples corresponding to each non-implication, which are explained in the main text. The equivalence between ``$\exists$ a directed path from $X$ to $Y$ in $\cG$'' and $(X\not\perp^d Y)_{\cG_{\mathrm{do}(X)}}$ is explained in the paragraph following Corollary~\ref{corollary:dsep-affects}. $X\longrightarrow Y$ and $X\xdashrightarrow{} Y$ imply $X\longrsquigarrow Y$ since solid and dashed arrows are simply special instances of the more general, squiggly arrow by Definition~\ref{definition: solidasharrows}. This graph is complete in the sense that, given any ordered pair of statements $(\phi_1,\phi_2)$ from the 10 that form the vertices of this graph, one can deduce whether or not $\phi_1 \Rightarrow \phi_2$ as follows: if there exists a directed path from $\phi_1$ to $\phi_2$ that consists only of the implication arrows (black), then $\phi_1 \Rightarrow \phi_2$ and otherwise, $\phi_1 \not\Rightarrow \phi_2$.}
    \label{fig: relations}
\end{figure}

 \begin{example}
 \label{example: eg3}
Consider the causal structure of Figure~\ref{fig:eg3}. Let the three variables $S$, $E$ and $H$ be binary and correlated as $H=S\oplus E$ and $S=E$. These relations imply that $H=0$ deterministically. Now, when we intervene on $E$, we can choose its value independently of $S$ and whenever we choose $E\neq S$, we will see that $H=1$ occurs with non-zero probability. In other words, there exists a value $e$ of $E$ such that $P(H=1|\mathrm{do}(e))\neq P(H=1)=0$ i.e., $E$ affects $H$. As $E$ is a direct cause of $H$ in $\cG$, this further implies that the causal arrow from $E$ to $H$ is a solid one, even though $E$ and $H$ are independent in both the pre and post-intervention causal models i.e., $(E \indep H)_{\cG}$ and $(E \indep H)_{\cG_{\mathrm{do}(X)}}$ both hold, the former since $H$ is deterministic in the original causal model, irrespective of the value of $E$ and the latter since $H$ is uniform in the post-intervention model, again irrespective of the value of $E$. Therefore the existence of an affects relation between two sets of observed variables does not imply correlation between them either in the pre or the post intervention causal model. Further, $S$ does not affect $H$ since the exogeneity of $S$ implies that $P_{\cG{\mathrm{do}(S)}}(H|S)=P_{\cG}(H|S)$ (Corollary~\ref{corollary:exogenous}), and the independence of $S$ and $H$ in $\cG$ gives  $P_{\cG}(H|S)=P_{\cG}(H)$.
 \end{example}
\begin{figure}[t!]
    \centering
\includegraphics[]{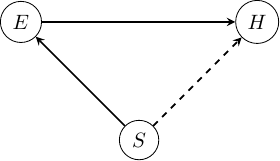}
    \caption[Operationally detectable causation without correlation]{\textbf{Affects relation does not imply correlation: } This is a causal structure for Example~\ref{example: eg3} which demonstrates a scenario where $E$ affects $H$ even though $P_{EH}=P_EP_H$, i.e., solid arrows can also be fine-tuned and the ability to detect causation through an active intervention does not imply that we will see correlation upon passive observation.}
    \label{fig:eg3}
\end{figure}

\begin{example}[Jamming]
\label{example: jamming}
Consider the causal structure of Figure~\ref{fig: jamming} where $B\xdashrightarrow{} A$, $B\xdashrightarrow{} C$ and the RVs $A$ and $C$ share an unobserved common cause $\Lambda$. By Definition~\ref{definition: solidasharrows} of the dashed arrows, we have $B$ does not affect $A$ and $B$ does not affect $C$. Suppose that $B$ affects the set $AC$. When $A$, $B$ and $C$ are binary, a probability distribution compatible with this situation is one where all 3 RVs are uniformly distributed and correlated as $B=A\oplus C$. Then, $A$ and $C$ individually carry no information about $B$ but $A$ and $C$ jointly determine the exact value of $B$. In this case, $B$ is a cause of $A$ and of $C$ but, due to fine-tuning, $B$ and $A$ are uncorrelated, as are $B$ and $C$, and there are no pairwise affects relations. This means that the causal influence of $B$ on $A$ (or $B$ on $C$) can only be detected when $A$, $B$ and $C$ are jointly accessed. The common cause is crucial to this example as explained in Figure~\ref{fig: jamming}, and the causal structure compatible with the distribution and affects relations of this example is not unique. An alternative causal structure that is compatible with correlations and affects relations of this example is where one of the dashed arrows $B\xdashrightarrow{} A$ or $B\xdashrightarrow{} C$ is dropped.
\end{example}
This example by itself makes no reference to space-time or the tripartite Bell scenario. However, if the variables $A$, $B$ and $C$ are embedded in a pairwise space-like separated way and taken to correspond to the output of Alice, input of Bob and output of Charlie respectively, this becomes a special case of the tripartite jamming scenario of~\cite{Grunhaus1996, Horodecki2019} (Figure~\ref{fig: trins}).\footnote{Barring the slight change of notation: In Figure~\ref{fig: trins}, $A$ and $C$ correspond to the inputs of Alice and Charlie while $X$ and $Z$ correspond to the outputs that are jammed by $B$. We do not make a distinction between inputs and outputs in general since we will also consider situations where the jamming variable is not exogenous for example.} In the rest of the paper, such examples, where an RV has dashed arrows to a set of RVs will be referred to as instances of ``jamming'' in accordance with the terminology of~\cite{Grunhaus1996}, irrespective of the space-time configuration. We will further discuss the relation of such causal models to space-time structure later in the paper.

\begin{figure}[t]
    \centering
  \subfloat[\label{fig: jamming}]{\includegraphics[]{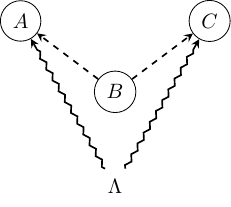}}\qquad\qquad\qquad\qquad\qquad\qquad\subfloat[\label{fig:example1}]{\includegraphics[]{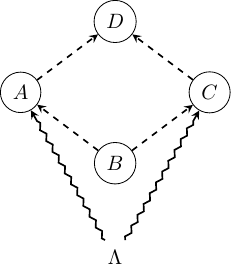}}
    \caption[Some fine-tuned causal structures]{\textbf{Some fine-tuned causal structures: } (a) The jamming causal structure of Example~\ref{example: jamming}. Note that the common cause $\Lambda$ is essential to this example,  because without $\Lambda$, $A$ and $C$ would be d-separated given $B$ which would imply the conditional independence $P_{AC|B}=P_{A|B}P_{C|B}$. The dashed arrows would imply the independence of $A$ and $B$ as well as $C$ and $B$ and hence the observed distribution would factorise as $P_{ABC}=P_AP_BP_C$. Then no pairs of disjoint subsets of $\{A,B,C\}$ would affect each other contrary to the original example. (b) Causal structure for Example~\ref{example: eg2} where $B$ affects $D$ even though there is no solid arrow path from $B$ to $D$. }
\end{figure}

\begin{example}
\label{example: eg2}
Consider a causal model over observed variables $\{A,B,C,D\}$ associated with the causal graph $\mathcal{G}$ given in Figure~\ref{fig:example1}. Here, there are no pairs of variables sharing an edge such that one of them affects the other. A correlation compatible with this graph is obtained by taking $B=A\oplus C=D$ where all variables are binary and uniformly distributed. Here, $B$ affects $D$ even though there are no solid arrow paths from $B$ to $D$. 
\end{example}

\begin{figure}[t!]
    \centering
 \subfloat[\label{fig: eg4a}]{  \includegraphics[]{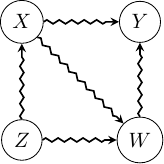}}\qquad\qquad \subfloat[\label{fig: eg4b}]{\includegraphics[]{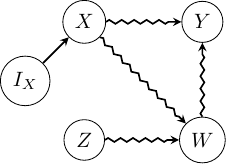}}\qquad\qquad\subfloat[\label{fig: eg4c}]{ \includegraphics[]{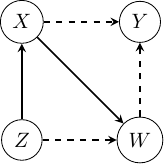}}
    \caption{\textbf{Dashed arrow (or non-affects relation) does not imply independence: } (a) The original causal structure $\cG$ of Example~\ref{example: eg4}, before the causal arrows are classified according to Definition~\ref{definition: solidasharrows}. (b) The corresponding causal structure $\cG_{\mathrm{do}(X)}$ when the node $X$ is intervened upon. (c) The causal structure $\cG$ after all the arrows have been classified as explained in the main text. The example shows that even though we have $X\xdashrightarrow{} Y$ in $\cG$, there exists a causal model compatible with this graph such that $X$ and $Y$ are correlated in $\cG$ in this causal model.}
    \label{fig: eg4}
\end{figure}

\begin{example}
\label{example: eg4}
Consider the causal structure of Figure~\ref{fig: eg4a} with the variables $X$, $Y$, $W$ and $Z$ taken to be binary. Suppose the causal mechanisms of the model are $X=Z$, $W=X\oplus Z$ and $Y=X\oplus W$ with the exogenous variable $Z$ being uniformly distributed. This reduces to $W=0$ (deterministically) and $Y=X=Z$. Since $Z$ is uniformly distributed, $P_{\cG}(Y)$ is also uniform and since $X$ and $Y$ are perfectly correlated in $\cG$, $P_{\cG}(Y|X)$ is deterministic. Now consider the graph $\cG_{\mathrm{do}(X)}$ shown in Figure~\ref{fig: eg4b}. The causal mechanism for $X$ here is fully specified by the distribution over $P_{I_X}$ which can be arbitrary. For the remaining variables we have $W=X\oplus Z$, $Y=X\oplus W$ and $Z$ is uniformly distributed, which gives $Y=Z$. The d-separation $(Z\perp^d X)_{\cG_{\mathrm{do}(X)}}$ implies the independence of $Z$ and $X$ in $\cG_{\mathrm{do}(X)}$ and hence the independence of $Y$ and $X$ in $\cG_{\mathrm{do}(X)}$ and since $Z$ is uniformly distributed here, so is $Y$ i.e., $P_{\cG_{\mathrm{do}(X)}}(Y|X)=P_{\cG_{\mathrm{do}(X)}}(Y)$ and both equal the uniform distribution. From before, we had noted that $P_\cG(Y)$ is also uniform, which gives $P_{\cG_{\mathrm{do}(X)}}(Y|X)=P_{\cG}(Y)$ or $X$ \emph{does not affect} $Y$. Therefore, by definition~\ref{definition: solidasharrows}, the causal arrow from $X$ to $Y$ must be a dashed one, even though we have seen that $(X\not\indep Y)_{\cG}$.
\end{example}
The remaining causal arrows of Figure~\ref{fig: eg4a} can also be classified as solid or dashed arrows as done for $X\longrsquigarrow Y$ in the above example.  For example, $X$ \emph{affects} $W$ can be established by noting that $P_{\cG_{\mathrm{do}(X)}}(W|X)$ is uniform (since $W=X\oplus Z$ with $X$ and $Z$ independent in $\cG_{\mathrm{do}(X)}$ and $Z$ is uniform) while $P_{\cG}(W)$ is deterministic. Therefore we have $X\longrightarrow W$ in $\cG$. Similarly, $W$ \emph{does not affect} $Y$, $Z$ \emph{affects} $X$ and $Z$ \emph{does not affect} $W$ can also be established and we obtain the graph of Figure~\ref{fig: eg4c} as the original causal structure $\cG$ once all the arrows of Figure~\ref{fig: eg4a} have been classified.

Further examples can be found in Appendix~\ref{appendix: examples} where we discuss how conditional independences and affects relations can be deduced from the causal model in our framework.

\section{The framework, Part 2: Space-time}
\label{sec:space-time}
We now turn to space-time structure and the relevant concepts needed for studying its relation to causality. Section~\ref{ssec: space-time} introduces our way of modelling space-time structure and the concept of \emph{ordered random variables} (definition~\ref{def: ORV}). In Section~\ref{ssec: embedding}, we define what it means to embed a causal model in a space-time structure (Definition~\ref{definition: embedding}). We then characterise in Section~\ref{ssec: compat}, what it means for a causal model to be compatible with an embedding in a space-time (Definition~\ref{definition: compatposet}), which formalises the requirement that signalling outside the space-time future is not possible using the affects relations of the embedded causal model. Finally, in Theorem~\ref{theorem:poset} of Section~\ref{sec: necsuff}, we provide necessary and sufficient conditions for compatibility.


\subsection{Space-time structure}
\label{ssec: space-time}
 We model space-time simply by a partially ordered set $\mathcal{T}$ without assuming any further structure/symmetries. A particular example of $\mathcal{T}$ is Minkowski space-time, where the partial order corresponds to the light-cone structure and the elements of $\mathcal{T}$ can be seen space-time coordinates in some frame of reference. Our results will only depend on the order relations of $\mathcal{T}$ and not on the representation of its particular elements. To make operational statements about $\mathcal{T}$, we must embed physical systems into it. In our case, we can only do so for the observed systems in the causal model which are random variables. We embed them in this space-time by assigning an element of $\mathcal{T}$ to each random variable which then specifies its space-time location (thereby producing an \emph{ordered random variable} or ORV), and assigning a subset of $\mathcal{T}$ to each ORV which specifies the locations in the space-time at which the ORV can be ``accessed''. Here, the order of an ORV corresponds to that of the space-time $\mathcal{T}$ (and not of the causal model) i.e., ORVs can be seen as abstract versions of space-time random variables.

\begin{definition}[Ordered random variable (ORV)]
\label{def: ORV}
Given a RV $X$, we can assign to it a location $O(X)\in\mathcal{T}$.  An ORV $\mathcal{X}$ is then the pair $\mathcal{X}:=(X,O(X))$. We can extend the definition of $O$ to ORVs, so that $O(\mathcal{X})$ is interpreted to mean $O(X)$.
\end{definition}

  We use $\prec$, $\succ$ and $\nprec\nsucc$ to denote the order relations for a given partially ordered set $\mathcal{T}$, where for $\alpha$, $\beta\in \mathcal{T}$, $\alpha \nprec\nsucc \beta$ corresponds to $\alpha$ and $\beta$ being unordered with respect to $\mathcal{T}$. This is different from $\alpha=\beta$ which corresponds to the two elements being equal. These relations carry forth in an obvious way to ORVs and we say for example that 2 ORVs $\mathcal{X}$ and $\mathcal{Y}$ are ordered as $\mathcal{X}\prec \mathcal{Y}$ iff $O(\mathcal{X})\prec O(\mathcal{Y})$. Note however that when we write $\mathcal{X}=\mathcal{Y}$ we mean $X=Y$ and $O(\mathcal{X})=O(\mathcal{Y})$.

 \begin{definition}
 \label{definition: incfuture}
 The \emph{inclusive future} of an ORV is the set 
 \begin{equation*}
    \overline{\mathcal{F}}(\mathcal{X}):=\{\alpha \in \mathcal{T}: \alpha\succeq O(\mathcal{X})\}.
\end{equation*}
 \end{definition}

Note that $\mathcal{X}\in \overline{\mathcal{F}}(\mathcal{X})$ but $\mathcal{X}\notin \mathcal{F}(\mathcal{X})$, hence the name ``inclusive'' future. Then, we say that an ORV $\mathcal{Y}$ lies in the inclusive future of an ORV $\mathcal{X}$ iff $O(\mathcal{Y})\in \overline{\mathcal{F}}(\mathcal{X})$. In a slight abuse of notation, we will simply write this as $\mathcal{Y}\in \overline{\mathcal{F}}(\mathcal{X})$, which is equivalent to $\mathcal{X}\preceq \mathcal{Y}$. Further, any probabilities written in terms of ORVs should be understood as being probability distributions over the corresponding random  variables. In the rest of the paper, whenever we use the term ``future'', this should be understood as inclusive future.

\begin{remark}
When considering causal loops or closed timelike curves (CTC)\footnote{By CTC we mean any situation in which a causal model whose causal structure has a loop is embedded in space-time (cf.\ Definition~\ref{definition: embedding}). This leads to causal influences in both directions between two points in the space-time.}, one typically imagines a cyclic space time whose light cone structure is not a partial order, but a pre-order. This is the case in general relativity where the space-time structure implies a causal structure and having a CTC is a property of the space-time. Here, we have separated causality from space-time such that causal loops are a property of the causal model (see Section~\ref{sec: loops}), and any causal loop embedded in a space-time (partial or pre-ordered) as described in the following section would form a CTC. We will consider how such cyclic causal models can be compatibly embedded in a space-time i.e., without leading to signalling outside the future, and the more interesting case is when we take a partially ordered space-time such as Minkowski space-time. Through this approach, we will see that it is possible to have a CTC in Minkowski space-time that does not lead to superluminal signalling, since it is possible for the signalling properties of a causal model to respect the partial order even while the causal relations are cyclic. The problem would in a sense be trivial if the space-time is also a preorder, since for any cyclic causal structure (which defines a pre-order relation), one can always find a corresponding pre-ordered space-time that compatibly embeds it.
\end{remark}

\subsection{Embedding of a causal model in a space-time structure}
\label{ssec: embedding}

We have discussed two types of order relations: the pre-order encoded by the arrows $\longrsquigarrow$ of the causal structure, and the partial order specified by the order relation $\prec$ of the partial order $\mathcal{T}$. These are two distinct concepts, and within our framework can be set independently of one another. We first formalise how a given causal model may be embedded in a space-time structure, and in the next section, we introduce a compatibility condition that connects the two that aims to capture when a causal structure can be embedded in the partial order $\mathcal{T}$. This compatibility condition is based on the idea of ensuring that it is impossible to signal outside the future as encoded by the partial order $\mathcal{T}$\footnote{It may be helpful to think of $\mathcal{T}$ as a Minkowski space-time, with the partial order specified by the light-cone structure.}.  Whether signalling is possible depends on where random variables can be accessed, and so we first introduce the concept of an accessible region, which is the subset of $\mathcal{T}$ at which it is possible to have a copy of a random variable.  Since we are dealing with classical random variables, it makes sense to imagine these being broadcast, i.e., sending a copy to all points in the accessible region.

\begin{definition}[Copy of a RV]
\label{definition:copy}
Consider a causal model over a set of observed variables $S$. A RV $X'\in S$ is a \emph{copy} of $X\in S$ if the only parent of $X'$ is $X$, and if $X'=X$.  It is often convenient to think of copying a random variable $X$ in the causal model, where the copy is not initially included in the model.  To do so, we augment the causal graph with a new node $X'$ whose only parent is the node $X$ and such that $X'=X$ (the graph has $X\longrightarrow X'$ added).  We usually do not draw the augmented causal model, but instead keep the copies implicit. We also extend the definition of a copy to ordered random variables so that $\mathcal{X}'$ is a copy of $\mathcal{X}$ whenever the corresponding RV $X'$ is a copy of the RV $X$.
\end{definition}

Note that each RV affects each of its copies. We can then define the accessible region of a RV to be the region of $\mathcal{T}$ in which it is possible to have a copy of the RV.  In essence, we can imagine each RV being copied throughout its accessible region.

\begin{definition}[Accessible region of a RV/ORV]
\label{definition:accreg}
Given a causal model over a set of observed variables $S$, and a partial order $\mathcal{T}$, for each random variable $X\in S$ we can define an \emph{accessible region} $\mathcal{R}_X\subseteq\mathcal{T}$ intended to represent the set of points in $\mathcal{T}$ at which it is possible to have a copy of $X$. The \emph{inaccessible region} of $X$ is then the complement $\tilde{\mathcal{R}}_X=\mathcal{T}\setminus\mathcal{R}_X$ and represents the set of points at which it is impossible to have a copy of $X$. We can naturally extend this definition to ORVs by taking the accessible region of an ORV $\mathcal{X}=(X,O(X))$ to be the accessible region of $X$.
\end{definition}

We also want a notion of accessible region for sets of RVs/ORVs. The accessible region of a set can be thought of as the locations at which there can be a copy of all of the random variables in the set. This motivates taking the intersection of the accessible regions of the individual elements, since if the accessible region of the set were any larger than this, it would contradict the definition of accessible region for at least one individual element of the set.

\begin{definition}[Accessible region of a set of RVs/ORVs]
\label{definition: subsets}
Given a set $S=\{S_i\}_i$ of RVs we define the accessible region of $S$ by $\mathcal{R}_{S}=\bigcap\limits_{S_i\in S}\mathcal{R}_{S_i}$.  For the empty set $\emptyset$, the accessible region is defined to be $\mathcal{R}_{\emptyset}:=\mathcal{T}$. 
\end{definition}

\begin{definition}[Embedding]
\label{definition: embedding}
Given a set of RVs $S$, an \emph{embedding of $S$} in a partially ordered set $\mathcal{T}$ produces a corresponding set of ORVs $\mathcal{S}$ by assigning a location $O(X)\in\mathcal{T}$, and an accessible region $\mathcal{R}_X$ to each RV $X$, such that the associated ORV is $\mathcal{X}=(X,O(X))$. An embedding of a set of RVs is called \emph{non-trivial} if no two RVs $X$ and $Y$ such that $X$ affects $Y$ are assigned the same location in $\mathcal{T}$. 
\end{definition}

The set of RVs $S$ we will wish to embed will typically be related by a causal model or a set of affects relations. We have seen that when analysing affects relations, it is useful to augment the original causal model with an additional set of RVs corresponding to the intervention nodes. In the following, whenever we refer to an embedding of a causal model or a set of affects relations in a partial order, this must be understood as an embedding of the original set of RVs $S$ associated with causal model/affects relations, the non-triviality of the embedding will also only concern the embedding of the original set of RVs $S$. For simplicitly, we will assume that every hypothetical intervention node $I_X$ that may be introduced to model interventions on an RV $X\in S$ is embedded at the same location as $X$ (even though $I_X$ affects $X$ by construction). Our results are not affected by this assumption, it is a mere simplification.

\subsection{Compatibility of a causal model with an embedding in space-time}
\label{ssec: compat}

Up to here there are no conditions on how the locations and accessible regions are set---in particular, these need not be related with the notion of future defined on $\mathcal{T}$. We now introduce a compatibility condition that connects these concepts together, which aims to capture the intuition that signalling outside the (inclusive) future should not be possible. As this intuition is non-trivial to formalise for general, unfaithful causal models, we will first motivate the important aspects of the definition with examples, before formally stating it. For this, we will first consider the case of faithful causal models, then unfaithful causal models with interventions only on single nodes and finally the general case of unfaithful causal models with joint interventions. For all the examples in the following paragraphs we will take $\mathcal{T}$ to be Minkowski space-time and embed RVs such that the accessible region of each RV coincides with its inclusive future.

\bigskip
\paragraph{Compatibility for faithful causal models: } For faithful causal models, if $X$ and $Y$ are 2 RVs, $X$ is a cause of $Y$ in a causal structure $\cG$ i.e., $X\longrsquigarrow\ldots\longrsquigarrow Y$ in $\cG$ is equivalent to $X$ affects $Y$. 
 Therefore, if we demand that whenever $X$ affects $Y$ for any two RVs $X$ and $Y$ in the model, $Y$ must be embedded in the future of $X$ in the space-time, this ensures that all causal influences propagate from past to future and consequently that there is no signalling outside the future for the given embedding of the model.\footnote{Note that such an embedding is always possible for acyclic causal models but impossible for causal models with certain types of causal loops (Lemma~\ref{lemma: TrivEmbedding}) and possible for causal models with certain other types of causal loops as we will show in Section~\ref{sec: loops}.}

\bigskip
\paragraph{Compatibility for unfaithful causal models with single node interventions: } The above condition for faithful models is insufficient to rule out such signalling in unfaithful models since affects and cause become inequivalent notions here, and we must also consider affects relations involving sets of RVs. For example, in the jamming causal structure (Example~\ref{example: jamming}), if we embed $A$ and $C$ outside the future of $B$, but such that there are points in the intersection of the futures of $A$ and $C$ that are also outside the future of $B$, then signalling is possible. We first consider affects relations of the form $X$ affects $S$ where $X$ is an RV and $S$ is a set of RVs. Operationally, this means that given access to a copy of all elements of $S$, one can learn information about the intervention performed on $X$. Then, in order to avoid signalling outside the future by means of the affects relation $X$ affects $S$, a necessary and sufficient condition on the embedding would be to take the accessible regions to coincide with the inclusive futures and $\mathcal{R}_S\subseteq\mathcal{R}_{\cX}$, which would ensure that the joint future of all elements in $S$ is contained in the future of $X$. Note that this does not imply that all causal influences (which may be hidden due to fine-tuning) must propagate from past to future, only that any observable signal propagates from past to future (cf. the jamming scenario of Figure~\ref{fig: trins}).

\bigskip
\paragraph{Compatibility for unfaithful causal models with multi-node interventions: } Consider a general affects relation of the form $S_1$ affects $S_2$ for two disjoint subsets $S_1$ and $S_2$ of RVs (possibly arising from an unfaithful causal model). If in analogy to the previous case, we demand that any compatible embedding must be such that $\mathcal{R}_{S_2}\subseteq\mathcal{R}_{S_1}$ with all accessible regions coinciding with the corresponding inclusive futures, this would be too restrictive in the present case. Take the simple Example~\ref{example: HOaffects1} where $Z\longrightarrow Y$ and $X$ is an isolated node with no in or out edges. Then clearly $XZ$ affects $Y$ but we would only require $Y$ to be in the future of $Z$ and not also in the future of $X$ (which trivially affects it given $Z$). On the other hand, in the causal structure of Example~\ref{example: HOaffects2}, $Y$ depends on both the exogenous nodes $X$ and $Z$ and we would expect that $Y$ must be embedded in the joint future of $X$ and $Z$ to avoid signalling outside the future. To establish that embedding $Y$ in the joint future of $X$ and $Z$ is necessary in the latter case and not the former and to avoid imposing too strong constraints on the embedding, we must also consider the higher-order affects relation $X$ affects $Y$ given do$(Z)$.

\bigskip

\paragraph{Operational meaning of a higher-order affects relation: }Operationally, the conditional HO affects relation $X$ affects $Y$ given $\{\mathrm{do}(Z),W\}$ means that an agent Alice who can intervene on $X$ can signal to an agent Bob having access to $Y$ if Bob also has access to information about interventions performed on some set $Z$ along with information about some other set $W$ (upon which an intervention was not performed). If the RVs in these sets are embedded in a space-time, in order for the affects relation $X$ affects $Y$ given $\{\mathrm{do}(Z),W\}$ to not lead to signalling outside the space-time future, we must embed the RVs such that the joint future of $Y$, $Z$ and $W$ (i.e., where they are jointly accessible by Bob) is contained in the future of $X$.

 Furthermore, a given HO affects relation, $X$ affects $Y$ given $\{\mathrm{do}(Z),W\}$ may itself contain some redundancies if $X$ is a set of RVs (as we have seen in Example~\ref{example: HOaffects2}), such that it can be reduced to the HO affects relation $\tilde{s}_X$ affects $Y$ given $\{\mathrm{do}(Z),W\}$ for some proper subset $\tilde{s}_X$ of $X$ (Lemma~\ref{lemma:reduce}). In such cases we only need to impose that the joint future (or joint accessible region) of $Y$ and $Z$ is contained in that of the smaller set $\tilde{s}_X$.

The following definition based on this intuition allows us to decide when a set of affects relations can be compatibility embedded in a space-time.

\begin{definition}[Compatibility of a set of affects relations with an embedding in a partial order ($\mathbf{compat}$)]
\label{definition: compatposet}
Let $\mathcal{S}$ be a set of ORVs formed by embedding a set of RVs $S$ in a partially ordered set $\mathcal{T}$ with embedding $\mathscr{E}$. Then a set of affects relations $\mathscr{A}$ is said to be \emph{compatible} with the embedding $\mathscr{E}$ if the following conditions hold:
\begin{itemize}
\item \textbf{compat1: } Let $\cS_1,\cS_2\subseteq\cS$ be disjoint non-empty subsets of ORVs, and $\cS_3,\cS_4$ be two more subsets (possibly empty) disjoint from each other and $\cS_1$ and $\cS_2$.  If $\cS_1$ affects $\cS_2$ given $\{\mathrm{do}(\cS_3),\cS_4\}$ is in $\mathscr{A}$ and is irreducible with respect to the affects relations in $\mathscr{A}$, then $\cR_{\cS_2\cS_3\cS_4}=\cR_{\cS_2}\bigcap \cR_{\cS_3}\bigcap \cR_{\cS_4}\subseteq \cR_{\cS_1} $ with respect to $\mathscr{E}$.

\item \textbf{compat2: } for all $\cX\in\cS$, $\mathcal{R}_{\cX}=\overline{\mathcal{F}}(\mathcal{X})$ with respect to $\mathscr{E}$.
\end{itemize}
\end{definition}

The definition is motivated by the desire to prevent signalling outside of the future. The condition \textbf{compat2} identifies the accessible region with the inclusive future, which is based on the ability to broadcast a RV to any location in its future. An alternative would be a weaker condition that requires the accessible region to be some subset of the future. The condition \textbf{compat1} is defined in terms of accessible regions, so could also be used with a weaker version of \textbf{compat2}. However, a weaker version would in effect place a constraint on broadcasting, and we do not use it here. We return to this in Section~\ref{sec: necsuff}.

This definition covers all the special cases previously discussed. For single variables, if $X$ affects $Y$ then $\cY$ should be in the future of $\cX$ (given \textbf{compat2} this is equivalent to taking the accessible region of $Y$ to be contained within that of $X$); this is \textbf{compat1} when $\cS_3$ is the empty set (in which case its accessible region is simply $\mathcal{T}$ by Definition~\ref{definition: subsets}) and $\cS_1=\cX$ and $\cS_2=\cY$ are single ORVs. When $\cS_2$ is a set of ORVs, this case ensures that the ORVs in $\cS_2$ are jointly accessible only in the future of the ORV $\cX$. This covers the particular case of jamming (Example~\ref{example: jamming}).

We now illustrate the definition by applying it to Examples~\ref{example: HOaffects1} to~\ref{example: HOaffects3}. In Example~\ref{example: HOaffects1}, $Z$ affects $Y$ implies that $Y$ must be in the future of $Z$ and $XZ$ affects $Y$ being a reducible affects relation does not add any further constraints, so we do not require $Y$ to be in the future of $X$. In Example~\ref{example: HOaffects2}, we again must embed $Y$ in the future of $Z$ but in this case, $XZ$ affects $Y$ is irreducible and therefore imposes the constraint that $Y$ must be in the joint future of $X$ and $Z$. In Example~\ref{example: HOaffects3}, in contrast to the previous example, we have $X$ affects $Y$ given do$(Z)$ even though $XZ$ does not affect $Y$. The former is irreducible as it involves single RVs, and implies that the joint future of $Y$ and $Z$ must be in the future of $X$, and since we also have $Z$ affects $Y$ which would require $Y$ to be in the future of $Z$, we can conclude that compatibility in this case forces $Y$ to be in the future of both $X$ and $Z$. Noting that $W$ also affects $Y$, this would require $Y$ to be in the future of $W$ as well.
\bigskip

\paragraph{Completeness of Definition~\ref{definition: compatposet}: } We now provide an argument to show that our definition of compatibility indeed fully captures the intuition of ``no signalling outside the future'' within our framework. Given \textbf{compat2} which we have motivated above, \textbf{compat1} is necessary to avoid agents from using the affects relation to signal outside the future, since a violation of compatibility would enable $\cS_2$, $\cS_3$ and $\cS_4$ to be accessed outside the future of $\cS_1$ and yet receive a signal from $\cS_1$ through the irreducible affects relation $\cS_1$ affects $\cS_2$ given $\{\mathrm{do}(\cS_3),\cS_4\}$.\footnote{Without \textbf{compat2}, \textbf{compat1} is necessary for ``no signalling outside the accessible region''. See Section~\ref{ssec: missing_assump} for further discussion.}

Further if a set of affects relations satisfy our definition with respect to some space-time embedding, this is sufficient to ensure that no agents who can access the associated ORVs can signal outside the future using those affects relations.\footnote{This applies given the setup assumptions of the framework, such as that interventions are performed independently on each node $X$ and correspond to an exogenous variable $I_X$ etc.} 
This is because the conditional HO affects relation $\cS_1$ affects $\cS_2$ given $\{\mathrm{do}(\cS_3),\cS_4\}$ (between arbitrary pairwise disjoint sets of ORVs) captures the most general way in which agents can signal to each other in our framework: an agent Alice may intervene on a set $\cS_1$ of observed nodes, and an agent Bob with access to another set of observed RVs $\cS_2$, can try to detect the effect of Alice's intervention and Bob may additionally have access to some combination of observational ($\cS_4$) and interventional data (do$(\cS_3)$) relating to other sets of the observed nodes. Therefore, demanding that $\overline{\mathcal{F}}(\cS_2)\bigcap\overline{\mathcal{F}}(\cS_3)\bigcap\overline{\mathcal{F}}(\cS_4)\subseteq \overline{\mathcal{F}}(\cS_1)$ holds for any space-time embedding of the RVs will be sufficient to ensure that this affects relations cannot be used to signal outside the space-time's future. However, this turns out to be too strong a sufficiency condition, and imposing this only for irreducible affects relations (as the definition does) is already sufficient. To see this,
suppose that $\cS_1$ affects $\cS_2$ given $\{\mathrm{do}(\cS_3),\cS_4\}$ is reducible. Then there exists a subset $s_1\subset \cS_1$ such that $s_1$ affects $\cS_2$ given $\{\mathrm{do}(\cS_3),\cS_4\}$ (cf.\ Lemma~\ref{lemma:reduce}). Without loss of generality, take this to be irreducible (if not, simply find a subset of $s_1$ that satisfies the same affects relation and repeat this argument), then requiring $\overline{\mathcal{F}}(\cS_2)\bigcap\overline{\mathcal{F}}(\cS_3)\bigcap\overline{\mathcal{F}}(\cS_4)\subseteq \overline{\mathcal{F}}(s_1)$ is sufficient to ensure that the reduced affects relation $s_1$ affects $\cS_2$ given $\{\mathrm{do}(\cS_3),\cS_4\}$ does not signal outside the future. By the reducibility of the original relation $\cS_1$ affects $\cS_2$ given $\{\mathrm{do}(\cS_3),\cS_4\}$, we have for $\tilde{s}_1:=\cS_1\backslash\{s_1\}$, $\tilde{s}_1$ does not affect $\cS_2$ given $\{\mathrm{do}(\cS_3 s_1),\cS_4\}$, which means the original affects relation does not require $\overline{\mathcal{F}}(\cS_2)\bigcap\overline{\mathcal{F}}(\cS_3)\bigcap\overline{\mathcal{F}}(\cS_4)\bigcap \overline{\mathcal{F}}(s_1)=\overline{\mathcal{F}}(\cS_2)\bigcap\overline{\mathcal{F}}(\cS_3)\bigcap\overline{\mathcal{F}}(\cS_4)$ to be in contained in the future of $\tilde{s}_1$, once we have imposed $\overline{\mathcal{F}}(\cS_2)\bigcap\overline{\mathcal{F}}(\cS_3)\bigcap\overline{\mathcal{F}}(\cS_4)\subseteq \overline{\mathcal{F}}(s_1)$ for the corresponding reduced relation.

While these arguments justify the completeness of our definition, they do not rule out the possibility of another definition that captures the same intuition. This would also depend on how ``no signalling outside the future'' is interpreted, and this can be done in several inequivalent ways (e.g., by taking the accessible regions to be a subset of the future in \textbf{compat2}), we have proposed one possible, natural way to formalise this. We discuss similar but distinct compatibility conditions in Section~\ref{sec: necsuff}.

\begin{remark}
\label{remark: HOaff_complete}
Given a set $\mathscr{A}$ of arbitrary conditional  
affects relations (including zeroth and HO relations), one can use the first part of Lemma~\ref{lemma: HOaff_complete} to convert this to a new set $\tilde{\mathscr{A}}$ containing only unconditional affects relations such that compatibility of $\mathscr{A}$ with an embedding $\mathscr{E}$ in a space-time $\mathcal{T}$ implies the compatibility of $\tilde{\mathscr{A}}$ with the same embedding. For this, form $\tilde{\mathscr{A}}$ from $\mathscr{A}$ by including every unconditional affects relation from $\mathscr{A}$, and for every conditional affects relation $S_1$ affects $S_2$ given $\{\mathrm{do}(S_3),S_4\}$ in $\mathscr{A}$, add the corresponding unconditional affects relation $S_1$ affects $\{S_2,S_4\}$ given do$(S_3)$ in $\tilde{\mathscr{A}}$, if the latter was not already included in $\mathscr{A}$ (note that the former implies the latter by part 1 of Lemma~\ref{lemma: HOaff_complete}). Now, every irreducible conditional relation $S_1$ affects $S_2$ given $\{\mathrm{do}(S_3),S_4\}$ in $\mathscr{A}$ imposes the condition $\overline{\mathcal{F}}(\cS_2)\bigcap\overline{\mathcal{F}}(\cS_3)\bigcap\overline{\mathcal{F}}(\cS_4)\subseteq\overline{\mathcal{F}}(\cS_1)$ on any compatible space-time embedding $\mathscr{E}$. By part 2 of Lemma~\ref{lemma: HOaff_complete}, irreducibility of $S_1$ affects $S_2$ given $\{\mathrm{do}(S_3),S_4\}$ in $\mathscr{A}$ implies irreducibility of $S_1$ affects $\{S_2,S_4\}$ given do$(S_3)$ in $\tilde{\mathscr{A}}$, and the latter imposes the same condition on the embedding, by Definition~\ref{definition: compatposet}. Every unconditional relation is present in both sets and hence imply the same conditions on the embedding. 

In summary, every affects relation of the form $S_1$ affects $S_2$ given $\{\mathrm{do}(S_3),S_4\}$ present in $\mathscr{A}$ can be replaced by $S_1$ affects $\{S_2,S_4\}$ given do$(S_3)$ for the purpose of applying Definition~\ref{definition: compatposet}.

\end{remark}

\begin{remark}
A complete set of affects relations for a causal model over a set $S$ of RVs is one where for any subsets $S_1,S_2,S_3,S_4$ of $S$ we know whether or not $S_1$ affects $S_2$ given $\{\mathrm{do}(S_3),S_4\}$. It is not always possible to deduce a complete set of affects relations from a causal model (as defined in Definition~\ref{def: causalmodel}), and in general a complete set may not be available. Use of a partial set of affects relations can be sufficient to deduce incompatibility with an embedding, and, given a causal model, a partial set can be deduced.  Note that we require causal models to define affects relations in the first place.
\end{remark}

\begin{definition}[Compatibility of a causal model with an embedding in a partial order]
We say that a causal model over a set of RVs $S$ is compatible with an embedding in a partial order if the set of affects relations $\mathscr{A}$ implied by the causal model are compatible with the embedding (cf.\ Definition~\ref{definition: compatposet}).
\end{definition}

\begin{remark}
  If $X\xdashrightarrow{}Y$, there is no affects relation between $X$ and $Y$ and our compatibility condition does not require that for the corresponding ORV $\cY$, $\cY\in\mathcal{R}_X$.  Although demanding this would be natural in light of common notions of causation, one of the motivations behind this line of research is to investigate what happens without this because the existence of such causal influences may not be operationally detectable. In other words, our compatibility condition does not imply that cause precedes effect with respect to the space-time order relation, but it does imply that signalling is not possible outside the future of the space-time structure. Interestingly, this does not rule out the possibility of causal models with causal loops than can be compatibly embedded in the space-time, as we show in an associated letter \cite{VilasiniColbeckPRL} (reviewing this argument in Section~\ref{sec: loops}).

\end{remark}

\subsection{Necessary and sufficient conditions for compatibility}
\label{sec: necsuff}

For compatibility of a set of affects relations with an embedding $\mathscr{E}$ in space-time, Definition~\ref{definition: compatposet} states that the conditions \textbf{compat1} and \textbf{compat2} must be satisfied. Consider now a similar condition which we call \textbf{compat1$'$($\cS, \mathscr{A}$)}, where we use the arguments in the brackets specify the set of ORVs and affects relations that the condition is applied to (since we will later apply it to a different set). With this convention, \textbf{compat1}$:=$\textbf{compat1($\cS, \mathscr{A}$)}.

\bigskip\noindent
\textbf{compat1$'$($\cS, \mathscr{A}$): }  Let $\cS_1,\cS_2\subseteq\cS$ be disjoint proper subsets of ORVs, and $\cS_3,\cS_4$ be two other subsets (possible empty) disjoint from themselves and $\cS_1$ and $\cS_2$. If $\cS_1$ affects $\cS_2$ given $\{\mathrm{do}(\cS_3),\cS_4\}$ is in $\mathscr{A}$ and is irreducible with respect to the affects relations in $\mathscr{A}$, then $\bigcap\limits_{s_{234}\in \mathcal{S}_2 \cS_3 \cS_4}\overline{\mathcal{F}}(s_{234}) \subseteq \bigcap\limits_{s_1\in \mathcal{S}_1}\overline{\mathcal{F}}(s_1)$ with respect to $\mathscr{E}$.

\bigskip
Note that \textbf{compat1$'$($\cS, \mathscr{A}$)} imposes no condition on the accessible regions but only on the space-time locations of the ORVs (which allow us to fully specify their inclusive futures), while \textbf{compat1} restricts the accessible regions. However, once \textbf{compat2} is imposed, \textbf{compat1} and \textbf{compat1$'$($\cS, \mathscr{A}$)} are essentially equivalent i.e., an equivalent definition of compatibility would be to use \textbf{compat1$'$($\cS, \mathscr{A}$)} and \textbf{compat2} instead of \textbf{compat1} and \textbf{compat2} in Definition~\ref{definition: compatposet}. We use \textbf{compat1} instead of \textbf{compat1$'$($\cS, \mathscr{A}$)} in the original definition to make it clear that this condition is related to the operational concept of ``accessibility'' of ORVs, which is captured by the accessible regions. In general, the accessible region of an ORV need not be fully specified by its space-time location or even be related to its future, but this is the case once \textbf{compat2} is assumed. The following theorem (proven in Appendix~\ref{appendix: proofs3}) and corollary establish certain useful connections between these concepts, and follow from Definition~\ref{definition: compatposet}.

\begin{restatable}{theorem}{CompatST}[Necessary and sufficient conditions for compatibility with an embedding in $\mathcal{T}$]
\label{theorem:poset}
    Let $\cS$ be set of ORVs embedded in a partial order $\mathcal{T}$ with respect to an embedding $\mathscr{E}$ and let $\mathscr{A}$ be a given set of affects relations on $\cS$. Further, consider forming an augmented set of ORVs $\mathcal{S}'$ by taking $\mathcal{S}$ and for each variable $\cX\in\mathcal{S}$, embedding a copy of $\cX$ at each point in its accessible region $\cR_{\cX}$ and form $\mathscr{A}'$ by adding to $\mathscr{A}$ that each variable affects each of its copies for all copies. Then the following statements hold.
    
    \begin{enumerate}
        \item If the set of affects relations $\mathscr{A}$ is compatible with the embedding $\mathscr{E}$ in $\mathcal{T}$, then $\mathbf{compat1}'(\cS', \mathscr{A}')$ holds i.e., $\mathbf{compat1}'(\cS', \mathscr{A}')$ is necessary for compatibility of $\mathscr{A}$ with the space-time embedding $\mathscr{E}$.

\item  $\mathbf{compat1'(\cS', \mathscr{A}')}$ implies that $\cR_{\cX}\subseteq \overline{\cF}(\cX)$ $\forall \cX\in \cS$, but not that the two sets $\cR_{\cX}$ and $\overline{\cF}(\cX)$ are necessarily equal $\forall \cX\in \cS$, i.e., $\mathbf{compat1'(\cS', \mathscr{A}')}$ is not sufficient for compatibility of $\mathscr{A}$ with the space-time embedding $\mathscr{E}$. 
    \end{enumerate}

\end{restatable}

The augmented sets $\mathcal{S}'$ and $\mathscr{A}'$ in the above theorem capture the idea of broadcasting classical RVs to each point in their accessible region. Imposing $\mathbf{compat1'}$ for the an embedding $\mathscr{E}$ of these sets in space-time then ensures that this broadcasting (i.e., finding copies of the RVs) is possible only within the future, but not necessarily to all locations in the future. Note that being able to find copies of an ORV $\cX$ only within its future does not by itself imply that any ORV $\cY$ affected by $X$ must be contained in its future. We then have the following corollary of the theorem. 
\begin{corollary}
\label{corollary:poset}
Let $S$ be a set of RVs and $\mathscr{A}$ be a set of affects relations over them. Then there exists a non-trivial embedding $\mathscr{E}$ of $S$ in a partial order $\cT$ compatible with $\mathscr{A}$ if and only if there exists a non-trivial embedding $\mathscr{E}'$ of the same affects relations that satisfies $\mathbf{compat1}'(\cS', \mathscr{A}')$.
\end{corollary}
That the existence of a non-trivial embedding $\mathscr{E}$ that satisfies \textbf{compat} implies the existence of one that satisfies \textbf{compat1$'$($\cS', \mathscr{A}'$)} follows directly from the necessary part of Theorem~\ref{theorem:poset}. The other direction follows because any non-trivial embedding $\mathscr{E}'$ that satisfies \textbf{compat1$'$($\cS', \mathscr{A}'$)} can be turned into a non-trivial embedding $\mathscr{E}$ that satisfies \textbf{compat} simply by taking $\mathscr{E}'$ and setting the accessible regions of ORVs to satisfy \textbf{compat2}. The important point to note is that the two embeddings $\mathscr{E}$ and $\mathscr{E}'$ need not be the same.

\section{Causal loops and their space-time embeddings}
\label{sec: loops}
 We have characterised a general class of causal models, defined when a given causal model can be said to be compatible with a space-time embedding and also compared related yet distinct conditions on the space-time embeddings. It is interesting to consider whether there are certain structural properties of the causal model alone that guarantee the existence of a non-trivial and compatible space-time embedding for that causal model. Clearly, the acyclicity of the causal structure is such a property, while this is certainly sufficient, a natural question is whether is it also necessary to guarantee the existence of such a space-time embedding. This question motivates us to define a broad set of possible theories $\mathbb{T}$  that are consistent with the principle of ``no signalling outside the future''. The set $\mathbb{T}$ consists of theories with the property that for every causal models that can arise in the theory, there exists a non-trivial and compatible embedding in a space-time (cf.\ Definition~\ref{definition: compatposet}).

This class of theories is quite general, it certainly includes quantum and standard GPTs and any theory that can be characterised using acyclic causal structure. In an associated Letter~\cite{VilasiniColbeckPRL}, we apply the framework developed here to construct an explicit operationally detectable causal loop that can be embedded in (1+1)-dimensional Minkowski space-time without superluminal signalling, which demonstrates that the set $\mathbb{T}$ can also include theories admitting causal loops. In this section, we characterise several different classes of causal loops that can arise in our framework, and we show that some of these classes can be ruled out by requiring that the causal model has a compatible space-time embedding while the results of the associated Letter~\cite{VilasiniColbeckPRL} show that some other classes cannot be ruled out in this manner. We provide further examples to argue that fully characterising the set of theories $\mathbb{T}$ may be a difficult task. By full characterisation, we mean finding a necessary and sufficient set of conditions on the set of possible affects relations (and/or correlations) of the causal model that guarantees the existence of a non-trivial compatible space-time embedding. Let us now take a closer look at the types of causal loops that can arise in our framework.

\subsection{Different classes of causal loops }
We have seen that due to fine-tuning, causation does not imply the existence of affects relations. This motivated the classification of causal arrows (Definition~\ref{definition: solidasharrows}) as solid or dashed based on the existence of suitable affects relations. Similarly, we can distinguish between different types of causal loops in our framework depending on whether they can be operationally detected through their affects relations. A causal loop simply corresponds to a directed cycle in a causal structure $\mathcal{G}$ involving at least two observed nodes i.e., two observed nodes $X$ and $Y$ in $\cG$ such that there exist directed paths from $X$ to $Y$ and from $Y$ to $X$. Often however, we may not know the full causal structure but only a set of affects relations $\mathscr{A}$ over the observed nodes of an underlying causal structure $\cG$. The set $\mathscr{A}$ might allow us to infer some, but not necessarily all the causal relationships in $\cG$. We then have the following two broad categories of causal loops, the former (affects causal loops) are operationally detectable via the their affects relations and the latter (hidden causal loops) are not operationally detectable through their affects relations or correlations.

\begin{definition}[Affects causal loops (ACL)]
\label{def: ACL}
Any set of affects relations $\mathscr{A}$ that can only arise in a causal model associated with a cyclic causal structure $\mathcal{G}$ are said to form/contain an affects causal loop. In other words, affects causal loops certify the cyclicity of the underlying causal structure through the observed affects relations.
\end{definition}

\begin{definition}[Hidden causal loop (HCL)]
\label{def: HCL}
Given a causal model whose causal structure contains a directed cycle, and a complete set of affects relations, we say that this causal model contains a \emph{hidden causal loop} if the same set of affects relations and the same correlations are also realisable in an acyclic causal structure.
\end{definition}

A HCL is by definition a causal loop since it corresponds to a directed cycle in the causal structure. These act as causal loops at the level of the causal mechanisms but cannot be detected at the operational level of affects relations (or correlations).
It can be the case that causal structures contain directed cycles without being an ACL, meaning that the affects relations of the associated causal model can also be obtained in an acyclic causal structure. This does not necessarily imply that both the affects relations and correlations can be generated in an acyclic causal structure, so ACLs and HCLs need not be complements of one another. Below we provide an example of a HCL.

\begin{example}[An operationally undetectable causal loop]
\label{example: funcloop}
Consider the causal structure of Figure~\ref{fig:funcloop} over the binary RVs $X$, $Y$ and $\Lambda$, where $X$ and $Y$ are observed nodes which are causes of each other (forming a causal loop) and $\Lambda$ is an unobserved common cause of the two. Suppose that the RVs are related as follows: $\Lambda$ is uniformly distributed and $X=\Lambda\oplus Y$ and $Y=\Lambda\oplus X$. Note that the given causal structure already implies a complete set of affects relations i.e., for each pair of the observed nodes, we know whether or not one affects the other. In this case, this is implied by the dashed arrows and we have that $X$ and $Y$ do not affect each other. Since $\Lambda$ is uniform, $P_{\cG_{\mathrm{do}(X)}}(Y|X)$ and $P_{\cG_{\mathrm{do}(Y)}}(X|Y)$ are both uniform, and in order to have the required (non-)affects relations, it must be that case that $P_{\cG}(X)$ and $P_{\cG}(Y)$ are both uniform. Along with the given functional dependences, this implies that $X$ and $Y$ are uncorrelated with each other. In other words, there are no affects relations or correlations between the set of observed nodes of this causal structure even though there is a causal loop. A causal structure over $X$ and $Y$ with no edges at all would also explain these observations. Therefore the directed cycle between $X$ and $Y$ in Figure~\ref{fig:funcloop} corresponds to a hidden causal loop. It is also worth noting that knowing the value of the exogenous variable $\Lambda$ is not enough to determine the value of $X$ or $Y$ with the given functional relations. Nevertheless, in Appendix~\ref{appendix: doMechanisms} we propose a method for uniquely determining the observed distribution in such examples, when the causal mechanisms are given.
\end{example}

 \begin{figure}[t]
         \centering
\includegraphics[]{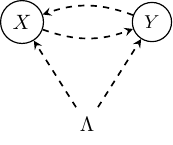}
	\caption{An operationally undetectable causal loop (Example~\ref{example: funcloop}).}
	\label{fig:funcloop}
     \end{figure}

We now focus on the more interesting class of causal loops, affects causal loops. Definition~\ref{def: ACL} only tells us that these are causal loops whose existence is operationally certified by the observable affects relations. It is natural to seek necessary and sufficient conditions on the set of affects relations such that they form an ACL. Here (and in Appendix~\ref{appendix: MoreLoops}), we propose several sufficient conditions which can be considered as definitions of different types of affects causal loops. We discuss 6 types here and 4 more in the appendix and provide examples to illustrate that none of these are necessary conditions i.e., there can be further types of ACLs not covered by these ten types. After defining the 6 types here, we will prove that these are indeed ACLs (Theorem~\ref{theorem: cycles}).

A first sufficient condition for the existence of an ACL is that there are two RVs $X$ and $Y$ that affect each other. Since affects implies cause, this tells us that $X$ and $Y$ must be causes of each other and hence that these affects relations are only realisable in a cyclic causal structure i.e., they lead to an ACL. A second condition is the presence of a chain of single RV affects relations from $X$ to $Y$ and from $Y$ to $X$. The latter can in general be a distinct condition from the former due to the non-transitivity of the affects relation (see Example~\ref{example:affects_transitive}), but can be shown to be an ACL. This gives us the following two types of causal loops. 
\begin{definition}[Affects causal loops, Type 1 (ACL1)]
\label{def: ACL1}
A set of affects relations $\mathscr{A}$ is said to contain a Type 1 affects causal loops if there exist two RVs $X$ and $Y$ such that $\{X$ affects $Y$, $Y$ affects $X\} \subseteq \mathscr{A}$.
\end{definition}

\begin{definition}[Affects causal loops, Type 2 (ACL2)]
\label{def: ACL2}
A set of affects relations $\mathscr{A}$ is said to contain a Type 2 affects causal loop if there exist RVs $X$, $Z_1$, $Z_2$, $\ldots$, $Z_k$ and $Y$ such that $X$ affects $Z_1$, $Z_1$ affects $Z_2$, $\ldots$, $Z_k$ affects $Y$ and $Y$ affects $X$ are all in $\mathscr{A}$.
\end{definition}

More generally, one can also consider affects relations involving sets of RVs. A first observation is that $S_1$ affects $S_2$ and $S_2$ affects $S_1$ for two sets of RVs does not imply the existence of a directed cycle in the causal structure.  For example, consider a causal structure $\cG$ with 4 nodes $A$, $B$, $C$ and $D$, all of which are observed such that the only edges in $\cG$ are the solid arrows $A\longrightarrow B$ and $C\longrightarrow D$, with $A$ affects $B$ and $C$ affects $D$. Then, if $S_1=AD$ and $S_2=BC$ we have $S_1$ affects $S_2$ and $S_2$ affects $S_1$ even though $\cG$ is clearly acyclic. However, if we take these to be irreducible affects relations, this will no longer be the case and we can certify the cyclicity of the causal structure from the affects relations, as we later show. This motivates more general set of sufficient conditions for the existence of affects causal loops. Two immediate possibilities are the following.

\begin{definition}[Affects causal loops, Type 3 (ACL3)]
\label{def: ACL3}
A set of affects relations $\mathscr{A}$ is said to contain a Type 3 affects causal loop if there exist two disjoint sets $S_1$ and $S_2$ of RVs such that $\{S_1$ affects $e_2$, $S_2$ affects $e_1\} \subseteq \mathscr{A}$ where $e_1\in S_1$, $e_2\in S_2$, and both affects relations are irreducible.
\end{definition}

\begin{definition}[Affects causal loops, Type 4 (ACL4)]
\label{def: ACL4}
A set of affects relations $\mathscr{A}$ is said to contain a Type 4 affects causal loop if there exist sets of RVs $S_1$, $S_2$, $\ldots$, $S_n$ where each pair $S_i$ and $S_{i+1\text{mod} n}$ is disjoint, such that $\{S_1$ affects $S_2$, $S_2$ affects $S_3$, $\ldots$, $S_{n-1}$ affects $S_n$, $S_n$ affects $S_1\} \subseteq \mathscr{A}$, and all these affects relations are irreducible.
\end{definition}

ACL1, ACL2, ACL3 and ACL4 imply cyclicity of the causal structure as shown in Theorem~\ref{theorem: cycles}. However, these are not the most general conditions on the affects relations with this property. There can be further conditions that are not equivalent to ACL1, ACL2, ACL3 or ACL4 which also imply cyclicity. The following is such a condition.

\begin{definition}[Affects causal loops, Type 5 (ACL5)]
\label{def: ACL5}
A set of affects relations $\mathscr{A}$ is said to contain a Type 5 affects causal loop if there exist sets of RVs $S_i\subseteq \hat{S}_i$ for $i=1,\ldots,n$ such that $\hat{S}_1$ affects $S_2$, $\hat{S}_2$ affects $S_3$, $\ldots$, $\hat{S}_{n-1}$ affects $S_n$ and $\hat{S}_n$ affects $S_1$ are all in $\mathscr{A}$, where all the affects relations are irreducible and every pair of sets connected by an affects relation is disjoint. Such a chain of affects relations is called a complete affects chain, in this case the affects chain is from the set $S_1$ to itself.
\end{definition}
Rather than considering a chain of irreducible affects relations from an RV or a set of RVs back into itself, one can consider multiple chains which taken together imply cyclicity and this would give yet another type of causal loop in our framework. For example we may have an irreducible affects relation $A$ affects $BC$. Along with another irreducible affects relation $BCD$ affects $A$, this would form a Type 5 affects causal loop. By Corollary~\ref{corollary:HOaffectsCause2}, these affects relations would tell us that $A$ is either a cause of $B$ or $C$ while $B$, $C$ and $D$ are all causes of $A$. Irrespective of whether $A$ is a cause of $B$ or of $C$, this implies the existence of a directed cycle in the causal structure. However, we could instead have started with the irreducible affects relations $A$ affects $BC$, $B$ affects $A$ and $C$ affects $A$. Since in general, $B$ affects $A$, and $C$ affects $A$ need not imply $BC$ affects $A$ (see Example~\ref{example: HOaffects3}), these affects relations may not constitute a Type 5 affects loop but they nevertheless imply cyclicity (using Corollary~\ref{corollary:HOaffectsCause2}). Note that $A$ affects $BC$ and $B$ affects $A$ alone (even if irreducible) do not necessarily imply cyclicity since the former tells us that $A$ is either a cause of $B$ or of $C$ and the latter that $B$ is a cause of $A$. That is, these affects relations can in principle be obtained in an acyclic causal model where $A$ is a cause of $C$ and $B$ is a cause of $A$. 
Generalising this idea, we have another type of affects loop, ACL6.

\begin{definition}[Affects causal loops, Type 6 (ACL6)]
\label{def: ACL6}
A set of affects relations $\mathscr{A}$ is said to contain a Type 6 affects causal loop if the following conditions are satisfied
\begin{enumerate}
    \item There exist disjoint sets of RVs $S_1$ and $S_2$ such that $S_1$ affects $S_2$ belongs to $\mathscr{A}$ and is irreducible.
    \item For each element $e_2\in S_2$, there exists a \emph{complete chain of irreducible affects relations} that connects it back to $S_1$, i.e., for each $e_2$, there exists sets of RVs $S_i\subseteq \hat{S}_i$ for $i=1,\ldots,n$ and $s_1\subseteq S_1$ such that $\{\hat{S}_2$ affects $S_3$, $\hat{S}_3$ affects $S_4$, $\ldots$, $\hat{S}_{n-1}$ affects $S_n$, $\hat{S}_n$ affects $s_1\} \subseteq \mathscr{A}$, where all the affects relations are irreducible and every pair of sets connected by an affects relation is disjoint.
\end{enumerate}
\end{definition}

There are further types of affects causal loops, all of which imply cyclicity of the causal structure. For example, we can also consider affects causal loops involving chains of conditional higher-order affects relations (Definition~\ref{definition:HOaffects}) and define analogues of ACL1-6 for this case. These can in general be distinct from ACL1-6 since it is possible to have a conditional HO affects relation $X$ affects $Y$ given $\{\mathrm{do}(Z),W\}$ without the unconditional zeroth-order affects relation $X$ affects $Y$. Even using unconditional zeroth-order affects relations alone, further distinct classes of affects causal loops are possible and four such classes (ACL7 to ACL10) are described in Appendix~\ref{appendix: MoreLoops}. The intuition behind them is as follows. The kind of chains of irreducible affects relations considered in the above definitions are such that for each subsequent pair of affects relations $\hat{S}_i$ affects $S_{i+1}$, the set $S_{i+1}$ is contained in $\hat{S}_{i+1}$. What if this were not the case, and we only had that $S_{i+1}\bigcap \hat{S}_{i+1}\neq \emptyset$? Let us call this an ``incomplete'' affects chain. The example before the last definition, with $A$ affects $BC$ and $B$ affects $A$ illustrates that this condition alone is not enough to guarantee cyclicity and to justify calling these affects relations a causal loop. One way is to add the affects relation $C$ affects $A$, which motivates the definition of ACL6 above. Another option is to add the irreducible affects relations $C$ affects $D$ and $D$ affects $BC$ and one can again show that the set of irreducible relations $\mathscr{A}=\{A$ affects $BC$, $B$ affects $A$, $C$ affects $D$, $D$ affects $BC\}$ is cyclic. One can however verify that this $\mathscr{A}$ does not correspond to any of the affects causal loops previously defined. There are two incomplete affects chains that complete each other, but no complete chain as required by the above types of ACL. In general, one might need to combine a given incomplete chain with several other complete or incomplete chains to guarantee cyclicity of the resulting set of affects relations, and the conditions therefore continue to get more complex. Even the additional classes of affects causal loops defined in Appendix~\ref{appendix: MoreLoops} do not exhaust all the possible types of affects causal loops that might be possible in our framework (we provide an example in the appendix to illustrate this).

The following theorem (proven in Appendix~\ref{appendix: proofs3}) shows that ACL1-6 are indeed affects causal loops in the sense of Definition~\ref{def: ACL}.
\begin{restatable}{theorem}{Cycles}
\label{theorem: cycles}
Any set of affects relations $\mathscr{A}$ containing an affects causal loop of Type 1, 2, 3, 4, 5 or 6 can only arise from a causal model over a cyclic causal structure i.e., these are indeed instances of affects causal loops according to Definition~\ref{def: ACL}.
\end{restatable}
     
\subsection{Possibility of compatibly embedding causal loops in space-time}  

In the previous section we discussed various properties of causal loops that follow from the causal model alone and without reference to space-time. Here we consider the space-time embeddings of such loops and whether affects causal loops can be compatibly and non-trivially embedded in a space-time structure. This turns out to indeed be possible for certain types of affects causal loops. This implies that for some causal loops their existence can be operationally certified (through observed affects relations, by virtue of being affects causal loops), and they can nevertheless be non-trivially embedded in space-time without leading to signalling outside the space-time future. While our framework can be applied to arbitrary partially ordered space-times, for the sake of illustration, we consider the case of (1+1)-dimensional Minkowski space-time in this section. Before we show the existence of embeddable causal loops in this case, we make the following observation.

\begin{restatable}{lemma}{NoAffectsLoops}
\label{lemma: NoAffectsLoops}
Let $S$ be a set of RVs and $\mathscr{A}$ be a set of affects relations over them.
\begin{enumerate}
    \item The absence of affects causal loops (Definition~\ref{def: ACL}) in $\mathscr{A}$ is a sufficient condition for the existence of a non-trivial embedding of $S$ in a space-time that $\mathscr{A}$ is compatible with. 
    \item If $\mathscr{A}$ is assumed to be a set of affects relations associated with a faithful causal model, then all causal loops are Type 1 affects causal loops and the existence of a non-trivial space-time embedding of $S$ that $\mathscr{A}$ is compatible with is both necessary and sufficient to rule out all causal loops and guarantee the acyclicity of the causal model that generates $\mathscr{A}$.
\end{enumerate}
\end{restatable}

The above lemma (proven in Appendix~\ref{appendix: proofs3}) shows that all the distinct classes of causal loops ACL2 to ACL6 (and ACL7 to ACL10 and other possible classes as described in Appendix~\ref{appendix: MoreLoops}) as well as the concept of hidden causal loops only arise in fine-tuned causal models. If fine-tuning is allowed, even the absence of affects causal loops does not rule out causal loops since we can have hidden causal loops which are operationally undetectable i.e., the absence of ACL does not imply acyclicity of the causal structure. Here, we first show that the absence of Type 1 and Type 2 affects causal loops is necessary for the existence of such a non-trivial and compatible space-time embedding. The results of the associated Letter~\cite{VilasiniColbeckPRL} show that this is no longer true for ACLs of higher types, in particular we construct an ACL of Type 4 there that does admit such a space-time embedding. This demonstrates that the absence of affects causal loops is not necessary for the existence of a non-trivial and compatible space-time embedding. 
We further show here that the absence of Type 1 and 2 loops is not sufficient for the existence of a non-trivial and compatible space-time embedding, 
since such an embedding is not guaranteed to exist for affects loops of other types i.e., for ACL3 and above there may or may not exist a non-trivial and compatible space-time embedding (this is discussed in Appendix~\ref{appendix: MoreLoops}).

Consider the affects causal loops of Types 1 and 2. Recall that a non-trivial space-time embedding is one where no two RVs such that one affects the other are assigned the exact same space-time location. A non-trivial space-time embedding is impossible for ACL1 and ACL2, since $\textbf{compat}$ applied to a set of affects relations containing an ACL2 implies that $\cX\preceq \cZ_1 \preceq \ldots \preceq \cZ_k \preceq Y \preceq \cX$ which can only be satisfied when $\cX\preceq \cY\preceq \cX$ i.e., $O(X)=O(Y)$, which corresponds to a trivial embedding. The latter step follows directly by applying $\textbf{compat}$ for ACL1. This is stated explicitly in the following Lemma. 

\begin{restatable}{lemma}{TrivEmbedding}
\label{lemma: TrivEmbedding}
Let $S$ be a set of RVs and $\mathscr{A}$ be a set of affects relations over them that contains affects causal loops of Types 1 or 2. The set $S$ cannot be non-trivially embedded in any space-time such that $\mathscr{A}$ is compatible with the embedding.
\end{restatable}

Now consider ACL3 formed by the irreducible affects relations $\mathscr{A}=\{AB$ affects $C$, $CD$ affects $A\}$. Applying \textbf{compat} to the first affects relation, we have that $C$ must be in the joint inclusive future of $A$ and $B$ i.e., $\cA\preceq \cC$ and $\cB \preceq C$. The condition \textbf{compat} for the second affects relation similarly implies that $\cC \preceq \cA$ and $\cD \preceq \cA$. Together these imply that $A$ and $C$ must be embedded at the same location while $B$ and $D$ cannot be in the future of this location. Since we neither have $A$ affects $C$ nor $C$ affects $A$ in $\mathscr{A}$, there is a non-trivial embedding. However, if we form $\mathscr{A}'$ by adding one or both of these affects relations to $\mathscr{A}$, there will no longer be any non-trivial and compatible embedding. In other words, affects causal loops of Type 3 can admit non-trivial and compatible space-time embeddings, but will always be \emph{degenerate}, i.e., they require two of the RVs to be embedded at the same location ($e_1$ and $e_2$ in Definition~\ref{def: ACL3}), as shown in the lemma below.

\begin{restatable}{lemma}{DegenEmbedding}
\label{lemma: DegenEmbedding}
Let $S$ be a set of RVs and $\mathscr{A}$ be a set of affects relations over them that contains affects causal loops of Type 3. The set $S$ cannot be embedded in any space-time such that the embedding is non-degenerate such that $\mathscr{A}$ is compatible with the embedding. However, there are non-trivial embeddings that $\mathscr{A}$ is compatible with.
\end{restatable}
\begin{proof}
By Definition~\ref{def: ACL3}, ACL3 implies that for two sets $S_1$ and $S_2$ of RVs, we have the irreducible affects relations $S_1$ affects $e_2$ and $S_2$ affects $e_1$ for some elements $e_1\in S_1$ and $e_2\in S_2$. Applying $\textbf{compat}$ (Definition~\ref{definition: compatposet}), this implies that $e_1$ must be embedded in the inclusive future of all elements $e_2'\in S_2$ and $e_2$ must be embedded in the inclusive future of all elements $e_1'\in S_1$. This is only possible if $e_1$ and $e_2$ are embedded at the same location, making the embedding degenerate. However, it can be the case that $\mathscr{A}$ does not contain or imply the any affects relations between $e_1$ and $e_2$, therefore the embedding may still be non-trivial.
\end{proof}

Can we embed an affects causal loop compatibly in space-time such that all RVs have distinct locations? The associated Letter~\cite{VilasiniColbeckPRL} shows that such a non-degenerate embedding is indeed possible for certain types of affects causal loops, with an explicit example. The causal loop proposed in ~\cite{VilasiniColbeckPRL} corresponds to a Type 4 ACL in the language of the present paper, we reproduce this example here completeness.

\begin{example}[An operationally detectable causal loop with a non-trivial, compatible space-time embedding \texorpdfstring{\cite{VilasiniColbeckPRL}}{PRL}]
\label{example: main}
Suppose we have the irreducible affects relations $\mathscr{A}=\{B$ affects $AC$, $AC$ affects $B\}$ which form a Type 4 ACL. Then $\textbf{compat}$ would require that $\overline{\mathcal{F}}(\cB)=\overline{\mathcal{F}}(\cA)\bigcap \overline{\mathcal{F}}(\cC)$. This can be satisfied even when $A$, $B$ and $C$ are embedded at distinct space-time locations, as shown in Figure~\ref{fig: egloop}. This figure shows that this affects causal loop involving the RVs $A$, $B$ and $C$ can be embedded in a (1+1)-dimensional Minkowski space-time without leading to signalling outside the space-time future. This is possible even if we embed $A$ and $C$ arbitrarily far in the past, as long as the earliest location where their lightcones intersect coincides with the location of $B$. By Theorem~\ref{theorem: cycles}, observation of the affects relations $\{B$ affects $AC$, $AC$ affects $B\}$ operationally certifies the existence of a causal loop i.e., that there exist at least one pair of RVs among $A$, $B$ and $C$ that are causes of each other. This causal loop corresponds to a closed timelike curve (CTC) once the RVs are embedded in a space-time, since it would imply bidirectional causal influences between two distinct space-time locations. Even if this CTC involves causal influences between RVs that occur far apart in time (in some reference frame), they will not allow any agent to signal superluminally since the affects relations are compatible with the space-time. This is true even if the agent can access all the RVs or any subset thereof. This is because both of the affects relations in $\mathscr{A}$ can only be verified in the joint future of $A$ and $C$, and the earliest point that they can do so is the location of $B$.
\end{example}

\begin{figure}[t!]
 \centering
\includegraphics[]{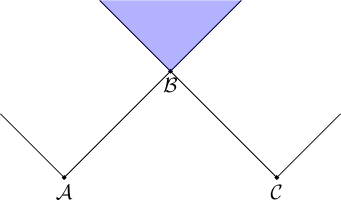}  
    \caption{\textbf{A non-trivial and compatible space-time embedding for an operationally detectable causal loop}  Example~\ref{example: main} describes a set of affects relations that forms an affects causal loop of Type 4. Such a causal loop is operationally detectable since the cyclicity of the underlying causal model can be certified operationally using the observed affects relations, as shown in Theorem~\ref{theorem: cycles}. 
    This figure illustrates a non-trivial and non-degenerate, yet compatible embedding of this causal loop in (1+1)-dimensional Minkowski space-time, where space and time are given along the horizontal and vertical axes respectively and black lines correspond to light cones. 
    Note that this embedding remains compatible even when the space-time RVs $\cA$ and $\cC$ are pushed to the far past of $\cB$ along the black line ($\cB$'s past light-like surface).}
     \label{fig: egloop}
\end{figure}
An explicit cyclic causal model in which the affects causal loop of the above example can arise is also provided in~\cite{VilasiniColbeckPRL}, we further discuss this in Appendix~\ref{appendix: examples} and Figure~\ref{fig: eqcyclemain} along with other examples that illustrate the concept of ``higher order affects relations'' introduced in this work. One can also use the framework developed here to construct several further examples of causal loops (of ACL4 or higher types) that can be compatibly embedded in  space-time. The example provided in our Letter~\cite{VilasiniColbeckPRL} suffices to illustrate the claim that such loops are even possible. We discuss further, the space-time embeddings of higher types of ACLs in Appendix~\ref{appendix: MoreLoops}.

\begin{remark}[Types of space-time embeddings]
\label{remark: EmbeddingTypes}
Apart from distinguishing between different types of causal loops (that arise due to fine-tuning of the underlying parameters of the causal model), one might also wish to distinguish between different types of space-time embeddings. Some useful distinctions that were made so far are between trivial and non-tivial embeddings and degenerate and non-degenerate embeddings. The former is useful because any set of affects relations can be compatibly embedded in a space-time through a trivial embedding where all RVs are embedded at the same location, and this does not tell us anything interesting. If we demand non-trivial embeddings, i.e., that two RVs connected by an affects relation only involving them are not embedded at the same location, then this rules out affects causal loops of Types~1 and~2, as shown in Lemma~\ref{lemma: TrivEmbedding} but not Type~3 loops. On the other hand, if we demand non-degenerate embeddings, i.e., that all RVs are embedded at distinct locations, we can rule out Type~3 affects causal loops as shown in Lemma~\ref{lemma: DegenEmbedding}, but not Type~4. Note that the compatible and non-degenerate space-time embedding of the Type~4 affects causal loop that we propose in the Letter~\cite{VilasiniColbeckPRL} (and discussed in Example~\ref{example: main}) is ``fine-tuned'' and is \emph{unstable} in the sense that small adjustments to the space-time embedding of the variables would break compatibility. In the case of Minkowski space-time, the requirement $\overline{\mathcal{F}}(\cB)=\overline{\mathcal{F}}(\cA)\bigcap \overline{\mathcal{F}}(\cC)$ that guarantees compatibility of the ACL4 in Example~\ref{example: main} confines the ORVs $\cA$ and $\cC$ to a surface that is one dimension smaller than the dimensions of the space-time, once the location of the ORV $\cB$ is fixed. This surface if simply the boundary of the past light cone of $\cB$. $\cA$ and $\cC$ can be placed anywhere on this surface, including arbitrarily far in the past of $\cB$ along its past light-like surface but cannot be placed out of this surface without violating compatibility. [Alternatively, once $\cA$ and $\cC$ are embedded, there is only one possible location for $\cB$.] Such examples of causal loops that do not lead to superluminal signalling involve a form of fine-tuning both at the level of the causal model and at the level of its space-time embedding.
\end{remark}

\begin{remark}[Open questions and challenges]
\label{remark: challenges}
As motivated in the above remark, one can consider further distinctions between space-time embeddings, such as whether they are unstable embeddings. We have seen that such embeddings arise in Example~\ref{example: main} and other examples of this section and Appendix~\ref{appendix: MoreLoops}. All the non-degenerate and compatible space-time embeddings of affects causal loops that we know so far (such as Example~\ref{example: main}) are such unstable embeddings. Therefore an interesting open question is whether demanding that an embedding is \emph{stable} would rule out some or all of the affects causal loops of Types~4 and higher.

It remains unclear what condition on the space-time embedding would rule out all possible types of affects causal loops. A main reason is that the general class of affects causal loops (i.e., operationally detectable causal loops) is not fully characterised, ACL1--ACL6 only provide various sufficient conditions that imply the existence of an affects causal loop but none of them, including the further classes ACL7--ACL10 discussed in Appendix~\ref{appendix: MoreLoops} are necessary. In all the classes other than ACL1 and ACL2, one can find causal loops that admit non-trivial and compatible space-time embeddings, but it is also possible to find ACLs of other types that have no non-trivial or compatible embeddings. Thus the question regarding necessary and sufficient conditions on affects relations that guarantee a non-trivial and compatible embedding (and similarly for other types of embeddings) also remains open.

While we have seen that there is a non-trivial and compatible embedding of the affects loop of Example~\ref{example: main} in (1+1)-dimensional Minkowski space-time, there is no such embedding of the same loop in (3+1)-dimensional Minkowski space-time~\cite{VilasiniColbeckPRL}. This is because the compatibility condition requires $B$ to be embedded at the earliest location in the joint future of $A$ and $C$ which is not possible in (3+1)-dimensional Minkowski space-time where a frame-dependent concept of earliest location in the joint future does not exist (in contrast to the (1+1)-dimensional case). This implies that the conditions for ruling out causal loops in a space-time can depend on the dimension of the space-time, and possibly other topological features. In particular, it remains a pertinent open question whether the existence of a non-trivial and compatible embedding in the space-time is sufficient to rule out all affects causal loops in (3+1)-dimensional Minkowski space-time. We leave these open questions as a challenge for future research in the field.
\end{remark}

The framework developed here, along with the results of the associated Letter~\cite{VilasiniColbeckPRL} illustrate the counter-intuitive possibilities offered by fine-tuning---if it is possible to have superluminal causal influences without superluminal signalling (as in non-local hidden variable theories~\cite{Bohm1952} or the jamming scenario~\cite{Grunhaus1996, Horodecki2019}), then we can also have causal loops that do not lead to superluminal signalling. The particularly interesting feature of such causal loops is that they can be operationally detected through their affects relations. These results have consequences for the claims of~\cite{Grunhaus1996, Horodecki2019} that certain constraints on correlations in Bell scenarios are necessary and sufficient for ruling out all types of causal loops. They suggest that neither directions of these claims can hold. This is discussed in the following section.

\section{Critical analysis of previous claims regarding relativistic causality}
\label{ssec: missing_assump}

Here we comment on two related works,~\cite{Grunhaus1996} where the concept of jamming non-local correlations were introduced and~\cite{Horodecki2019} where these were further analysed and generalised to so-called ``relativistic causal correlations''. The results and assumptions of~\cite{Grunhaus1996, Horodecki2019} are not stated in the same mathematical language as ours, and hence some translation is needed to use our framework. For this discussion, we consider a tripartite Bell experiment, i.e., we consider six random variables: the settings ($A$, $B$, $C$), and corresponding outcomes ($X$, $Y$, $Z$) that are embedded into Minkowski space-time satisfying the following constraints.

\begin{definition}[Embeddings of the form $\mathscr{E}^{\text{jam}}$]
Random variables $A$, $B$, $C$, $X$, $Y$ and $Z$ have an embedding of the form $\mathscr{E}^{\text{jam}}$ if the following conditions are satisfied
\begin{enumerate}
    \item $\{\cA,\cX\}\nprec\nsucc\{\cB,\cY\}$, $\{\cA,\cX\}\nprec\nsucc\{\cC,\cZ\}$ $\{\cB,\cY\}\nprec\nsucc\{\cC,\cZ\}$, $\cA\preceq \cX$, $\cB\preceq \cY$, $\cC\preceq \cZ$,\footnote{The conditions $\cA\prec\cX$, $\cB\prec\cY$, $\cC\prec\cZ$ are more natural than the last three relations, but, as in~\cite{Horodecki2019}, we allow for the possibility of instantaneous measurements here.}
    \item $\overline{\mathcal{F}}(\mathcal{X})\bigcap\overline{\mathcal{F}}(\mathcal{Z})\subseteq \overline{\mathcal{F}}(\mathcal{B})$
\end{enumerate}
\end{definition}
The first of these conditions corresponds to space-time constraints for a tripartite Bell scenario.  $A$, $B$ and $C$ can be thought of as settings with $X$, $Y$ and $Z$ corresponding outcomes. The first condition then represents the space-like separation of the three parts of the experiment, with each setting embedded in the space-time past of the corresponding outcome. The second condition is an additional restriction on the space-time location of the RVs related to the particular jamming scenario we wish to consider, demanding that the joint future of $X$ and $Z$ is in the future of $B$. These conditions define a family of embeddings; the locations of $A$, $B$, $C$, $X$, $Y$ and $Z$ in Figure~\ref{fig: trins} satisfy the conditions. Tripartite Bell experiments carried out in space-like separated configurations satisfying the first conditions of $\mathscr{E}^{\text{jam}}$ are normally associated with a set of no-signalling constraints on the possible correlations. However, given both conditions of $\mathscr{E}^{\text{jam}}$, the works~\cite{Grunhaus1996, Horodecki2019} consider a relaxed set of no-signalling conditions, as follows.

\begin{definition}[Relaxed tripartite no-signalling conditions~\cite{Grunhaus1996} (\textbf{NS3$'$})] The relaxed no-signalling conditions \textbf{NS3$'$} associated with an embedding of the form $\mathscr{E}^{\text{jam}}$ correspond to the following constraints on the observed distribution $P_{XYZ|ABC}$. 
\begin{align}
\label{eq: tripartiteNS}
    \begin{split}
         P_{XY|AB}(x,y|a,b)&:=\sum_zP_{XYZ|ABC}(x,y,z|a,b,c)=\sum_z P_{XYZ|ABC}(x,y,z|a,b,c')\quad \forall x,y,a,b,c,c'\\
          P_{YZ|BC}(y,z|b,c)&:=\sum_xP_{XYZ|ABC}(x,y,z|a,b,c)=\sum_x P_{XYZ|ABC}(x,y,z|a',b,c)\quad \forall y,z,a,a',b,c\\
         P_{X|A}(x|a)&:=\sum_{y,z} P_{XYZ|ABC}(x,y,z|a,b,c)=\sum_{y,z} P_{XYZ|ABC}(x,y,z|a,b',c')\quad \forall x,a,b,b',c,c'\\
         P_{Z|C}(z|c)&:=\sum_{x,y} P_{XYZ|ABC}(x,y,z|a,b,c)=\sum_{x,y} P_{XYZ|ABC}(x,y,z|a',b',c)\quad \forall z,a,a',b,b',c\\
    \end{split}
\end{align}
\end{definition}
Note that these conditions imply $P_{Y|ABC}(y|abc)$ is independent of $a$ and $c$, so that $P_{Y|B}$ is well defined. The idea behind these relaxed conditions is that they allow $P_{XZ|ABC}$ to depend on $B$ (which would normally be forbidden) on the grounds that the joint future of $X$ and $Z$ is contained in that of $B$ in the embedding $\mathscr{E}^{\text{jam}}$, and hence information about $B$ can remain in its future (as explained in Section~\ref{ssec:intro_jamming}).

Since the conditions \textbf{NS3$'$} involve only the observed correlations, they do not by themselves tell us about causation. Therefore, without making further assumptions about the underlying causal model, they cannot be necessary and sufficient conditions to rule out superluminal signalling or causal loops. For instance, a set of correlations violating \textbf{NS3$'$} could arise from a single unobserved common cause of all six variables without any direct causes, which would not lead to any superluminal signalling. When referring to such conditions on correlations as ``no-signalling'' conditions, we often implicitly assume some notion of ``free choice'' for the settings (see~\cite{CR_ext,CR2013,Horodecki2019} for definitions of free choice). In the causal modelling framework, free choice can be modelled by taking the settings $A$, $B$ and $C$ to be exogenous i.e., as having no prior causes. Given the exogeneity of $A$, $B$ and $C$, \textbf{NS3$'$}  capture the signalling possibilities through interventions on these variables (cf.\ Corollary~\ref{corollary:exogenous}), such as $C$ does not affect $XY$ given $AB$, etc. Thus, in the language of the present paper, the result of~\cite{Grunhaus1996} can be stated as saying that given an embedding of the form $\mathscr{E}^{\text{jam}}$, and with $A$, $B$ and $C$ exogenous, the conditions \textbf{NS3$'$} are sufficient to prevent superluminal signalling by interventions on $A$, $B$ and $C$.

A stronger claim is made by~\cite{Horodecki2019}, that the conditions \textbf{NS3$'$} are necessary and sufficient for ensuring no superluminal signalling and no causal loops with such an embedding and they termed the correlations satisfying \textbf{NS3$'$} ``relativistic causal correlations''. Within the framework introduced in this paper, if interventions are also allowed on $X$, $Y$ and $Z$, the sufficiency of \textbf{NS3$'$} for ruling out superluminal signalling in the space-time embedding $\mathscr{E}^{\text{jam}}$ does not hold. 
The reason is that there are causal models satisfying \textbf{NS3$'$} as well as having $X$ affects $Y$. The latter involves intervention on a non-exogenous node $X$, and implies that Alice can signal to Bob. With an embedding of the form $\mathscr{E}^{\text{jam}}$, this is superluminal. This is also captured by our definition of compatibility (Definition~\ref{definition: compatposet}), according to which the relation $X$ affects $Y$ is not compatible with an embedding of the form $\mathscr{E}^{\text{jam}}$. 

The claim of~\cite{Horodecki2019} is not just about the impossibility of superluminal signalling but also about ruling out causal loops. As we have seen, correlations satisfying (\textbf{NS3$'$}) that allow for jamming must be fine-tuned regardless of the causal structure. In fine-tuned causal models, several distinct classes of causal loops are possible, some of which are operationally undetectable (or hidden --- cf.\ Definition~\ref{def: HCL}) while others are operationally detectable but nevertheless do not lead to superluminal signalling as we show in an associated Letter~\cite{VilasiniColbeckPRL} (see also Example~\ref{example: main}). The former, by virtue of being operationally undetectable can never be ruled out only from the correlations (\textbf{NS3$'$} or otherwise) or affects relations. Operationally detectable causal loops require the consideration of non-trivial affects relations between sets of RVs, which are not detectable from the correlations alone. For instance \textbf{NS3$'$} cannot detect the existence of causal loops between outcome variables (e.g., $X$ is a cause of $Y$ and $Y$ is a cause of $X$) since when a common cause is included, this common cause can explain the correlations.
Therefore, the claim of~\cite{Horodecki2019} that \textbf{NS3$'$} rules out causal loops does not hold within our framework, even when restricting to the case where the settings are exogenous. More generally, in the absence of alternative frameworks for formalising these questions, it remains unclear how conditions on correlations such as \textbf{NS3$'$} could be necessary and sufficient for ruling out causal loops (as claimed in \cite{Horodecki2019}) without further assumptions. 

Furthermore, although we can rule out certain types of operationally detectable causal loops by demanding certain properties of the space-time embedding e.g., affects causal loops of Types~1--3 (cf.\ Lemmas~\ref{lemma: TrivEmbedding} and~\ref{lemma: DegenEmbedding}), the absence of operationally detectable causal loops is not necessary to ensure no superluminal signalling (cf.\ Example~\ref{example: main}). While \textbf{NS3$'$} is necessary to prevent superluminal signalling within the embedding $\mathscr{E}^{\text{jam}}$, it is not necessary to rule out affects causal loops in this embedding: it is possible to have an acyclic causal model over the settings $A$, $B$ and $C$, and outcomes $X$, $Y$ and $Z$ that violates \textbf{NS3$'$} and leads to superluminal signalling in the embedding $\mathscr{E}^{\text{jam}}$. The implication that superluminal signalling implies causal loops holds within the theory of special relativity. Here we do not want to assume it, which allows us to consider more general relations between these principles and our framework can hence be used also in theories with a preferred frame for instance. 
A more detailed analysis of previous works such as~\cite{Horodecki2019} and~\cite{Grunhaus1996} and the possibilities of superluminal signalling/causal loops in the jamming scenario are carried out in upcoming work~\cite{Jamming2}.

\section{Summary and conclusions}
\label{sec: conclusions}

We have developed a general mathematical framework for studying causality by clearly separating between operational and space-time related notions and characterising their compatibility. This has foundational relevance for understanding causality in quantum and more general theories, as well as practical applications for cryptography, information processing tasks in space-time and causal discovery. We have mainly focused on two notions of causality: the operational notion of causality defined through an extension of the causal modelling approach~\cite{Pearl2009, Henson2014} and relativistic causality which is associated with a space-time structure.

We formulated the operational notion of causality under minimal assumptions while allowing for causal influences to be fine-tuned, cyclic and mediated by latent non-classical systems. On the other hand, relativistic causality can be understood as the condition that ``causal influences can propagate only from past to future in the space-time'', where it has several implications such as ``it is impossible to signal outside the future'', ``it is possible to signal everywhere in the future and nowhere else'', ``in Minkowski space-time it is impossible to have causal loops'', and ``it is impossible to broadcast classical information outside the future''. Often one or more of these implications are taken in isolation to represent the condition of relativistic causality. Within the theory of special relativity these are related (e.g., the possibility of superluminal signalling leads to causal loops), but without assuming relativity they may not be and hence need to be independently formalised.

 Within our framework we have formalised several of the above concepts and shown that these are distinct conditions in general. Our compatibility condition (Definition~\ref{definition: compatposet}) ensures that a causal model does not lead to signalling outside the future when embedded in a space-time structure. An alternative compatibility condition discussed in Section~\ref{sec: necsuff} captures the idea of broadcasting classical variables within the space-time future.  Cyclic causal models involve causal loops and when embedded in space-time, as described in Section~\ref{ssec: embedding}, these allow for causal influences going in both directions between two distinct space-time points. Thus, the embedded cyclic causal structure can be understood as a closed time-like curve (CTC). Applying this framework, we have shown in an associated Letter~\cite{VilasiniColbeckPRL} that it is mathematically possible to have such CTCs in Minkowski space-time, and that their existence can be operationally detected without leading to superluminal signalling. This establishes that no superluminal signalling and no causal loops/closed timelike curves are in general different conditions. In the present paper, we have gone beyond this particular example and identified several different classes of operationally detectable causal loops (or affects causal loops) in our framework and characterised properties of their space-time embeddings. Should one such operational detection be made (which we do not expect) it would certify the physical existence of retro-causation. Such constructions are possible because our framework does not require all causal influences to respect the partial order of the space-time but only that signalling possibilities are constrained by the space-time. 

In particular, this work also serves as the first causal modelling framework for a class of post-quantum theories (``jamming theories'') previously proposed in the literature~\cite{Grunhaus1996, Horodecki2019} which are known to be more general than standard GPTs. Previously, these theories have been analysed focusing on the correlations they generate, but a causal modelling framework enables us to systematically study the effect of active interventions on arbitrary physical systems in such theories which provides more information about the underlying causal structure than correlations alone. Using this, we analysed previous claims regarding the compatibility of such theories with principles such as no superluminal signalling and no causal loops, which suggests that these claims cannot hold without further assumptions. In future work~\cite{Jamming2} we apply the framework developed here to such post-quantum jamming scenarios to characterise the signalling possibilities and new properties of theories admitting such correlations.

To allow us to deal with fine-tuned causal structures, we introduced higher-order affects relations.  Our results show these to be a useful tool for inferring causation in the presence of fine-tuning that also has operational meaning in terms of signalling through joint interventions on multiple systems. When a particular phenomenon has two possible causal explanations, one of which is fine-tuned, the fine-tuned explanation is often considered undesirable because it usually more complicated and involves features that cannot be operationally verified. Fine-tuning complicates causal reasoning and the majority of the literature on causal models typically assumes the absence of fine-tuning. Explanations of quantum correlations in terms of classical causal models are typically rejected as such explanations involve fine-tuning~\cite{Wood2015}, and instead faithful explanations in terms of quantum causal models are often preferred. On the other hand, fine-tuning occurs in many cryptographic scenarios, as well as jamming correlations (cf.\ Proposition~\ref{prop:finetune}, Example~\ref{example: jamming}). Using a causal modelling approach allows for a clear distinction to be made between undesirable and potentially useful forms of fine-tuning. The former correspond to causal influences that can never be operationally detected while the latter can be operationally detected by considering more general, joint interventions on sets of random variables.
In this work, we have shown that several causal modelling concepts that are equivalent in the absence of fine-tuning, become distinct concepts in the presence of fine-tuning. We presented several technical results relating these various concepts in general fine-tuned and cyclic causal models with latent non-classical causes, which can have useful applications for the causal discovery problem in the presence of fine-tuning. 


In our framework, we have modelled space-time as a discrete partially ordered set on which we embedded a separate operational causal model, and considered compatibility between the two notions. There are two different ways in which this space-time can be interpreted. One is to regard it as a fundamental background on which physics given by the causal model is embedded, such that every observed node in the model can be associated with a location in the space-time. Then our results tell us that the absence of signalling outside the future when the causal model is embedded in the space-time, does not allow us to identify the space-time order relation with the causal order, and that these are distinct concepts. A second interpretation is to understand the space-time order as an emergent property of the physics given by the causal model. For instance, our compatibility condition could be interpreted as a way to infer which space-time orders could occur alongside the operational predictions of the model, if we consider the direction of signalling in the model to constrain the order relations of the space-time. Our works (the present paper and~\cite{VilasiniColbeckPRL}) show that even in cyclic causal models, it can be possible to single out a preferred direction (namely the direction of signalling) from the operational predictions of the model, while at the same time certifying that the underlying causal model is cyclic. 

The present work focuses on the signalling possibilities allowed by the causal model, rather than the strength of signalling or correlations, even though the framework developed here can in principle model both. In~\cite{Kempf2021}, the strength of correlations was considered as a way to capture properties such as distance that are associated with an underlying space-time, with the hope that space-time can be seen as emergent.
In approaches to quantum gravity, such as causal set theory~\cite{Bombelli1987, Surya2019, Dukovski2013}, an active line of research is to derive geometric properties of a continuum space-time from order-theoretic properties of discrete graphs that capture the causal relations of the space-time. The present work, along with a related follow-up work~\cite{VilasiniRenner2022} suggest a possible direction of inquiry for connecting the research on non-classical causal models with such approaches to quantum gravity, and we leave these interesting directions for future work.

To summarise, our results highlight the importance of separating a) operational and space-time related notions of causality b) correlation, causation and signalling (by considering interventions) and c) distinct notions of causality within the operational/space-time categories mentioned in a). 

\section{Open questions}
\label{sec: discussion}
The work presented here provides a platform for analysing a number of problems in quantum foundations and causality in a new light. We discussed specific interesting and open questions related to the characterisation of causal loops within our framework in Remark~\ref{remark: challenges}. Here we  place our work within a broader context and discuss the associated open questions.
\bigskip

\paragraph{Other notions of causality:} 
While this work elucidates the relationships between a number of different notions of causality, there are many more that may be considered. For instance, another operational notion of causality is that of process terminality~\cite{Coecke2017} which says that discarding all the outputs of a causal process is equivalent to discarding the process. Further, approaches such as the process matrix framework~\cite{Oreshkov2012} aim to formulate causality more generally in the absence of a fixed background space-time (which we have assumed here). Other setup assumptions in these approaches mean that, for instance, post-quantum jamming scenarios cannot be modelled.\footnote{The process framework assumes a tensor product structure between the local operations of various parties, and once communication between parties is forbidden, the framework can only produce correlations compatible with standard no-signalling theories and not the relaxed no-signalling conditions of~\cite{Horodecki2019} that permit jamming.} Here several conditions such as causal orderedness, causal separability, satisfaction of causal inequalities have been proposed, which serve as causality criteria under different assumptions. 
The precise relationships between all these notions of causality, their operational meaning  and implications for the physics of information processing remains open. In a related work involving one of the authors~\cite{VilasiniRenner2022}, the present approach of disentangling operational and space-time notions of causation and characterising their compatibility, is applied to operational scenarios described by the process matrix framework~\cite{Oreshkov2012}. There, further connections between indefinite and cyclic causation are established in quantum scenarios and a more general class of space-time embeddings is considered that allows for spacetime embeddings of quantum systems where the systems are not nonlocalised to a single space-time location but may possibly be delocalised over a space-time region.
These results (along with previous works such as~\cite{Barrett2020}) relating indefinite causation to definite cyclic causation indicate that cyclic and non-classical causal models can have applications also to scenarios where a background space-time structure is not assumed.  
Although we do not consider it here, our framework can also be used to analyse frame-dependent notions of causality associated with Minkowski space-time (e.g., whether compatibility and other properties of an embedding can depend on the choice of classical reference frame).\footnote{For example, one can consider a different partially ordered set to represent space-time structure from the perspective of different frames, such that classical frame transformations such as Lorentz transformations could be viewed as invertible maps between these partially ordered sets.}
Another intriguing prospect for future research would be to consider compatibility between operational causal models and space-time related information in more exotic regimes where a global space-time structure may not exist but agents infer space-time information using their local (possibly quantum) reference frames~\cite{Zych2019, Giacomini2019, Castro_Ruiz_2020}.

\bigskip
\paragraph{Affects relations and d-separation: } We use affects relations (Definition~\ref{definition: affects}), based on the notion of interventions, to distinguish between correlation and causation. In acyclic causal structures~\cite{Pearl2009, Henson2014} and in classical cyclic causal structures~\cite{Forre2017}, existing frameworks prescribe how the post-intervention distribution can be calculated from the observed distribution and/or the underlying causal mechanisms. In non-classical cyclic causal structures, such a characterisation is not available. In Section~\ref{sec: causmod}, we used the d-separation condition (Definition~\ref{definition: compatdist}) on the observed distribution to obtain a partial characterisation which suffices for the current purpose, but this does not fully specify the post-intervention distribution. In Appendix~\ref{appendix: doMechanisms}, we outline a possible method for obtaining the post-intervention distribution given the underlying causal mechanisms.  [Although this method may not always recover the d-separation condition~\ref{definition: compatdist}, this does not impact the results of this paper.] Generalising our framework to also include non-classical cyclic causal models that do not obey the d-separation condition, by using the causal mechanisms as primitives would allow our results regarding space-time compatibility and affects loops to be applied to this larger class of models. This would provide a general framework for causally modelling fine-tuned and cyclic non-classical causal models such that any post-intervention scenario can be completely specified by the original causal model.\footnote{In the current framework, the causal model is defined in terms of the observed distribution and therefore not all affects relations can be deduced from the model's specification. When characterised instead in terms of the causal mechanisms, the affects relations should become deducible.} Another observation made in Appendix~\ref{appendix: doMechanisms} is that the presence of causal loops could allow us to distinguish between faithful, non-classical explanations and unfaithful classical explanations (e.g., using non-local hidden variables) of quantum correlations, which cannot be operationally distinguished otherwise. This suggests that it might be possible to operationally distinguish hidden variable interpretations of quantum theory such as Bohmian mechanics from inherently ``quantum'' interpretations, in the presence of causal loops. Formalising this observation would be another interesting line of investigation.
\bigskip
\paragraph{Causal loops and paradoxes: } In this work we have considered space-time to be modelled by a partial order.  The theory of general relativity allows for the possibility of more exotic space-time structures. These possibilities led to investigations of closed time-like curves and there are mathematical models of CTCs that are logically consistent and do not lead to time travel paradoxes~\cite{Deutsch1991,Bennett2005, Svetlichny2011, Lloyd2011a, Lloyd2011}. Two inequivalent models have been developed to make sense of information flow in the presence of CTCs, Deutsch's CTCs (DCTCs)~\cite{Deutsch1991} and post-selected CTCs (PCTCs)~\cite{Bennett2005, Svetlichny2011, Lloyd2011a, Lloyd2011}. DCTCs and PCTCs are known to have different amounts of computational power~\cite{Aaronson2008, Lloyd2011} and to provide different resolutions to the grandfather and unsolved theorem paradoxes~\cite{Lloyd2011}. In our framework, grandfather-type paradoxes are forbidden by the assumption that a valid joint probability distribution observed variables exists, which implies that the underlying causal mechanisms (e.g., functional equations in the classical case) must be mutually consistent.\footnote{An example of a paradoxical scenario is a 2-cycle between binary variables $X$ and $Y$ where the influence $X\longrightarrow Y$ defines the functional dependence $Y=X$ and $Y\longrightarrow X$ gives the dependence $X=Y\oplus1$. These equations are mutually inconsistent and there is no joint distribution $P_{XY}$ compatible with these dependences.} 

The unproved theorem paradox on the other hand is not ruled out in the current framework, and can depend on how the framework is further instantiated with causal mechanisms. For example, in classical cyclic causal models, an assumption regarding the unique solvability of the underlying functional dependences is often considered. In particular, this could be seen as the requirement that any information involved in a loop (such as the unproved theorem) must be fully and uniquely determined by the mechanisms of the causal model thereby eliminating the paradox of a proof that ``came from nowhere''. An analogous condition on the causal model for ruling out the unproved theorem paradox in the quantum case is far from clear, since the causal mechanisms in this case are not deterministic functional equations. These questions can be explored within a full formalisation of our framework in terms of causal mechanisms, along the lines discussed in the previous paragraph (and Appendix~\ref{appendix: doMechanisms}). It is interesting to consider whether there are connections between the CTCs that can be embedded in Minkowski space-time without superluminal signalling (such as Example~\ref{example: main}) and DCTCs or PCTCs, or which physical principles rule out such CTCs.

\bigskip
\paragraph{Causality in time-symmetric formulations of quantum theory: } Unitary quantum mechanics is time symmetric while operational quantum theory has the possibility of irreversible measurements.  There are several proposals for modelling quantum and more general experiments in a time symmetric manner while still making operational predictions and retrodictions about measurements~\cite{Aharonov1964,Oeckl2019,Di_Biagio2021, hardy2021,chiribella2021}. Predictions and retrodictions indicate the direction of inference, not necessarily of causation and the role of causality (in terms of a causal modelling paradigm) is not fully understood in these frameworks. A notable approach for making operational statements in a time-symmetric setting is the two-time state formalism~\cite{Aharonov1964,Aharonov1991,Silva2014}, which describes measurements on pre- and post-selected quantum states, where the former can be thought of as evolving forward-in-time and the latter, backward-in-time. It is interesting to consider how this time-symmetric approach can be modelled in a causal framework.

\bigskip
\paragraph{Causal discovery in the presence of fine-tuning: } Causal discovery is the problem of inferring a fully or partially unknown causal structure from observed correlations, possibly combined with additional information about interventions. Fine-tuning makes this task more challenging and causal discovery algorithms typically assume that the underlying causal model is not fine-tuned~\cite{Pearl2009,Spirtes2001}, even in cases where the underlying model is classical and has no unobserved nodes.  Relaxations of this assumption have been considered where certain forms of fine-tuning (but not all) have been allowed~\cite{Zhalama2017}. Intuitively, use of higher-order affects relations appears useful for causal discovery in the presence of fine-tuning, and the examples of Section~\ref{ssec: HOaffects} show the usefulness of HO affects to distinguish between causal structures with the same correlations. We believe this deserves future exploration.

\bigskip

\paragraph{Indefinite space-time locations: } In the present work, we have embedded causal models in a space-time structure by assigning a single space-time location to each observed system. More generally, we can have, both in theory and practice, systems whose space-time locations have some classical or quantum uncertainty or protocols involving quantum systems that are delocalised over space and in time~\cite{Chiribella2013, Procopio2015, Rubino2017, Portmann2017, VilasiniRenner2022}.  
It would therefore be of interest to generalise our methods to allow for such superpositions. In a related work~\cite{VilasiniRenner2022}, a method to do this for finite dimensional quantum systems in a discrete space-time (i.e., a partially ordered set as considered here) is proposed, which has applications for physically characterising so-called indefinite causal order processes~\cite{Oreshkov2012}. 
\begin{acknowledgements}
A preliminary version of this work and upcoming work~\cite{Jamming2} is available in VV's PhD thesis~\cite{VVThesis} (Chapter 6). VV thanks Maarten Grothus for useful feedback on the framework. We also thank Mirjam Weilenmann and Lorenzo Maccone for insightful discussions. VV acknowledges support from the PhD scholarship of the Department of Mathematics, University of York and the ETH Postdoctoral Fellowship from ETH Z\"{u}rich. 
\end{acknowledgements}

\appendix

\section{Identifying conditional independences and affects relations: Examples} 
\label{appendix: examples}

Here we provide examples that better illustrate some of the definitions and rules of the framework. In particular, how one can deduce the conditional independences and affects relations in a given causal model. For this the following lemmas will be useful, these can be regarded as generalisations of Corollary~\ref{corollary:dsep-affects} and Lemma~\ref{lemma: correl-affects} from the unconditional zeroth-order case to the case of general conditional higher-order affects relations, and are proven in Appendix~\ref{appendix: proofs4}.
\begin{restatable}{lemma}{HOappendixA}
\label{lemma:HOappendix1}
For any pairwise disjoint subsets $X$, $Y$, $Z$ and $W$ of the observed nodes $S$ of a causal model, we have
\begin{enumerate}
    \item $(XZW\perp^d Y)_{\cG_{\mathrm{do}(X Z)}}$ $\Rightarrow$ $XZ$  does not affect $Y$ given $W$.
    \item $(XZW\perp^d Y)_{\cG_{\mathrm{do}(X Z)}}$ $\Rightarrow$ $X$ does not affect $Y$ given $\{\mathrm{do}(Z),W\}$.
\end{enumerate}
\end{restatable}

\begin{restatable}{lemma}{HOappendixB}
\label{lemma:HOappendix2}
For any pairwise disjoint subsets $X$, $Y$, $Z$ and $W$ of the observed nodes $S$ of a causal model, we have
\begin{enumerate}
    \item $(XZW\nindep Y)_{\cG_{\mathrm{do}(X Z)}}$ $\Rightarrow$ $XZ$ affects $Y$ given $W$.
    \item $(XZW\nindep Y)_{\cG_{\mathrm{do}(X Z)}}$ \emph{and} $(ZW\perp^d Y)_{\cG_{\mathrm{do}(Z)}}$ $\Rightarrow$ $X$ \emph{affects} $Y$ given $\{\mathrm{do}(Z),W\}$.
\end{enumerate}
\end{restatable}

We now summarise how one may use these results to deduce some of the conditional independences and affects relations from a given causal model.
\begin{itemize}
    \item \emph{Conditional independences:} Given a causal graph $\cG$ with the set of observed nodes $S$, some of the conditional independences satisfied by the joint distribution $P_S$ can be identified using Definition~\ref{definition: compatdist} i.e., by listing all the conditional independences implied by d-separation relations in $\cG$. Further independences may be found if there are dashed arrows emanating from exogenous nodes, since $X\xdashrightarrow{} Y$ implies $X$ does not affect $Y$ (by Definition~\ref{definition: solidasharrows}) which implies $X\indep Y$ if $X$ is exogenous (cf.~ Corollary~\ref{corollary:exogenous}). Lemma~\ref{lemma: CI} can also be used to list further independences not directly implied by d-separation in $\cG$. There may still be more conditional independences in $P_S$ that cannot be listed using the methods mentioned above. Since we allow for fine-tuning and latent systems, there could in principle be arbitrarily many independences in $P$, but those mentioned above are sufficient for compatibility with the causal model. 
    
    \item \emph{Affects relations:} The existence of an affects relation $X$ affects $Y$ given $\{\mathrm{do}(Z),W\}$ can be deduced by applying Lemma~\ref{lemma:HOappendix2} (for the zeroth-order case, $X$ affects $Y$ is deduced by applying Lemma~\ref{lemma: correl-affects}). The non-affects relation $X$ does not affect $Y$ given $\{\mathrm{do}(Z),W\}$ can be deduced by applying Lemma~\ref{lemma:HOappendix1} and also the contrapositives of Lemmas~\ref{lemma:HOaffectsCause},~\ref{lemma:HOaffects1} and~\ref{lemma:HOaffects2}. For example, $(X\perp^dYW)_{\cG_{\mathrm{do}(X)}}$ implies that $X$ does not affect $Y$ given $\{\mathrm{do}(Z),W\}$ by the second statement of Lemma~\ref{lemma:HOaffectsCause}. In the zeroth-order case, the non-affects relation $X$ does not affect $Y$ is deduced by applying Corollary~\ref{corollary:dsep-affects}. The direction of the lemmas is important to note here, for instance the converse of Lemma~\ref{lemma: correl-affects} cannot be used to deduce that $X$ does not affect $Y$ since this implication does not hold, unless $X$ is exogenous (cf.\ non-implication 2 of Figure~\ref{fig: relations} and Corollary~\ref{corollary:exogenous}). Therefore one can check for non-independences and d-separations in the post-intervention causal model to identify affects and non-affects relations respectively. One may be able derive further results of this sort or exploit structural aspects of particular causal models to derive additional independences and affects relations. Again, due to fine-tuning and latent systems, in general, this identification may not be exhaustive even after this is done. However, in the case of causal models with no latent nodes (which are by definition, classical), it would indeed be exhaustive, as explained in Remark~\ref{remark:HOaffects}.

\end{itemize}
In case some or all of the causal mechanisms are also given in addition to the observed distributions, it may be possible to identify further independences and affects relations in the model. We now apply these rules to specific examples where it is possible to deduce all the conditional independences and affects relations involved, these are tabulated in Figure~\ref{fig:table_examples} for 3 out of 4 of the examples considered here, all of which correspond to fine-tuned causal models. The fourth example corresponds to a faithful but cyclic causal model, and therefore the d-separation condition~\ref{definition: compatdist} and zeroth-order affects relations (given by the causal arrows themselves) completely characterise the scenario.

\begin{figure}[t]
\centering\includegraphics[scale=0.85]{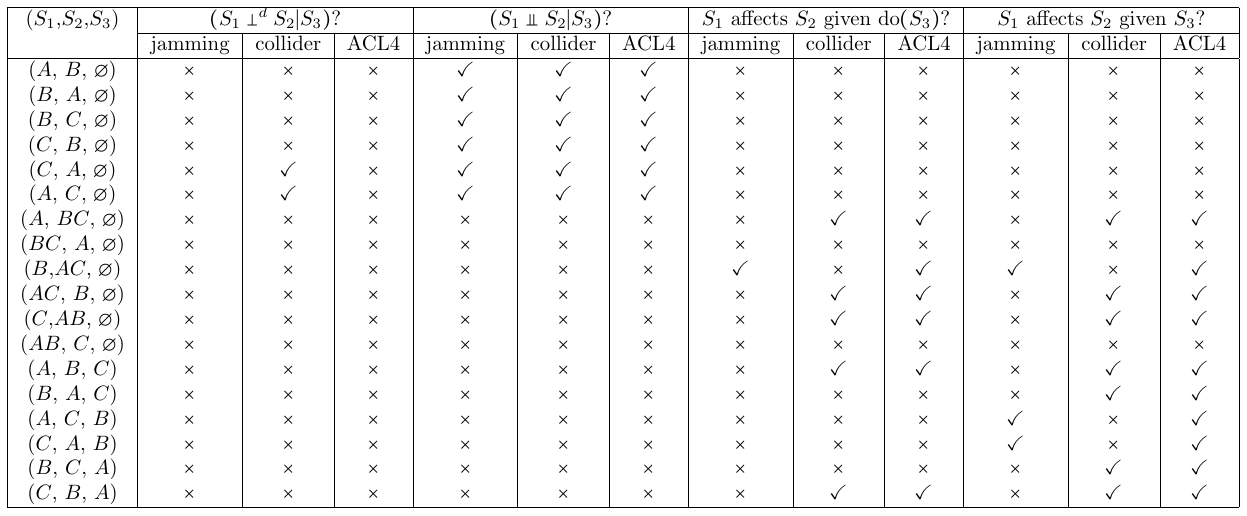}
    \caption{Table of all possible d-separations, conditional independences and affects relations for jamming, fine-tuned collider and Type 4 affects causal loop examples (Sections~\ref{sssec: egjamming},~\ref{sssec: egcollider},~\ref{sssec: egcycle}), all of which involve the three observed RVs $A$, $B$ and $C$. All affects relations, when they do exist, are irreducible.
    }
    \label{fig:table_examples}
\end{figure}

\subsection{Jamming (Figure~\ref{fig: jamming})}
\label{sssec: egjamming}
In the jamming causal structure $\mathcal{G}^{\mathrm{jam}}$ of Figure~\ref{fig: jamming} and Example~\ref{example: jamming}, Definition~\ref{definition: compatdist} does not impose any conditional independences on the observed distribution $P_{ABC}$ since $\Lambda$ is unobserved.\footnote {If $\Lambda$ in Figure~\ref{fig: jamming} were observed, $A$ and $C$ would be d-separated given $\{B,\Lambda\}$ and we would have the conditional independence $P_{AC|B\Lambda}=P_{A|B\Lambda}P_{C|B\Lambda}$.} However, from Definition~\ref{definition: solidasharrows} of dashed arrows we know that $B$ affects neither $A$ nor $C$ individually and we are given that $B$ affects $AC$. Using the exogeneity of $B$ (cf.~ Corollary~\ref{corollary:exogenous}), this implies the independences $A\indep B$ and $C \indep B$  and the non-independence $B\not \indep AC$ in $\cG^{\mathrm{jam}}$. Now, consider an intervention on $A$. The post-intervention causal structure $\mathcal{G}^{\mathrm{jam}}_{\mathrm{do}(A)}$ only has the edges $B\xdashrightarrow{} C$ and $\Lambda \longrsquigarrow C$ (along with $I_A\longrightarrow A$ of course). The d-separation $(A\perp^d C)_{\mathcal{G}^{\mathrm{jam}}_{\mathrm{do}(A)}}$ implies the independence $(A\indep C)_{\mathcal{G}^{\mathrm{jam}}_{\mathrm{do}(A)}}$ (Definition~\ref{definition: compatdist}) and also that $A$ does not affect $C$ (Corollary~\ref{corollary:dsep-affects}). Similarly, we can derive the relations $C$ does not affect $A$, $A$ does not affect $BC$,  $C$ does not affect $AB$ and $AC$ does not affect $B$. Combined with the lack of any zeroth-order affects relations between any two of the RVs, this implies that there are no affects relations of order $1$ or higher (by the contrapositive of Lemma~\ref{lemma:HOaffects1}). Further, using Lemma~\ref{lemma: CI} and the exogeneity of $B$, we can derive $AB$ does not affect $C$ as follows. In the causal structure $\cG^{\mathrm{jam}}_{\mathrm{do}(AB)}$, $A$ is d-separated from $B$ and $C$, while $B$ and $C$ are independent of each other due to the exogeneity of $B$ and the dashed arrow connecting them. Using the lemma, this gives $(AB\indep C)_{\cG^{\mathrm{jam}}_{\mathrm{do}(AB)}}$ which can be explicitly written as $P_{\cG^{\mathrm{jam}}_{\mathrm{do}(AB)}}(C|A,B)=P_{\cG^{\mathrm{jam}}_{\mathrm{do}(AB)}}(C)$. The right hand side can be simplified in the following two steps. Firstly as $P_{\cG^{\mathrm{jam}}_{\mathrm{do}(AB)}}(C)=P_{\cG^{\mathrm{jam}}_{\mathrm{do}(A)}}(C)$ noting that $\cG^{\mathrm{jam}}_{\mathrm{do}(AB)}$ and $\cG^{\mathrm{jam}}_{\mathrm{do}(A)}$ are effectively the same graph due to the exogeneity of $B$. Then the d-separation $(A\perp^d C)_{\cG^{\mathrm{jam}}_{\mathrm{do}(A)}}$ implies the independence $P_{\cG^{\mathrm{jam}}_{\mathrm{do}(A)}}(C|A)=P_{\cG^{\mathrm{jam}}_{\mathrm{do}(A)}}(C)$, which along with $A$ does not affect $C$ (as noted earlier) gives $P_{\cG^{\mathrm{jam}}_{\mathrm{do}(A)}}(C)=P_\cG^{\mathrm{jam}}(C)$. Putting this together, we get $P_{\cG^{\mathrm{jam}}_{\mathrm{do}(AB)}}(C|A,B)=P_\cG^{\mathrm{jam}}(C)$ i.e., $\{A.B\}$ does not affect $C$. Similarly, one can obtain $BC$ does not affect $A$. All these are summarised in Figure~\ref{fig:table_examples}.

\subsection{Fine-tuned collider (Figure~\ref{fig: collider})}
\label{sssec: egcollider}
In the causal structure of Figure~\ref{fig: collider}, the independence $A\indep C$ follows from the d-separation condition (Definition~\ref{definition: compatdist}), while $A\indep B$ and $C\indep B$ follow from the dashed arrow structure. These are the same independences as the case in the previous example of jamming with unobserved $\Lambda$ (where $A\indep C$ was an additional independence in the jamming example but follows from d-separation in this case). Thus the distribution $P_{ABC}$ from Example~\ref{example: jamming} is compatible with both the jamming (Figure~\ref{fig: jamming}) as well the fine-tuned collider (Figure~\ref{fig: collider}) causal structures.\footnote{Note that this is essentially the one-time pad example from earlier.} However interventions on the two causal structures yield different results. We have $AC$ affects $B$ for the fine-tuned collider (since $AC$ consists of exogenous nodes and is correlated with $B$) but not for the jamming case. We also have $A$ affects $BC$ and $C$ affects $AB$ for the fine-tuned collider even though $A$ and $C$ do not individually affect $B$ due to the dashed arrow structure. This follows from the exogeneity of $A$ and $C$ and the joint correlations $A=B\oplus C$. Further, $AB$ does not affect $C$ since these sets become d-separated upon intervention on $AB$ and by a similar reasoning, $BC$ does not affect $A$ and $B$ does not affect $AC$ (in contrast with the jamming case where $B$ affects $AC$). As for higher-order affects relations, $A$ affects $B$ given do$(C)$ and $C$ affects $B$ given do$(A)$ are the only ones, and these follow by applying Lemma~
\ref{lemma:HOappendix2} to this example. Again, these conclusions are summarised in Figure~\ref{fig:table_examples}. 
\begin{figure}[t]
 \centering\subfloat[\label{fig: collider}]{\includegraphics[]{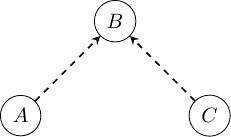}}\qquad\qquad
 \subfloat[\label{fig: eqcyclemain}]{\includegraphics[]{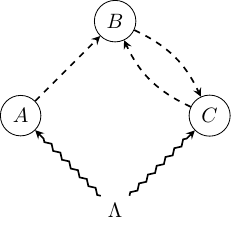}}\qquad\qquad
  \subfloat[\label{fig: causalloop}]{\includegraphics[]{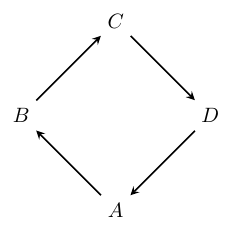}}
    \caption{\textbf{Some fine-tuned and/or cyclic causal structures: } \textbf{(a)} A fine-tuned collider \textbf{(b)} A Type 4 affects causal loop \textbf{(c)} A Type 1 affects causal loop}
    \label{fig: examples}
\end{figure}

\subsection{A Type 4 affects causal loop (\texorpdfstring{\cite{VilasiniColbeckPRL}}{PRL},  Figure~\ref{fig: eqcyclemain})}
\label{sssec: egcycle}
Example~\ref{example: main} outlined a cyclic causal model proposed in~\cite{VilasiniColbeckPRL} for demonstrating the mathematical possibility of compatibly embedding affects causal loops in Minkowski space-time. Here, we reproduce further details of this model and discuss its properties it in the context of the more general framework developed in the present paper.
Consider the cyclic causal structure $\cG^{\mathrm{ACL}4}$ of Figure~\ref{fig: eqcyclemain} along with the following classical causal mechanisms where all 4 variables are taken to be binary: $A=\Lambda$, $B=A\oplus C$, $C=B\oplus \Lambda$, where the exogenous variable $\Lambda$ is uniformly distributed. One can check that the distribution $P_{ABC}$ obtained through these mechanisms would be the same as that of the jamming as well as the fine-tuned collider examples above, but the affects relations differ from those of these examples and instead correspond to those of Example~\ref{example: main} which is an affects causal loop of Type 4. To obtain these affects relations, first note that in the causal model of $\cG^{\mathrm{ACL}4}_{\mathrm{do}(A)}$, $\Lambda$ is no longer a parent of $A$, but using the remaining causal mechanisms $B=A\oplus C$ and $C=B\oplus \Lambda$ (which remain the same), we can still obtain $A=\Lambda$. Therefore the intervention on $A$ does not change the observed distribution and $A$ and $B$ continue to be independent in $\cG^{\mathrm{ACL}4}$ as well as $\cG^{\mathrm{ACL}4}_{\mathrm{do}(A)}$, and in both graphs the marginal distributions over $A$, $B$ and $C$ are uniform, which gives $A$ does not affect $B$. On the other hand, $B$ does not affect $A$ can be established simply from the d-separation $(B\perp^d A)_{\cG^{\mathrm{ACL}4}_{\mathrm{do}(B)}}$. In the causal model of $\cG^{\mathrm{ACL}4}_{\mathrm{do}(C)}$, neither $B$ nor $\Lambda$ are parents of $C$ but the remaining mechanisms $A=\Lambda$ and $B=A\oplus C$ give $C=B\oplus \Lambda$. Again, the observed distribution here is the same as the pre-intervention distribution, which gives $C$ does not affect $B$. By a similar argument, $B$ does not affect $C$ can also be established. Further, we have both $B$ affects $AC$ (as in the jamming case) \emph{and} $AC$ affects $B$ (as in the fine-tuned collider) since $P_{\cG^{\mathrm{ACL}4}_{\mathrm{do}}(A C)}(B|A,C)$ and $P_{\cG^{\mathrm{ACL}4}_{\mathrm{do}}(B)}(A,C|B)$ are deterministic while $P_{\cG^{\mathrm{ACL}4}}(B)$ and $P_{\cG^{\mathrm{ACL}4}}(A,C)$ are uniform. We also have $A$ affects $BC$ and $C$ affects $AB$ as in the fine-tuned collider, which can be verified using the causal mechanisms given.\footnote{Note that in the absence of the causal mechanisms, many of the affects/non-affects relations may not be identifiable. For example, to deduce that $AB$ does not affect $C$ in the jamming case, we used Lemma~\ref{lemma: CI} along with the fact that $B$ was exogenous in $\cG^{\mathrm{jam}}$. However, the same argument cannot be applied here since $B$ is not exogenous in $\cG^{\mathrm{ACL}4}$.} As in the jamming case, we also get $AB$ does not affect $C$ and $BC$ does not affect $A$. The higher-order affects relations here are identical to the previous example, and obtained in a similar manner and these are given in Figure~\ref{fig:table_examples}. Furthermore, even though this corresponds to an affects causal loop the existence of which can be operationally certified and this causal model admits a non-trivial and compatible embedding in Minkowski space-time as explained i.e., it does not lead to superluminal signalling~\cite{VilasiniColbeckPRL} (see also Example~\ref{example: main}).

\subsection{A Type 1 affects causal loop (Figure~\ref{fig: causalloop})}
\label{sssec: egloops}
Consider a faithful causal model associated with the cyclic causal structure of Figure~\ref{fig: causalloop}. Here all 4 nodes are observed and hence classical. The faithfulness implies that the only conditional independences are those implied by d-separation. Since $B\perp^d D|AC$ and $A\perp^d C|BD$ are the only d-separations, $B\indep D|AC$ and $A\indep C|BD$ are the only conditional independences. Further, the faithfulness also implies that every causal arrow is associated with an affects relation, which is also reflected in the fact that all arrows are solid arrows, and we have $A$ affects $B$, $B$ affects $C$, $C$ affects $D$ and $D$ affects $A$, which forms an affects causal loop of Type 1 (there is an affects relation is both direction between every pair of RVs). One can easily check that the only irreducible affects relations are the zeroth-order affects relations between single RVs, as one would expect for faithful causal models. To further illustrate the kind of causal loops allowed in this framework, consider the pairwise correlations $A=B$, $B=C$, $C=D$ and $D\neq A$. Since this system of equations has no solutions, there exists no joint distribution $P_{ABCD}$ from which the pairwise marginals producing these correlations can be obtained. Such examples correspond to grandfather type paradoxes and cannot be modelled in frameworks that demand the existence of a valid joint probability distribution over all variables involved in a causal loop. On the other hand, examples of solid arrow directed cycles where the functional dependences of the loop variables admit solutions such as $A=B=C=D$ (with any probability) or the examples considered in~\cite{Pearl2013} for other cyclic causal structures can be modelled in our framework. Additionally, there can also be Type 1 and Type 2 affects causal loops that do not involve any solid arrows, for example through a concatenation of structures such as that of Figure~\ref{fig:example1}. We discuss causal loops in more detail in Appendix~\ref{appendix: doMechanisms}, also in the case of quantum causal structures.

\section{Further classes of affects causal loops and their space-time embeddings}
\label{appendix: MoreLoops}

As motivated in the main text (see the paragraph after Definition~\ref{def: ACL6}), we can consider further classes of affects causal loops that are distinct from ACL1, $\ldots$, ACL6. The intuition is that the chain of irreducible affects relations involved in these previous definitions is such that for any two adjacent affects relations in the chain the second set of the first is contained in the first set of the second. Relaxing this containment condition can lead to a violation of Theorem~\ref{theorem: cycles}, as also explained in the main text. So we can consider relaxing this condition and only requiring a non-trivial intersection between the sets (which would make the chain ``incomplete''), as long as we include additional conditions on the affects relations that will again guarantee cyclicity of the causal structure. Here we propose four more classes of affects causal loops ACL7, ACL8, ACL9 and ACL10 based on this idea, illustrate them with examples which also show that there can in general be more classes of affects causal loops even beyond these.

\begin{definition}[Affects causal loops, Type 7 (ACL7)]
\label{def: ACL7}
A set of affects relations $\mathscr{A}$ is said to contain a Type 7 affects causal loop if the following conditions are satisfied
\begin{enumerate}
    \item There exist disjoint sets of RVs $S_1$ and $S_2$ such that $S_1$ affects $S_2$ belongs to $\mathscr{A}$ and is irreducible.
    \item There exists a chain of irreducible affects relations (possibly incomplete) $\mathscr{C}_{s_2}$ from some subset $s_2\subseteq S_2$ to $S_1$ i.e., there exists sets of RVs $s_2\subseteq S'_2$,  $S_3$, $S'_3$, $\ldots$ $S_n$, $S'_n$, $s_1\subseteq S_1$ such that $\{S'_2$ affects $S_3$, $S'_3$ affects $S_4$, $\ldots$ ,$S'_{n-1}$ affects $S_n$, $S'_n$ affects $s_1\} \subseteq \mathscr{A}$, where all the affects relations are irreducible, every pair of sets connected by an affects relation is disjoint, $S_i\bigcap S'_i\neq \emptyset$ for all $i\in \{3,\ldots,n\}$ and $S_2\bigcap S_2'=s_2$. Each pair $(S_i,S_i')$ such that $S_i\not\subseteq S'_i$ for $i\in \{2,\ldots,n\}$ is called an incomplete node of the affects chain $\mathscr{C}_{s_2}$, a complete affects chain has no incomplete nodes.
    \item For each affects chain $\mathscr{C}_{s_2}$ that connects the subset $s_2$ of $S_2$ back to $S_1$ as above, and each incomplete node $(S_i,S_i')$ in $\mathscr{C}_{s_2}$  (for $i\in \{2,\ldots,n\}$), there exists a complete affects chain $\mathscr{D}^{\mathscr{C}}_{s_2}$ in $\mathscr{A}$ from $S_i\backslash (S_i\bigcap S_i')$ to $S_i$.
\end{enumerate}
\end{definition}

\begin{definition}[Affects causal loops, Type 8 (ACL8)]
\label{def: ACL8}
A set of affects relations $\mathscr{A}$ is said to contain a Type 8 affects causal loop if the following conditions are satisfied
\begin{enumerate}
    \item There exist disjoint sets of RVs $S_1$ and $S_2$ such that $S_1$ affects $S_2$ belongs to $\mathscr{A}$ and is irreducible.
    \item For each element $e_2\in S_2$, there exists a chain of irreducible affects relations (possibly incomplete) $\mathscr{C}_{e_2}$ that connects it back to $S_1$ i.e., for each $e_2$, there exists sets of RVs $e_2\in S'_2$,  $S_3$, $S'_3$, $\ldots$ $S_n$ $S'_n$, $s_1\subseteq S_1$ such that $\{S'_2$ affects $S_3$, $S'_3$ affects $S_4$, $\ldots$ ,$S'_{n-1}$ affects $S_n$, $S'_n$ affects$s_1\} \subseteq \mathscr{A}$, where all the affects relations are irreducible, every pair of sets connected by an affects relation is disjoint, $S_i\bigcap S'_i\neq \emptyset$ for all $i\in \{3,\ldots,n\}$ and $S_2\bigcap S_2'=e_2$. 
    \item For each element $e_2\in S_2$, each affects chain $\mathscr{C}_{e_2}$ that connects it back to $S_1$ as above, and each incomplete node $(S_i,S_i')$ in $\mathscr{C}_{e_2}$ (for $i\in \{2,\ldots,n\}$), there exists a complete affects chain $\mathscr{D}^{\mathscr{C}}_{e_2}$ in $\mathscr{A}$ from $S_i\backslash (S_i\bigcap S_i')$ to $S_i$.
\end{enumerate}
\end{definition}

The following theorem (proven in Appendix~\ref{appendix: proofs5}) generalises Theorem~\ref{theorem: cycles} to ACL7 and ACL8, and justifies categorising them as affects causal loops.

\begin{restatable}{theorem}{CyclesAppendix}
\label{lemma: CyclesAppendix}
Any set of affects relations $\mathscr{A}$ containing an affects causal loops of Type 7 or Type 8 can only arise from a causal model over a cyclic causal structure.
\end{restatable}
More generally, for a given affects chain $\mathscr{C}_{s_2}$ in Definition~\ref{def: ACL7}, and an incomplete node $(S_i,S_i')$ in $\mathscr{C}_{s_2}$, instead of a single complete affects chain $\mathscr{D}^{\mathscr{C}}_{s_2}$ we could consider a set of incomplete affects chains that serve the same purpose and for which an analogous theorem holds. For example, for each incomplete node $(S_i,S_i')$ of $\mathscr{C}_{s_2}$, there can exist an incomplete affects chain $\mathscr{D}^{\mathscr{C}}_{s_2}$ in $\mathscr{A}$ from $S_i\backslash (S_i\bigcap S_i')$ to $S_i$, such that for each incomplete node $(R_j,R_j')$ of $\mathscr{D}^{\mathscr{C}}_{s_2}$, there exists another complete affects chain in $\mathscr{A}$ from $R_j\backslash (R_j\bigcap R_j')$ to $R_j$. This could go on recursively for arbitrarily many chains depending on the number of RVs appearing in $\mathscr{A}$. This recursive definition defines yet another class ACL9, and an analogous recursive version of ACL8 would define another class ACL10. Theorem~\ref{lemma: CyclesAppendix} for ACL9 and ACL10 follows through similar arguments, so we note this point without proof. We illustrate these new classes with some examples, along with an example to show that these (ACL1-ACL10) do not cover all possible affects causal loops.

\begin{example}[A Type 7 affects causal loop]
\label{example: ACL7}
Consider the set of irreducible affects relations $\mathscr{A}=\{X$ affects $Y$, $Y$ affects $AB$, $A$ affects $X$, $C$ affects $AB$, $B$ affects $C\}$. One can check that  $\mathscr{A}$ does not contain affects causal loops of Types 1 to 6, since no affects relation in  $\mathscr{A}$ is such that every element of the second set has a complete affects chain leading it back to the first set. It however contains at least one Type 7 affects loop. For the affects relation $Y$ affects $AB$, we have the incomplete chain $\mathscr{C}_A=\{A$ affects $X$, $X$ affects $Y\}$ that connects $A$ to $Y$ with the incomplete node $(S_2=AB,S_2'=A)$. Then $S_2\backslash (S_2\bigcap S_2')=\{B\}$ and we have the complete affects chain $\mathscr{D}^{\mathscr{C}}_A=\{B$ affects $C$, $C$ affects $AB\}$ that connects $S_2\backslash (S_2\bigcap S_2')$ to $S_2$ as required.
\end{example}

\begin{example}[A Type 9 affects causal loop]
\label{example: ACL9}
Consider the set of irreducible affects relations $\mathscr{A}=\{X$ affects $Y$, $Y$ affects $AB$, $A$ affects $X$, $C$ affects $AB$, $B$ affects $CD$, $D$ affects $E$, $E$ affects $CD\}$. This set is similar to the previous example, but does not contain a Type 7 causal loop (or ACLs of any lower types). It does contain a Type 9. For the affects relation $Y$ affects $AB$, there is an incomplete chain $\mathscr{C}_A=\{A$ affects $X$, $X$ affects $Y\}$ that connects $A$ to $Y$ as before. However, we have no complete chains from $S_2\backslash (S_2\bigcap S_2')=\{B\}$ to $S_2=AB$ as before, only the incomplete chain $\mathscr{D}^{\mathscr{C}}_A=\{B$ affects $CD$, $C$ affects $AB\}$. The incomplete node $(R_j,R_j')$ of $\mathscr{D}^{\mathscr{C}}_A$ is $(R_j=CD, R'_j=C)$ and we have a complete affects chain $\{D$ affects $E$, $E$ affects $CD\}$ from $R_j\backslash (R_j\bigcap R_j')=\{D\}$ to $R_j$.
\end{example}

\begin{example}[An affects causal loop not covered by Types 1 to 10]
\label{example: ACL11}
Consider the set of irreducible affects relations $\mathscr{A}=\{X$ affects$Y$, $Y$ affects $AB$, $A$ affects $X$, $C$ affects $AB$, $B$ affects $CD$, $BD$ affects $AC\}$. With some effort, one can see that $\mathscr{A}$ does not contain affects causal loops of Types 1--10. It nevertheless implies cyclicity, as follows. Applying Corollary~\ref{corollary:HOaffectsCause2}, we have that $X$ is a cause of $Y$, $Y$ is either a cause of $A$ or of $B$ and $A$ is a cause of $X$. If $Y$ is a cause of $A$, we already have a directed cycle, so consider the case where $Y$ is a cause of $B$. Using the remaining affects relations, we have $C$ is a cause of either $A$ or $B$, $B$ is a cause of either $C$ or $D$. Irrespective of whether $C$ is a cause of either $A$ or $B$, if $B$ is a cause of $C$, we would have a directed cycle, so we must take $B$ to be a cause of $D$ to avoid this. The last affects relation implies that $B$ is either a cause of $A$ or of $C$. Irrespective of the choice here and the choice of whether $C$ is a cause of $A$ or of $B$, we can verify that there will always be a directed cycle. Hence this set of affects relations is an affects causal loop that is not of a previously defined Type.
\end{example}

Consider now, the space-time embedding for the affects relations of Example~\ref{example: ACL7} in Minkowski space-time. Imposing $\textbf{compat}$ (Definition~\ref{definition: compatposet}) on the affects relations $\mathscr{A}=\{X$ affects $Y$, $Y$ affects $AB$, $A$ affects $X$, $C$ affects $AB$, $B$ affects $C\}$ implies that $\cY$ must contain the joint inclusive future of $\cA$ and $\cB$ but $\cA$ is in the past of $\cX$ which is in the past of $\cY$. The only way this can be satisfied is if $\cB$ is in the future of $\cA$ such that the joint inclusive future of $\cA$ and $\cB$ coincides with the inclusive future of $\cB$. The last two affects relations then imply that $\cB$ and $\cC$ must be embedded at the same location and since we have $B$ affects $C$, this embedding is trivial. This implies that there is no non-trivial and compatible embedding of these affects relations in Minkowski space-time. In other words, the absence of affects causal loops of Types 1-6 does not guarantee the existence of a non-trivial and compatible space-time embedding. The presence of Type 3 and above ACLs  does not rule out the existence of such an embedding as we have seen in Example~\ref{example: main}, in contrast to the case of Type 1 and 2 ACLs (Lemma~\ref{lemma: TrivEmbedding}). This suggests that for each individual Type $i$ of affects causal loops other than Types 1 and 2, the existence of a non-trivial and compatible space-time embedding is neither necessary nor sufficient for there to be no affects causal loops of that Type. By Lemma~\ref{lemma: DegenEmbedding}, for Type 3, the existence of a non-degenerate and compatible space-time embedding is sufficient but not necessary to rule out ACL3.

\section{Do-conditionals from causal mechanisms in quantum cyclic causal models}
\label{appendix: doMechanisms}

In Section~\ref{sec: causmod} we outlined how interventions and do-conditionals (i.e., the post intervention distribution) are defined in our framework, and Theorem~\ref{theorem:dorules} provides some conditions under which the post and pre intervention distributions can be related. Ideally though, one would expect that it should be possible to fully specify the post-intervention distribution if we are given all the underlying causal mechanisms involved in the causal structure. For example, in the classical case, the structural equations of the causal model~\cite{Pearl2009} provide these causal mechanisms. Here for each node $X$ in the causal structure, the dependence of $X$ on its parents par$(X)$ corresponds to a stochastic map, which can be written in terms of a deterministic function $X=f_X(\text{par}(X),E_X)$ by including an additional exogenous random variable $E_X$ for each node $X$. This is called a \emph{structural equation}. If the structural equations for all the nodes and the distributions of the parentless nodes are known, then the complete post-intervention distribution can be calculated. This has been shown to be the case for classical cyclic causal models in~\cite{Forre2017}. An intervention $\mathrm{do}(x)$ on $X$, corresponds to updating the structural equation for $X$ to $X=x$ while keeping the remaining structural equations the same. Another important result for the classical case derived in~\cite{Forre2017} is that the d-separation property or the \emph{global directed Markov condition} of Definition~\ref{definition: compatdist} is recovered whenever all the random variables are discrete and the structural equations of the causal model satisfy a property known as \emph{ancestrally unique solvability} (auSEP). Roughly, this property demands that the structural equation for each node must admit a unique solution given the values of the node's ancestors. We need not define this concept formally for our purposes here.

Ideally we would like to extend these ideas to quantum and post-quantum cyclic causal structures, where the causal mechanisms involve measurements and transformations on non-classical systems, which cannot be expressed using deterministic structural equations. In the non-classical case, it is unclear what conditions allow for the d-separation condition to be recovered. Even to make this question precise in the non-classical case, one would need to specify the analog of structural equations for such causal models which is an open problem. Here we present a possible method for defining general cyclic causal models from given (possibly non-classical) causal mechanisms and for calculating observed and interventional distributions in the model, without assuming d-separation. We explain the method using the following example before sketching how it might generalize to a larger class of causal models.

\begin{example}[A quantum cyclic causal model]
\label{example: qloop}
Consider the cyclic variation of the bipartite Bell causal structure illustrated in Figure~\ref{fig: Qcycle1}. Let the common cause $\Lambda$ correspond to the Bell state $\ket{\psi_{\Lambda}}=\frac{1}{\sqrt{2}}(\ket{00}+\ket{11})$. Suppose that $A$ and $B$ are the settings of local measurements on the two subsystems such that when these variables take the value $0$, it denotes a measurement in the $\{\ket{0},\ket{1}\}$ basis on the associated subsystem, and the value $1$ denotes a measurement in the $\{\ket{+},\ket{-}\}$ basis. $X$ and $Y$ are the binary outcomes of these measurements where $\ket{0}$ or $\ket{+}$ correspond to outcome 0 and $\ket{1}$ or $\ket{-}$ correspond to 1. The additional constraints coming from the causal loop are that $B=X$ and $A=Y$. This specifies all the causal mechanisms, how do we calculate the observed distribution $P_{XYAB}$?
\end{example}

\paragraph{A method based on post-selection: } One method is to first calculate the observed correlations for the specified state and measurements in the original Bell scenario (Figure~\ref{fig: Bell2}), and then post-select on the observations that obey the loop conditions $B=X$ and $A=Y$. More formally, this corresponds to transforming the original cyclic causal structure of Figure~\ref{fig: Qcycle1} to the acyclic causal structure of Figure~\ref{fig: Qcycle2} by cutting off the edges $A\longrsquigarrow X$ and $B\longrsquigarrow Y$ and replacing them with the edges $A^*\longrsquigarrow X$ and $B^*\longrsquigarrow Y$ by introducing two exogenous nodes $A^*$ and $B^*$. Then the inputs $A^*$ and $B^*$ and outputs $X$ and $Y$ along with the shared system $\Lambda$ define a Bell scenario, while the variables $A=Y$ and $B=X$ can simply be seen as local post-processings of the outcomes. We can then calculate the observed probabilities for this acyclic causal structure using the Born rule, and post-select on $A^*=A$ and $B^*=B$, which effectively achieves the post-selection $A=Y$ and $B=X$ in the original Bell scenario (Figure~\ref{fig: Bell2}). This distribution needs to be renormalised to obtain the observed distribution $P_{XYAB}$. This is calculated in Figure~\ref{fig: Qcycle} and can be used to find all the affects relations. An intervention on $A$ would cut off the arrow from $Y$ to $A$. This means $A$ does not affect $X$ since $A$ is effectively exogenous in the post-intervention causal structure and will be uncorrelated with $X$ since $\Lambda$ is the maximally entangled state. Similarly $B$ does not affect $Y$. However, $AB$ affects $XY$ since a joint intervention on $A$ and $B$ takes us back to the original Bell scenario in which these sets are correlated, and correlation in the post-intervention causal structure implies an affects relation (cf.\ Lemma~\ref{lemma: correl-affects}) and it can be checked that this affects relation is irreducible. We also have $X$ affects $B$ and $Y$ affects $A$ due to the loop conditions $A=Y$ and $B=X$. In addition, $XY$ affects $AB$, which is also irreducible. With a bit more effort, one can also check that we have $A$ affects $Y$ and $B$ affects $X$. Therefore we have two Type 1 affects causal loops (Definition~\ref{def: ACL1}) $\{A$ affects $Y,$ $Y$ affects $A\}$, and similarly for $B$ and $X$. We also have a Type 4 affects causal loop (Definition~\ref{def: ACL4}) formed by the irreducible relations $\{AB$ affects $XY$, $XY$ affects $AB\}$. In this example, $\cG_{\mathrm{do}(A,B)}$ corresponds to a quantum causal structure (the Bell scenario) while $\cG_{\mathrm{do}(X,Y)}$ is a simple classical causal structure. Then the (observed) arrows $\longrsquigarrow$ of $\cG$ can be classified into dashed and solid arrows as: $A\xdashrightarrow{} X$, $B\xdashrightarrow{} Y$, $X\longrightarrow B$ and $Y\longrightarrow A$. 
The post-intervention distribution is fully specified here because, all interventions (except that on $\Lambda$ alone) are associated with acyclic post-intervention graphs and for interventions on the exogenous $\Lambda$, the post and pre-intervention distributions coincide. \medskip

 \paragraph{Applying the method to fine-tuned explanations of non-classical correlations: } It is known that certain non-classical correlations arising in the bipartite Bell causal structure cannot be obtained in the same causal structure if the common cause $\Lambda$ was classical. However, these correlations can be easily generated in the classical, fine-tuned causal structure of Figure~\ref{fig: Qcycle3}, which differs from the original causal structure by the inclusion of fine-tuned causal influences from each party's input to the other party's output. We now explain how this is achieved and then apply the post-selection method explained above to create a causal loop in Figure~\ref{fig: Qcycle3} by adding $X\longrightarrow B$ and $Y\longrightarrow A$. This will demonstrate that, even though the same non-classical correlations and affects relations can be obtained in the original Bell causal structure and its fine-tuned classical counterpart~\ref{fig: Qcycle3}, the two causal structures behave differently in the presence of causal loops.

 First consider the PR box, which is one of the maximally non-classical correlations of the Bell causal structure. It is defined by the condition $X\oplus Y=A.B$ where all the variables are binary. This is easily generated in the classical causal structure of Figure~\ref{fig: Qcycle3} by the structural equations $\Lambda=E$, $Y=E$ and $X=E\oplus A.B$ (where $E$ is binary and uniformly distributed). Other non-classical correlations can be obtained by adding some ``noise'' to this PR box example. Let $\Lambda=(E,F)$ correspond to two variables $E$ and $F$ both binary, and the former distributed uniformly. Then the structural equations $Y=E$ and $X=E\oplus F \oplus A.B$ for different distributions over the exogenous variable $F$ correspond to the PR box mixed with different levels of noise.\footnote{Note that the model can be symmetrised by including an additional, uniformly distributed binary variable $G$ in the description of $\Lambda=(E,F,G)$ and using the structural equations $ X=E\oplus (G\oplus 1)(A.B\oplus F)$ and $Y=E\oplus G(A.B\oplus F)$.}
  \begin{align}
 \label{eq: Bellfinetune}
     \begin{split}
         X&=A.B\oplus E\oplus F,\\
          Y&=E.
     \end{split}
 \end{align}
Therefore, the causal mechanisms that allow us to produce non-classical correlations $P_{XYAB}$ in the acyclic causal structure~\ref{fig: Qcycle3} are the functional dependences~\eqref{eq: Bellfinetune} along with a specification of the distributions over the exogenous variables $E$ and $F$ that constitute $\Lambda$. $E$ is uniform while $F$ can vary depending on the correlation to be generated. We now construct the causal loop by including the additional arrows $X\longrsquigarrow B$ and $Y\longrsquigarrow A$ and by effectively post-selecting on the loop condition $A=Y$ and $B=X$. These, along with the causal mechanisms~\eqref{eq: Bellfinetune} of the acyclic case define the mechanisms for the cyclic causal structure.  
We will now see that these causal mechanisms are incompatible with each other. We have $Y=E$, $X=E\oplus F\oplus A.B$, $A=Y$ and $B=X$, which gives $X=E.X\oplus E\oplus F$ and $Y=E$. Therefore for $(E,F)=(0,0)$, we have $(X,Y)=(0,0)$ and for $(E,F)=(0,1)$ we have $(X,Y)=(1,0)$. However for $(E,F)=(1,0)$ we get $X=X\oplus 1$ which does not have a solution. For $(E,F)=(1,1)$ we get $X=X$ which is not a unique solution. Therefore if we demand unique solvability, we must require $E=0$ deterministically which contradicts the initial assumption that $E$ is uniform. Even if we do not require uniqueness, we can not have $(E,F)=(1,0)$ and forbidding this would make $E$ and $F$ correlated and non-uniform. 

Therefore, in the classical, fine-tuned explanation of the Bell correlations, adding the loop is not consistent with the causal mechanisms that generate the non-classical correlations in the absence of the loop--- in particular, they are in conflict with the preparation of the exogenous variable $\Lambda$. If we have a consistent loop, then intervention on $A$ and $B$ will no longer recover the original non-classical correlations. This is in contrast to the faithful case analysed in Figure~\ref{fig: Qcycle} (and explained previously in the text), when $\mathrm{do}(A,B)$ gives back the non-classical correlations of the Bell scenario. This suggests that certain (non-local) hidden variable explanations for quantum correlations (in a Bell experiment) can in principle be distinguished from the explanation provided by standard quantum mechanics in the presence of causal loops. We have only shown this for a particular set of functions or causal mechanisms for generating the former and it would be interesting to consider if this generalises, in particular to causal mechanisms provided by Bohmian mechanisms~\cite{Bohm1952}, a non-local hidden variable theory.

\afterpage{\FloatBarrier}
\begin{figure}[t!]
    \centering
\subfloat[]{\includegraphics[]{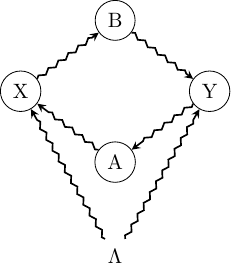}\label{fig: Qcycle1}}\qquad\subfloat[]{\includegraphics[]{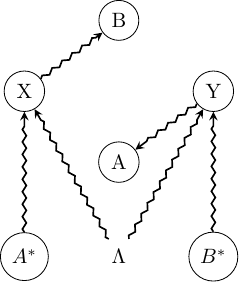}\label{fig: Qcycle2}}\qquad\subfloat[]{\includegraphics[]{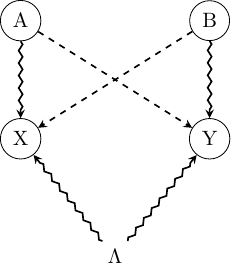}\label{fig: Qcycle3}}\qquad\qquad\subfloat[]{\includegraphics[]{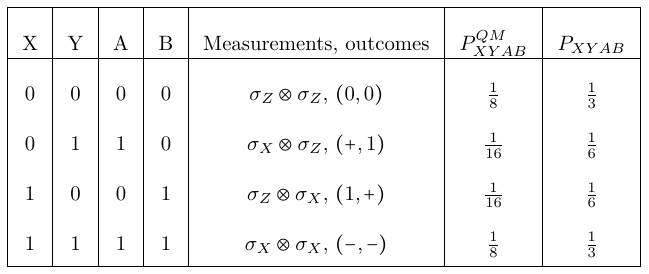}\label{fig: Qcycle4}}
    \caption[A cyclic quantum causal model]{\textbf{A cyclic quantum causal model: } \textbf{(a)} A cyclic variation of the bipartite Bell causal structure (Figure~\ref{fig: Bell2}). \textbf{(b)} A method to calculate the observed distribution of (a) when $\Lambda$ is non-classical involves this intermediate causal structure. This is obtained from (a) by copying the nodes $A$ and $B$ and removing the directed cycle as shown. This gives an acyclic causal structure for which the distribution $P_{XYABA^*B^*}$ can be calculated using known methods. Then, post-selecting on $A=A^*$ and $B=B^*$ gives the distribution $P_{XYAB}$ for the original cyclic causal structure of (a). \textbf{(c)} A classical causal, fine-tuned structure that can generate, all non-classical correlations of the bipartite Bell causal structure. Creating a causal loop in this case by adding the arrows $X\longrsquigarrow B$ and $Y\longrsquigarrow A$ does not lead to the same predictions as (a), which corresponds to adding these arrows to the original Bell causal structure. This method explained in the main text. \textbf{(d)} The table provides the observed distribution for Example~\ref{example: qloop} calculated using the proposed method. The only values of $A$, $B$, $X$ and $Y$ that are compatible with the loop conditions $A=Y$ and $B=X$ are those listed here, and the fifth column lists the measurements and outcomes that these values correspond to, according to Example~\ref{example: qloop}. $P_{XYAB}^{QM}$ denotes the probabilities of the measurements and outcomes listed in the fifth column calculated using the Born rule. These values are sub-normalised, and upon renormalisation, the observed distribution $P_{XYAB}$ for the cyclic causal structure (a) is obtained. Note that the d-separation condition~\ref{definition: compatdist} is satisfied in this case.}
\label{fig: Qcycle}
\end{figure}

\paragraph{Generalising to other causal structures: } The idea behind the post-selection method employed for Example~\ref{example: qloop} above can in principle be generalised to other non-classical, cyclic causal structures where every directed cycle includes at least one edge $W\longrsquigarrow Z$ connecting classical nodes $W$ and $Z$. The intuition is that cutting off such an edge in every directed cycle and replacing it with an edge $W^*\longrsquigarrow Z$, by introducing an additional, exogenous variable $W^*$ would result in a directed acyclic graph (DAG). One can then apply the generalised causal model framework of~\cite{Henson2014} to obtain the observed distribution in this DAG and then post-select on $W=W^*$ for all the edges that were cut off. Then a way to recover the d-separation condition (using the result of~\cite{Forre2017}) would be to check whether these exists a classical causal model for the same cyclic causal structure that produces identical observed correlations and satisfies the auSEP property. Note that this classical causal model need not yield the same post-intervention distributions. In the example of Figure~\ref{fig: Qcycle1}, an intervention on $A$ and $B$ gives the Bell scenario, which as we know produces non-classical correlations that cannot be obtained in the corresponding classical causal model~\cite{Wood2015}. Finally, it would be interesting to compare this method with the framework of post-selected closed time-like curves~\cite{Lloyd2011}. 

\begin{remark}
We note that assuming the d-separation condition of Definition~\ref{definition: compatdist} as a primitive property of the framework rules out certain cyclic causal structures from being described in our current framework. In the classical case, these are precisely those cyclic causal models that do not satisfy auSEP or those involving continuous random variables (due to the result of~\cite{Forre2017}). An example of such a causal model is given in~\cite{Neal2000}, and~\cite{Forre2017} proposes a generalisation of d-separation called $\sigma-$separation through which they derive a \emph{generalised global directed Markov} condition that applies to classical causal models involving continuous variables and/or do not satisfy auSEP. This reduces to d-separation in the acyclic case. Therefore, one option would be to replace d-separation with $\sigma$-separation in Definition~\ref{definition: compatdist} to generalise our framework for cyclic causal models. Doing so would not affect the results of the main paper, but would only enlarge the class of causal models to which they can be applied.
\end{remark}

\section{Proofs of all results}
\label{appendix: ProofsJamming}

\subsection{Proofs of Lemma~\ref{lemma: CI} and Theorem~\ref{theorem:dorules}}
\label{appendix: proofs1}
\CI*

\begin{proof}
The conditional independence $S_1\indep S_2|S_3$ stands for $P_{S_1S_2|S_3}=P_{S_1|S_3}P_{S_2|S_3}$, which implies
\begin{equation}
\label{eq: lemma11}
    P_{S_1|S_2S_3}=P_{S_1|S_3}.
\end{equation}
The three d-separation relations $S\perp^d S_i$ for $i\in \{1,2,3\}$ imply that $S$ is d-separated from every subset of the union $S_1 S_2 S_3$. This implies the following independences by Definition~\ref{definition: compatdist} of compatibility of the distribution $P$ with the causal model represented by $\mathcal{G}$,
\begin{equation}
    \label{eq: lemma12}
    P_{S|S'}=P_S\quad \forall S'\subseteq S_1 S_2  S_3.
\end{equation}
Now consider the conditional distribution $P_{S_2|SS_1S_3}$. We have,
\begin{align}
\label{eq: lemma13}
    \begin{split}
       P_{S_2|SS_1S_3}&=\frac{P_{S_2SS_1S_3}}{P_{SS_1S_3}} \\
       &=\frac{P_{S_3}P_{S_2|S_3}P_{S_1|S_2S_3}P_{S|S_1S_2S_3}}{P_{SS_1S_3}}\\
       &=\frac{P_{S_3}P_{S_2|S_3}P_{S_1|S_3}P_{S}}{P_{S}P_{S_1S_3}}\\
       &=P_{S_2|S_3},
    \end{split}
\end{align}
where we have used Equations~\eqref{eq: lemma11} and~\eqref{eq: lemma12} in the third line, noting that $P_{S|S_1S_3}=P_S\Rightarrow P_{SS_1S_3}=P_{S}P_{S_1S_3}$. Equation~\eqref{eq: lemma13} is equivalent to $P_{SS_1S_2|S_3}=P_{SS_1|S_3}P_{S_2|S_3}$ which denotes the conditional independence $S S_1\indep S_2|S_3$. The conditional independence $S_1\indep S S_2|S_3$ can be derived analogously due to the symmetry between $S_1$ and $S_2$. 

Finally, we have
\begin{align}
        P_{S_2|SS_3}=\frac{P_{S_2SS_3}}{P_{SS_3}}=\frac{P_{S_3}P_{S_2|S_3}P_{S|S_2S_3}}{P_{S}P_{S_3}}=P_{S_2|S_3},
\end{align}
and similarly $P_{S_1|SS_3}=P_{S_1|S_3}$. Together with Equation~\eqref{eq: lemma13}, this implies $P_{S_2|SS_1S_3}=P_{S_2|SS_3}$. This is equivalent to $P_{S_1S_2|SS_3}=P_{S_1|SS_3}P_{S_2|SS_3}$ which denotes the final conditional independence $S_1\indep S_2|S S_3$.
\end{proof}

\DoRules*

\begin{proof}
\textbf{Rule 1:} We first note that the graph $\cG_{\mathrm{do}(X)}$ differs from $\cG_{\overline{X}}$ only by the inclusion of the additional nodes $I_{X_i}$ and corresponding edge $I_{X_i}\longrsquigarrow X_i$ for each $X_i\in X$. Therefore, the d-separation relation $(Y\perp^d Z|X W)_{\cG_{\overline{X}}}$  for the latter implies the same relation $(Y\perp^d Z|XW)_{\cG_{\mathrm{do}(X)}}$ for the former. 
Using the compatibility condition of Definition~\ref{definition: compatdist} for the graph $\cG_{\mathrm{do}(X)}$, this implies the conditional independence of $Y$ and $Z$ given $XW$ for the distribution $P_{\cG_{\mathrm{do}(X)}}$ i.e.,   $P_{\cG_{\mathrm{do}(X)}}(y,z|x,w)=P_{\cG_{\mathrm{do}(X)}}(y|x,w)P_{\cG_{\mathrm{do}(X)}}(z|x,w)$. This conditional independence is equivalently expressed by the required Equation~\eqref{eq: rule1}.

\bigskip
\textbf{Rule 2:} $\cG_{\overline{X},\underline{Z}}$ is the graph where all incoming arrows to $X$ and outgoing arrows from $Z$ are removed in $\cG$. Hence, the d-separation condition $(Y\perp^d Z|XW)_{\cG_{\overline{X},\underline{Z}}}$ implies that the only paths between $Y$ and $Z$ in the graph $\cG_{\overline{X}}$ that are not blocked by $X$ and $W$ are paths involving an outgoing arrow from $Z$. These are precisely the paths that get removed in going from $\cG_{\overline{X}}$ to $\cG_{\overline{X},\underline{Z}}$, resulting in the required d-separation there. 
The same statement holds for the graph $\cG_{\mathrm{do}(X)}$ (by the argument used in the proof of Rule 1), and also for the graph $\cG_{\mathrm{do}(X), I_Z}$ which corresponds to adding the nodes $I_{Z_i}$ and edges $I_{Z_i}\longrsquigarrow Z_i$ to $\cG_{\mathrm{do}(X)}$ for each $Z_i\in Z$. The latter holds true since the addition of the $I_{Z_i}$ nodes and $I_{Z_i}\longrsquigarrow Z_i$ edges cannot create any additional paths between $Z$ and $Y$ that are left unblocked by $X$ and $W$. This implies that the only paths between $Y$ and the set $I_Z:=\{I_{Z_i}\}_i$ not blocked by $X$ and $W$ in $\cG_{\mathrm{do}(X), I_Z}$ are paths from $I_Z$, going through $Z$ and involving an outgoing arrow from $Z$ i.e., paths involving the subgraph $I_Z\longrsquigarrow Z \longrsquigarrow\ldots$. All these paths would get blocked when conditioning additionally on $Z$. This gives $(Y\perp^d I_Z|XWZ)_{\cG_{\mathrm{do}(X), I_Z}}$, which through the compatibility condition (Definition~\ref{definition: compatdist}) implies the conditional independence $(Y\indep I_Z|XWZ)_{\cG_{\mathrm{do}(X), I_Z}}$, equivalently expressed as
\begin{equation}
    \label{eq:rule2proof}
    P_{\cG_{\mathrm{do}(X), I_Z}}(y|x,w,z,I_Z=\mathrm{idle})=P_{\cG_{\mathrm{do}(X), I_Z}}(y|x,w,z,I_Z=\mathrm{do}(z)) \quad \forall y,x,w,z.
\end{equation}

Using Equations~\ref{eq: intervention1} and~\ref{eq: intervention2}, we have $P_{\cG_{\mathrm{do}(X), I_Z}}(y|x,w,z,I_Z=\mathrm{idle})=P_{\cG_{\mathrm{do}(X)}}(y|x,w,z)$ and $P_{\cG_{\mathrm{do}(X), I_Z}}(y|x,w,z,I_Z=\mathrm{do}(z))=P_{\cG_{\mathrm{do}(X Z)}}(y|x,w,z)$ respectively $\forall y,x,w,z$. Along with Equation~\eqref{eq:rule2proof}, this gives the required Equation~\ref{eq: rule2}. In other words, once $X$, $W$ and $Z$ are given, $Y$ does not depend on whether the given value $z$ of $Z$ was obtained through an intervention ($I_Z=\mathrm{do}(z)$) or passive observation (i.e., where $I_{Z_i}=\mathrm{idle}$ for all $i$, which is the causal model where no interventions are made on elements of $Z$).

\bigskip
\textbf{Rule 3:} Consider the graph $\cG_{\mathrm{do}(X)I_Z}$ which is the post-intervention graph with respect to the nodes $X$ augmented with $I_{Z_i}\longrsquigarrow Z_i$ for all $Z_i \in Z$. In this graph, suppose we had the d-separation relation $(Y\perp^d I_Z|XW)_{\cG_{\mathrm{do}(X)I_Z}}$. By Definition~\ref{definition: compatdist}, this would result in the conditional independence $(Y\indep I_Z|XW)_{\cG_{\mathrm{do}(X)I_Z}}$ which can be expressed as
$$P_{\cG_{\mathrm{do}(X)I_Z}}(Y|W,X,I_Z=\mathrm{idle})=P_{\cG_{\mathrm{do}(X)I_Z}}(Y|W,X,I_Z=\mathrm{do}(z))\quad \forall z$$
Using the defining rules~\eqref{eq: intervention1} and \eqref{eq: intervention2} then gives $P_{\cG_{\mathrm{do}(X)I_Z}}(Y|W,X,I_Z=\mathrm{idle})=P_{\cG_{\mathrm{do}(X)}}(Y|W,X)$ and $P_{\cG_{\mathrm{do}(X)I_Z}}(Y|W,X,I_Z=\mathrm{do}(z))=P_{\cG_{\mathrm{do}(X Z)}}(Y|W,X, Z=z)$ $\forall z$, and consequently $P_{\cG_{\mathrm{do}(X Z)}}(Y|W,X,Z)=P_{\cG_{\mathrm{do}(X)}}(Y|W,X)$ which is the required Equation~\eqref{eq: rule3}. Therefore, showing that the d-separation condition $(Y\perp^d Z|XW)_{\cG_{\overline{X},\overline{Z(W)}}}$ implies the d-separation relation $(Y\perp^d I_Z|XW)_{\cG_{\mathrm{do}(X)I_Z}}$ would complete the proof. This is shown by contradiction. Suppose that $(Y\perp^d Z|XW)_{\cG_{\overline{X},\overline{Z(W)}}}$ and $(Y\not\perp^d I_Z|XW)_{\cG_{\mathrm{do}(X)I_Z}}$. Then there must exist a path from a member $I_{Z_i}$ of $I_Z$ to a member $Y_j$ of $Y$ in $\cG_{\mathrm{do}(X)I_Z}$
that is unblocked by $X$ and $W$. There are two possibilities for such a path: either it contains the subgraph $I_{Z_i}\longrsquigarrow Z_i \longrsquigarrow \ldots Y_j$ or the subgraph $I_{Z_i}\longrsquigarrow Z_i \longlsquigarrow \ldots Y_j$. Denoting these possibilities as cases 1 and 2 respectively, let $\mathscr{P}$ be the shortest such path. We will show that a contraction arises in each case.

 \emph{Case 1:} Consider the first case where $\mathscr{P}$ contains the subgraph $I_{Z_i}\longrsquigarrow Z_i \longrsquigarrow \ldots Y_j$. Note $(Y\not\perp^d I_Z|XW)_{\cG_{\mathrm{do}(X)I_Z}}$ (which we have assumed) implies $(Y\not\perp^d Z_i|XW)_{\cG_{\mathrm{do}(X)}}$. Along with the assumption that $(Y\perp^d Z_i|XW)_{\cG_{\overline{X},\overline{Z(W)}}}$, this implies that there exists a path from $Z_i$ to $Y$ in $\cG_{\mathrm{do}(X)}$ unblocked by $X$ and $W$ that passes through some member $Z_k$ of $Z(W)$ which would blocked when the incoming arrows to $Z_k$ are removed. This leads to the following subcases where the path from $Z_i$ to $Y_j$ in $\cG_{\mathrm{do}(X)}$ contains the following subgraphs:
    \begin{itemize}
        \item \emph{Case 1a:} $Z_i\longrsquigarrow\ldots\longrsquigarrow Z_k\longlsquigarrow\ldots Y_j$ or
        \item \emph{Case 1b:}$Z_i\longrsquigarrow\ldots\longlsquigarrow Z_k\longlsquigarrow\ldots Y_j$ or
         \item \emph{Case 1c:} $Z_i\longrsquigarrow\ldots\longrsquigarrow Z_k\longrsquigarrow\ldots Y_j$
    \end{itemize}

None of these can occur for the following reasons. In \emph{Case 1a}, some descendant of $Z_k$ must be in $W$ for the path to be unblocked in $\cG_{\mathrm{do}(X)}$ but by definition, $Z(W)$ (which contains $Z_k$) is the set of all nodes in $Z$ that do not have descendants in $W$. In \emph{Case 1b}, the path between $Z_i$ and $Z_k$ must contain a collider. For this path to be unblocked by $X$ and $W$ in $\cG_{\mathrm{do}(X)}$ the collider node must have a descendant in $W$ but the other requirement that this path must be blocked in $\cG_{\overline{X}, \overline{Z(W)}}$ implies that the same collider node must be a member of $Z(W)$ which by definition does not have any descendants in $W$, yielding a contradiction. In \emph{Case 1c}, there is either a directed path from $Z_k\in Z$ to $Y_j\in Y$ in $\cG_{\mathrm{do}(X)}$ or a collider in the path between $Z_k$ and $Y_j$. The latter is ruled out by the same argument used in Case 1b. If there is a directed path from $Z_k$ to $Y_j$, then there is a directed path from $I_{Z_k}$ to $Y_j$ in $\cG_{\mathrm{do}(X)I_Z}$ i.e., there is a path from a member of $Z$ to $Y$ that is unblocked by $X$ and $W$ in $\cG_{\mathrm{do}(X)I_Z}$ and that is shorter than the shortest path $\mathscr{P}$, which is not possible.

Finally, consider \emph{Case 2} where the path $\mathscr{P}$ contains the subgraph $I_{Z_i}\longrsquigarrow Z_i \longlsquigarrow \ldots Y_j$. The initial assumption that $(Y\not\perp^d I_Z|XW)_{\cG_{\mathrm{do}(X)I_Z}}$ implies that the collider node $Z_i$ must have descendants in the conditioning set $W$ i.e., $Z_i\not\in Z(W)$. However, in this case we will violate the assumption that $(Y\perp^d Z|XW)_{\cG_{\overline{X},\overline{Z(W)}}}$. On the other hand, to satisfy this d-separation, we would require $Z_i\in Z(W)$ but this would violate $(Y\not\perp^d I_Z|XW)_{\cG_{\mathrm{do}(X)I_Z}}$. 
Hence we have shown that $(Y\perp^d Z|XW)_{\cG_{\overline{X},\overline{Z(W)}}}$ and $(Y\not\perp^d I_Z|XW)_{\cG_{\mathrm{do}(X)I_Z}}$ can never be simultaneously satisfied and hence that $(Y\perp^d Z|XW)_{\cG_{\overline{X},\overline{Z(W)}}}$ implies $(Y\not\perp^d I_Z|XW)_{\cG_{\mathrm{do}(X)I_Z}}$ which in turn implies the required Equation~\eqref{eq: rule3}.

\end{proof}

\subsection{Proofs of Lemmas~\ref{lemma:HOaffectsCause},~\ref{lemma:HOaffects1},~\ref{lemma:HOaffects2},~\ref{lemma:reduce},~\ref{lemma:reduce2},~\ref{lemma: HOaff_complete} and Corollary~\ref{corollary:HOaffectsCause2}}
\label{appendix: proofs2}
\HOaffectsCause*
\begin{proof}
\begin{enumerate}
    \item We prove this by contradiction. The relation $X$ is not a cause of $Y$ is equivalent to the absence of any directed paths from $X$ to $Y$ in $\cG$ i.e., $(X\perp^d Y)_{\cG_{\mathrm{do}(X)}}$ and consequently $(X\perp^d Y)_{\cG_{\mathrm{do}(X Z)}}$, for any subset $Z$ of observed nodes, pairwise disjoint to $X$ and $Y$. 
    Since $Z$ is effectively exogenous in $\cG_{\mathrm{do}(X Z)}$, $(X\perp^d Y)_{\cG_{\mathrm{do}(X Z)}}$ implies $(X\perp^d Y|Z)_{\cG_{\mathrm{do}(X Z)}}$. Applying Rule 3 of Theorem~\ref{theorem:dorules} (noting the relation between $\cG_{\bar{Z}\bar{X}}$ and $\cG_{\mathrm{do}(X Z)}$) to the latter implies that $P_{\cG_{\mathrm{do}(X Z)}}(Y|X,Z)=P_{\cG_{\mathrm{do}(Z)}}(Y|Z)$ which is equivalent to $X$ does not affect $Y$ given do$(Z)$.

    \item This follows from the first part of Lemma~\ref{lemma: HOaff_complete} (proven later in this appendix) and the first part of this lemma. By Lemma~\ref{lemma: HOaff_complete}, $X$ affects $Y$ given $\{\mathrm{do}(Z),W\}$ implies $X$ affects $YW$ given do$(Z)$, which in turn implies that $X$ must either be a cause of $Y$ or of $W$, by the first part, proven above. 
\end{enumerate}
\end{proof}

\HOaffectsA*
 \begin{proof}
  To establish the lemma, we show that it is not possible to have $Z$ does not affect $Y$ given $W$ \emph{and} $XZ$ does not affect $Y$ given $W$ whenever $X$ affects $Y$ given $\{\mathrm{do}(Z),W\}$. Writing these three conditions out we have
  \begin{subequations}
\begin{equation}
\label{eq:HOaffects1}
    P_{\mathcal{G}_{\mathrm{do}(Z)}}(Y|Z,W)= P_{\mathcal{G}}(Y|W),
\end{equation}
\begin{equation}
\label{eq:HOaffects2}
    P_{\mathcal{G}_{\mathrm{do}(X Z)}}(Y|X,Z,W)= P_{\mathcal{G}}(Y|W),
\end{equation}
\begin{equation}
\label{eq:HOaffects3}
         P_{\mathcal{G}_{\mathrm{do}(X Z)}}(Y|X,Z,W)\neq P_{\mathcal{G}_{\mathrm{do}(Z)}}(Y|Z,W).
\end{equation}
\end{subequations}
Equations~\eqref{eq:HOaffects1} and~\eqref{eq:HOaffects2} imply $P_{\mathcal{G}_{\mathrm{do}(X Z)}}(Y|X,Z,W)=P_{\mathcal{G}_{\mathrm{do}(Z)}}(Y|Z,W)$ in contradiction with Equation~\eqref{eq:HOaffects3}.
\end{proof}

\HOaffectsB*
\begin{proof}
By Lemma~\ref{lemma:HOaffects1}, if
$X$ affects $Y$ given $\{\mathrm{do}(Z),W\}$ then there are only three possibilities 1) $Z$ affects $Y$ given $W$ and $XZ$ does not affect $Y$ given $W$, 2) $Z$ does not affect $Y$ given $W$ and $XZ$ affects $Y$ given $W$, and 3) $Z$ affects $Y$ given $W$ and $XZ$ affects $Y$ given $W$ i.e., the only case where the required conclusion does not follow is 1). Then the proof will be complete if we show that whenever $X$ consists only of exogenous nodes the undesired case does not arise. We show this by establishing that for exogenous $X$, $Z$ affects $Y$ given $W$ implies $XZ$ affects $Y$ given $W$. Suppose by contradiction that $XZ$ does not affect $Y$ given $W$ i.e., $P_{\cG_{\mathrm{do}(X Z)}}(Y|X,Z,W)= P_{\cG}(Y|W)$. By the exogeneity of $X$, this becomes $P_{\cG_{\mathrm{do}(Z)}}(Y|X,Z,W)=P_{\cG}(Y|W)$ or equivalently, $P_{\cG_{\mathrm{do}(Z)}}(Y,X,Z,W)= P_{\cG}(Y|W)P_{\cG_{\mathrm{do}(Z)}}(X,Z,W)$. Summing over values of $X$ and rearranging gives $P_{\cG_{\mathrm{do}(Z)}}(Y|Z,W)= P_{\cG}(Y|W)$ which is equivalent to $Z$ does not affect $Y$ given $W$. Therefore $Z$ affects $Y$ given $W$ implies $XZ$ affects $Y$ given $W$ whenever $X$ is exogenous.
\end{proof}

\Reduce*
\begin{proof}
By definition, if $X$ affects $Y$ given $\{\mathrm{do}(Z),W\}$ is reducible, then there exists a proper subset $s_X$ of $X$ such that $s_X$ does not affect $Y$ given $\{\mathrm{do}(Z \tilde{s}_X),W\}$. We now show that for every such $s_X$, its complement $\tilde{s_X}:=X\backslash s_X$ is such that $\tilde{s}_X$ affects $Y$ given $\{\mathrm{do}(Z),W\}$. We show this by contradiction. Assume that $\tilde{s}_X$ does not affect $Y$ given $\{\mathrm{do}(Z),W\}$ while $X$ affects $Y$ given $\{\mathrm{do}(Z),W\}$ and $s_X$ does not affect $Y$ given $\{\mathrm{do}(Z \tilde{s}_X),W\}$. Explicitly, these correspond to the following conditions, noting that $s_X \tilde{s}_X=X$:
\begin{subequations}
\begin{equation}
\label{eq:reduce1}
    P_{\mathcal{G}_{\mathrm{do}(\tilde{s}_X Z)}}(Y|\tilde{s}_X,Z,W)= P_{\mathcal{G}_{\mathrm{do}(Z)}}(Y|Z,W),
\end{equation}
\begin{equation}
\label{eq:reduce2}
    P_{\mathcal{G}_{\mathrm{do}(X Z)}}(Y|X,Z,W)\neq P_{\mathcal{G}_{\mathrm{do}(Z)}}(Y|Z,W),
\end{equation}
\begin{equation}
\label{eq:reduce3}
     P_{\mathcal{G}_{\mathrm{do}(X Z)}}(Y|X,Z,W) =  P_{\mathcal{G}_{\mathrm{do}(\tilde{s}_X Z)}}(Y|\tilde{s}_X,Z,W).
\end{equation}
\end{subequations}
Equations~\eqref{eq:reduce1} and~\eqref{eq:reduce3} imply that $P_{\mathcal{G}_{\mathrm{do}(X Z)}}(Y|X,Z,W)= P_{\mathcal{G}_{\mathrm{do}(Z)}}(Y|Z,W)$, which contradicts Equation~\eqref{eq:reduce2}.
\end{proof}

\ReduceB*
\begin{proof}
The proof is similar to that of Lemma~\ref{lemma:reduce}. $X_1$ affects $Y$ given $\{\mathrm{do}(Z),W\}$ and $X_2$ does not affect $Y$ given $\{\mathrm{do}(Z X_1),W\}$ are equivalent to
\begin{subequations}
\begin{equation}
     P_{\mathcal{G}_{\mathrm{do}(X_1 Z)}}(Y|X_1,Z,W)\neq P_{\mathcal{G}_{\mathrm{do}(Z)}}(Y|Z,W),
\end{equation}
\begin{equation}
     P_{\mathcal{G}_{\mathrm{do}(X_1 X_2 Z)}}(Y|X_1,X_2,Z,W)= P_{\mathcal{G}_{\mathrm{do}(X_1 Z)}}(Y|X_1,Z,W).
\end{equation}
\end{subequations}
These yield $ P_{\mathcal{G}_{\mathrm{do}(X_1 X_2 Z)}}(Y|X_1,X_2,Z,W)\neq P_{\mathcal{G}_{\mathrm{do}(Z)}}(Y|Z,W)$ which is equivalent to $X_1 X_2$ affects $Y$ given $\{\mathrm{do}(Z),W\}$.
\end{proof}

\HOaffComplete*
\begin{proof}
\begin{enumerate}
    \item We prove this through the contrapositive. Suppose that $X$ does not affect $YW$ given do$(Z)$ i.e.,
    \begin{equation}
        P_{\cG_{\mathrm{do}(X Z)}}(Y,W|X,Z)= P_{\cG_{\mathrm{do}(Z)}}(Y,W|Z)
    \end{equation}
    Summing over values of $Y$ on both sides, we have $P_{\cG_{\mathrm{do}(X Z)}}(W|X,Z)= P_{\cG_{\mathrm{do}(Z)}}(W|Z)$ i.e., $X$ does not affect $W$ given do$(Z)$. Hence
    \begin{equation}
    \label{eq: HOaffComplete1}
    \begin{split}
         P_{\cG_{\mathrm{do}(X Z)}}(Y,W|X,Z)/P_{\cG_{\mathrm{do}(X Z)}}(W|X,Z)&= P_{\cG_{\mathrm{do}(Z)}}(Y,W|Z)/P_{\cG_{\mathrm{do}(Z)}}(W|Z)\\
         \Rightarrow  P_{\cG_{\mathrm{do}(X Z)}}(Y|X,Z,W)&= P_{\cG_{\mathrm{do}(Z)}}(Y|Z,W),
    \end{split}
    \end{equation}
        which is equivalent to $X$ does not affect $Y$ given $\{\mathrm{do}(Z),W\}$.
    \item Suppose that $X$ affects $Y$ given $\{\mathrm{do}(Z),W\}$ is irreducible i.e.,
    \begin{equation}
        P_{\cG_{\mathrm{do}(X Z)}}(Y|X,Z,W)\neq P_{\cG_{\mathrm{do}(\tilde{s}_X Z)}}(Y|\tilde{s}_X,Z,W),\qquad \forall s_X\subset X
    \end{equation}
    where $s_X\tilde{s}_X:=X$. If $X$ affects $YW$ given $\mathrm{do}(Z)$ is reducible, then there exists a partition of $X=s_X\tilde{s}_X$ such that 
    \begin{equation}
        P_{\cG_{\mathrm{do}(X Z)}}(Y,W|X,Z)= P_{\cG_{\mathrm{do}(\tilde{s}_X Z)}}(Y,W|\tilde{s}_X,Z).
    \end{equation}
    As in the proof of part 1., this implies that $P_{\cG_{\mathrm{do}(X Z)}}(Y|X,Z,W)= P_{\cG_{\mathrm{do}(\tilde{s}_X Z)}}(Y|\tilde{s}_X,Z,W)$, which contradicts the first equation. Therefore $X$ affects $Y$ given $\{\mathrm{do}(Z),W\}$ is irreducible implies $X$ affects $YW$ given $\mathrm{do}(Z)$ is irreducible.
    \item For the forward direction, it is again convenient to use the contrapositive, i.e., to show that $X$ does not affect $Y$ given $\{\mathrm{do}(Z),W\}$ and $X$ does not affect $W$ given $\mathrm{do}(Z)$ imply $X$ does not affect $YW$ given $\mathrm{do}(Z)$.  The first two statements are
    \begin{align*}
        P_{\cG_{\mathrm{do}(X Z)}}(Y|X,Z,W)&= P_{\cG_{\mathrm{do}(Z)}}(Y|Z,W)\quad\text{and}\\
        P_{\cG_{\mathrm{do}(X Z)}}(W|X,Z)&= P_{\cG_{\mathrm{do}(Z)}}(W|Z).
    \end{align*}
Multiplying these gives
\begin{align*}
    P_{\cG_{\mathrm{do}(X Z)}}(Y|X,Z,W)P_{\cG_{\mathrm{do}(X Z)}}(W|X,Z)=P_{\cG_{\mathrm{do}(Z)}}(Y|Z,W)P_{\cG_{\mathrm{do}(Z)}}(W|Z)
\end{align*}
which rearranges to
\begin{align*}
    P_{\cG_{\mathrm{do}(X Z)}}(Y,W|X,Z)=P_{\cG_{\mathrm{do}(Z)}}(Y,W|Z),
\end{align*}
which is $X$ does not affect $YW$ given $\mathrm{do}(Z)$.

For the reverse direction, we note that we have shown $X$ affects $W$ given do$(Z)$ implies $X$ affects $YW$ given do$(Z)$ in the proof of part 1 of this lemma. From the main statement of part 1. we also have $X$ affects $Y$ given $\{\mathrm{do}(Z),W\}$ implies $X$ affects $YW$ given do$(Z)$. Therefore we have $X$ affects $Y$ given $\{\mathrm{do}(Z),W\}$ or $X$ affects $W$ given do$(Z)$ implies $X$ affects $YW$ given do$(Z)$. \qedhere
\end{enumerate}
\end{proof}

\HOcauseB*
\begin{proof}
\begin{enumerate}
    \item Given that $X$ affects $Y$ given do$(Z)$ is irreducible, we know that for every $s_X\subset X$, 
    $s_X$ affects $Y$ given do$(Z \tilde{s}_X)$, where $s_X \tilde{s}_X:=X$. In particular, this means that for every element $e_X\in X$,  $e_X$ affects $Y$ given do$(Z \tilde{e}_X)$. Then by using Lemma~\ref{lemma:HOaffectsCause}, we know that $e_X$ is a cause of $Y$, which by Definition~\ref{def: cause} means that there exists a directed path from $e_X$ to at least one element $e_Y\in Y$ which in turn means that $e_X$ is a cause of $e_Y$.
    \item  By parts 1 and 2 of Lemma~\ref{lemma: HOaff_complete}, $X$ affects $Y$ given $\{\mathrm{do}(Z),W\}$ implies $X$ affects $YW$ given do$(Z)$ and the irreducibility of the former implies the irreducibility of the latter which in turn implies (by the first part of the current lemma) that for every $e_X\in X$, there exists $e_{YW}\in YW$ such that $e_X$ is a cause of $e_{YW}$.
\end{enumerate}
\end{proof}

\subsection{Proofs of Theorems~\ref{theorem:poset},~\ref{theorem: cycles} and Lemma~\ref{lemma: NoAffectsLoops}}
\label{appendix: proofs3}
\CompatST*
\begin{proof}
\emph{1. }  If \textbf{compat}$(\cS,\mathscr{A})$ holds then $\cR_{\cX}=\overline{\cF}(\cX)$ for all $\cX\in\cS$. Hence by Definition~\ref{definition: subsets} of accessible regions for sets of ORVs, we have \textbf{compat1$'$($\cS, \mathscr{A}$)}. The remaining affects relations in $\mathscr{A}'$ are all of the form $\cX$ affects $\cX'$ where $\cX'$ is a copy of $\cX$, and so, since the location of $\cX'$ is in $\cR_{\cX}=\overline{\cF}(\cX)$, \textbf{compat1$'$($\cS', \mathscr{A}'$)} also holds.

\bigskip
\noindent
\emph{2. } \textbf{compat1$'$($\cS', \mathscr{A}'$)} when applied to the affects relations of the form $\cX$ affects $\cX'$ when $\cX'$ is a copy of $\cX$ tells us that $\overline{\cF}(\cX')\subseteq\overline{\cF}(\cX)$ for every copy $\cX'$ of $\cX$, while Definition~\ref{definition:accreg} tells us that $\overline{\cF}(\cX')\subseteq\cR_{\cX}$ for every copy $\cX'$ of $\cX$.  If $\cR_{\cX}\nsubseteq\overline{\cF}(\cX)$ then it would be possible for a copy of $X$ to be accessible outside its future, and hence that $\overline{\cF}(\cX')\nsubseteq\overline{\cF}(\cX)$, contradicting \textbf{compat1$'$($\cS', \mathscr{A}'$)}. Therefore $\cR_{\cX}\subseteq\overline{\cF}(\cX)$ must hold.

\end{proof}

\Cycles*
\begin{proof}

Noting that all affects causal loops of Types 1, 2, 3 and 4 are also affects causal loops of Type 5, proving the theorem for ACL5 and ACL6 would imply the required result for ACL1, $\ldots$, ACL6.
\begin{enumerate}
    \item \emph{Proof for ACL5} Applying Corollary~\ref{corollary:HOaffectsCause2} to all affects relations in $S_i\subseteq \hat{S}_i$, $i=1,\ldots,n$ such that $\{\hat{S}_1$ affects $S_2$, $\hat{S}_2$ affects $S_3$, $\ldots$,$\hat{S}_{n-1}$ affects $S_n$, $\hat{S}_n$ affects $S_1\} \subseteq \mathscr{A}$, we know that each element of $\hat{S}_i$ must be a cause of some element of $S_{i+1\!\!\mod n}$. Following the chain, this implies that each element $e^1\in S_1\subseteq \hat{S}_1$ is a cause of some element $e^2\in S_1$. If $e^2=e^1$ we are done.  If not, we can continue the chain from $e^2$ until we return to an element $e^3\in S_1$. If $e^3=e^1$ or $e^3=e^2$ we are done; otherwise we continue. Since $S_1$ is finite, we must eventually return to an element of $S_1$ we already considered, establishing a causal loop.

       \item \emph{Proof for ACL6} Applying Corollary~\ref{corollary:HOaffectsCause2} to the first condition of ACL6 (Definition~\ref{def: ACL6}) we have that for every RV $e^1\in s_1$, there exists an RV $e_2\in S_2$ such that $e^1$ is a cause of $e_2$. Applying the Corollary~\ref{corollary:HOaffectsCause2} to the second condition, we have that $e_2\in S_2\subseteq \hat{S}_2$ must be a cause of some element $e^2\in s_1$. Either $e^1=e^2$ and we are done or we continue the chain as in the proof for ACL5.
\end{enumerate}
\end{proof}

\NoAffectsLoops*
\begin{proof}
\begin{enumerate}
    \item By definition, any set of affects relations that does not contain an affects causal loop is such that the cyclicity of the underlying causal structure is not guaranteed by the affects relations. In other words, it is possible to have the same set of affects relations in a causal model with an acyclic causal structure $\mathcal{G}$. Every causal model over an acyclic causal structure admits a non-trivial space-time embedding since an acyclic causal structure is a directed acyclic graph (DAG) and every DAG implies a partial order. This embedding would be such that the causal arrows $\longrsquigarrow$ of $\cG$ flow from past to future in the embedded space-time, which ensures no signalling outside the space-time's future.
    \item  In faithful causal models, any RV $X$ is a cause of an RV $Y$ if and only if $X$ affects $Y$. [This follows because if $X$ is a cause of $Y$ then $X \nperp Y$ in $\cG_{\mathrm{do}(X)}$. Since faithful, $X\not\indep Y$ in $\cG_{\mathrm{do}(X)}$ and then Lemma~\ref{lemma: correl-affects} gives $X$ affects $Y$. The converse is Lemma~\ref{lemma:HOaffectsCause} (which does not rely on faithfulness).]
    The existence of a causal loop between $X$ and $Y$ corresponds to $X$ being a cause of $Y$ and $Y$ being a cause of $X$ which is equivalent to $X$ affects $Y$ and $Y$ affects $X$. The latter is the definition of a Type~1 affects causal loop (Definition~\ref{def: ACL1}). Hence, under the faithfulness assumption, the absence of a Type~1 ACL is equivalent to the acyclicity of the underlying causal structure. As argued in part 1 above, any acyclic causal structure can be non-trivially and compatibly embedded in any space-time structure.
\end{enumerate}
\end{proof}

\subsection{Proofs of Lemmas~\ref{lemma:HOappendix1} and~\ref{lemma:HOappendix2}}
\label{appendix: proofs4}

\HOappendixA*
\begin{proof}
\begin{enumerate}
    \item We use Definition~\ref{definition: compatdist} on the d-separation relation $(XZW\perp^d Y)_{\cG_{\mathrm{do}(X Z)}}$ to obtain the conditional independence
\begin{equation}
    P_{\cG_{\mathrm{do}(X Z)}}(Y|X,Z,W)= P_{\cG_{\mathrm{do}(X Z)}}(Y).
\end{equation}
Then noting that $(XZW\perp^d Y)_{\cG_{\mathrm{do}(X Z)}}$ implies $(XZ\perp^d Y)_{\cG_{\mathrm{do}(X Z)}}$, we can apply Corollary~\ref{corollary:dsep-affects} to the latter d-separation relation to obtain $P_{\cG_{\mathrm{do}(X Z)}}(Y)=P_{\cG}(Y)$. Combined with the above equation, this gives 
\begin{equation}
    P_{\cG_{\mathrm{do}(X Z)}}(Y|X,Z,W)= P_{\cG}(Y).
\end{equation}
Now, we show that $P_{\cG}(Y)=P_{\cG}(Y|W)$ must hold in this case, which would (using the above equation) imply that $XZ$ does not affect $Y$ given $W$. Suppose by contradiction that $P_{\cG}(Y)\neq P_{\cG}(Y|W)$, which would imply that $(Y\nperp^d W)_{\cG}$. The assumed d-separation $(XZW\perp^d Y)_{\cG_{\mathrm{do}(X Z)}}$ implies that $(W\perp^d Y)_{\cG_{\mathrm{do}(X Z)}}$. The only way that we could have d-connection between $Y$ and $W$ in $\cG$ but not in $\cG_{\mathrm{do}(X Z)}$ is through the existence of a directed path between $XZ$ and $Y$ in $\cG$ which gives $(XZ\perp^d Y)_{\cG_{\mathrm{do}(X Z)}}$, contradicting our assumption $(XZW\perp^d Y)_{\cG_{\mathrm{do}(X Z)}}$. This establishes the first part. 
\item We show that $(XZW\perp^d Y)_{\cG_{\mathrm{do}(X Z)}}$ implies $(ZW\perp^d Y)_{\cG_{\mathrm{do}(Z)}}$, which in turn implies that $Z$ does not affect $Y$ given $W$. Then along with the first part of the lemma, this gives us $(XZW\perp^d Y)_{\cG_{\mathrm{do}(X Z)}}$ $\Rightarrow$ $XZ$ does not affect $Y$ given $W$ and  $Z$ does not affect $Y$ given $W$. Then using Lemma~\ref{lemma:HOaffects1}, this implies that $X$ \emph{does not affect} $Y$ given $\{\mathrm{do}(Z),W\}$, which is the required conclusion.

Suppose that $(XZW\perp^d Y)_{\cG_{\mathrm{do}(X Z)}}$ but $(ZW\nperp^d Y)_{\cG_{\mathrm{do}(Z)}}$. There are two ways that this is possible
\begin{enumerate}
    \item[(i)] $(Z\nperp^d Y)_{\cG_{\mathrm{do}(Z)}}$ : By assumption, we have $(XZW\perp^d Y)_{\cG_{\mathrm{do}(X Z)}}$, which implies $(Z\perp^d Y)_{\cG_{\mathrm{do}(X Z)}}$. The only way we can then have $(Z\nperp^d Y)_{\cG_{\mathrm{do}(Z)}}$ is through the existence of a directed path from $X$ to $Y$ in $\cG_{\mathrm{do}(Z)}$. This gives $(X\perp^d Y)_{\cG_{\mathrm{do}(X Z)}}$, which contradicts our assumption.
    \item[(ii)] $(W\nperp^d Y)_{\cG_{\mathrm{do}(Z)}}$ : The assumption $(XZW\perp^d Y)_{\cG_{\mathrm{do}(X Z)}}$ implies $(W\perp^d Y)_{\cG_{\mathrm{do}(X Z)}}$. If the d-connection $(W\nperp^d Y)_{\cG_{\mathrm{do}(Z)}}$ is due to a directed path from $W$ to $Y$ in $\cG_{\mathrm{do}(Z)}$, this path must go through $X$ in order to ensure that $(W\perp^d Y)_{\cG_{\mathrm{do}(X Z)}}$. However, this would violate the original assumption $(XZW\perp^d Y)_{\cG_{\mathrm{do}(X Z)}}$ as it would lead to a directed path from $X$ to $Y$ in $\cG_{\mathrm{do}(X Z)}$. On the other hand, if the d-connection $(W\nperp^d Y)_{\cG_{\mathrm{do}(Z)}}$ is due to a common cause, it is not possible to have the d-connection $(W\perp^d Y)_{\cG_{\mathrm{do}(X Z)}}$, which also contradicts the assumed d-separation.
\end{enumerate}
The above establishes that $(XZW\perp^d Y)_{\cG_{\mathrm{do}(X Z)}}$ implies $(ZW\perp^d Y)_{\cG_{\mathrm{do}(Z)}}$, and $(ZW\perp^d Y)_{\cG_{\mathrm{do}(Z)}}$ implies $Z$ does not affect $Y$ given $W$ follows from the first part of the proof (with $Z$ playing the role of $XZ$). 
\end{enumerate}

\end{proof}

\HOappendixB*

\begin{proof}
\begin{enumerate}
    \item 
    The given dependence $(XZW\nindep Y)_{\cG_{\mathrm{do}(X Z)}}$ is equivalent to 
\begin{equation}
    \exists x,x',y,z,z',w,w'\ \ \text{s.t.}\ \  P_{\cG_{\mathrm{do}(X Z)}}(Y=y|X=x,Z=z,W=w)\neq P_{\cG_{\mathrm{do}(X Z)}}(Y=y|X=x',Z=z',W=w')
\end{equation}
     Suppose that $XZ$ does not affect $Y$ given $W$ i.e.,
\begin{equation}
    P_{\cG_{\mathrm{do}(X Z)}}(Y=y|X=x,Z=z,W=w)=P_{\cG}(Y=y|W=w)\qquad\forall x,y,z,w
\end{equation}
It is not possible to satisfy both of these equations and $(XZW\nindep Y)_{\cG_{\mathrm{do}(X Z)}}$ must imply $XZ$ affects $Y$ given $W$. 
    \item Firstly, the d-separation $(ZW\perp^d Y)_{\cG_{\mathrm{do}(Z)}}$ implies that $Z$ does not affect $Y$ given $W$, which follows from part 1. of Lemma~\ref{lemma:HOappendix1}. From part 1. above, we have $(XZW\nindep Y)_{\cG_{\mathrm{do}(X Z)}}$ implies $XZ$ affects $Y$ given $W$. We now show that $Z$ does not affect $Y$ given $W$ and $XZ$ affects $Y$ given $W$ implies that $X$ affects $Y$ given $\{\mathrm{do}(Z),W\}$, which would complete the proof. Writing out the first two conditions, we have
\begin{equation}
    P_{\cG_{\mathrm{do}(Z)}}(Y|Z,W)=P_{\cG}(Y|W),
\end{equation}
\begin{equation}
     P_{\cG_{\mathrm{do}(X Z)}}(Y|X,Z,W)\neq P_{\cG}(Y|W).
\end{equation}
Together, these imply that $P_{\cG_{\mathrm{do}(X Z)}}(Y|X,Z,W)\neq P_{\cG_{\mathrm{do}(Z)}}(Y|Z,W)$ i.e., $X$ affects $Y$ given $\{\mathrm{do}(Z),W\}$.
\end{enumerate}
\end{proof}

\subsection{Proof of Theorem~\ref{lemma: CyclesAppendix}}
\label{appendix: proofs5}
\CyclesAppendix*
\begin{proof}

The proofs for ACL7 and ACL8 are similar. We describe the proof for ACL8 here, and at the end explain how it also applies to ACL7.
Applying Corollary~\ref{corollary:HOaffectsCause2} to the affects relations $\{S'_2$ affects $S_3$, $S'_3$ affects $S_4$, $\ldots$ ,$S'_{n-1}$ affects $S_n$, $S'_n$ affects $s_1\} \subseteq \mathscr{A}$ in the second condition of ACL8 (Definition~\ref{def: ACL8}) we have that for each element $e'_2\in S_2$ there exists an element $e_3\in S_3$ of which it is a cause, for each element $e'_3\in S_3$ there exist an element $e_4\in S_4$ of which it is a cause, $\ldots$ , for each element $e'_n\in S_n$ there exist an element $e_1\in s_1\subseteq S_1$ of which it is a cause. This does not immediately imply that there is a directed path from $S'_2$ to $s_1$, since for example the element $e'_3\in S_3$ of which $e'_2\in S_2$ is a cause might not belong to the next set $S'_3$ in the chain, i.e., we could have $e_3\in S_3\backslash (S_3\bigcap S'_3)$ if $(S_3,S_3')$ forms an incomplete node of $\mathscr{C}_{e_2}$. In this case, the third condition of Definition~\ref{def: ACL8} tells us that there is another complete affects chain $\mathscr{D}^{\mathscr{C}}_{e_2}$ that connects $S_3\backslash (S_3\bigcap S'_3)$ to $S_3$. Since this is a complete affects chain, we can apply the same argument as in the proof of Theorem~\ref{theorem: cycles} to conclude that for each element in $S_3\backslash (S_3\bigcap S'_3)$, there exists an element $e^*_3\in S_3$ of which it is a cause. We consider two cases depending on whether we have $e^*_3\in S_3\backslash (S_3\bigcap S'_3)$ or $e^*_3\in S_3\bigcap S'_3$. We will show that in the former case, the affects relations in the secondary chain $\mathscr{D}^{\mathscr{C}}_{e_2}$ already guarantees cyclicity while the latter case, these (the set formed by such secondary chains, one for every incomplete node) guarantee cyclicity when taken together with those in the primary chain $\mathscr{C}_{e_2}$. 

In the first case, $\mathscr{D}^{\mathscr{C}}_{e_2}$ corresponds to a Type 5 affects causal loop since it involves a complete chain of irreducible affects relations from a set $S_3\backslash (S_3\bigcap S'_3)$ on to itself. The cyclicity claim for this case then follows from Theorem~\ref{theorem: cycles}. Therefore, we now consider the case where for each incomplete node $(S_i,S_i')$ of $\mathscr{C}_{e_2}$, the corresponding element $e^*_i\in S_i$ belongs to the intersection of the sets $S_i\bigcap S'_i$. Then, applying Corollary~\ref{corollary:HOaffectsCause2} repeatedly to each pair of affects relations in $\{S'_2$ affects $S_3$, $S'_3$ affects $S_4$, $\ldots$ ,$S'_{n-1}$ affects $S_n$, $S'_n$ affects $s_1\} \subseteq \mathscr{A}$, we can conclude that for every element $e'_2\in S_2'$, there exists an element $e_1\in s_1\subseteq S_1$ such that $e_2'$ is a cause of $e_1$. By Definition~\ref{def: ACL8} (second condition), we considered such a set $\mathscr{A}$ of affects relations for every element $e_2\in S_2$, defining $S_2'$ such that it contains $e_2$. Since the above argument holds for all sets of affects relations $\mathscr{A}$ defined as above and for all elements of $S'_2$, this implies that for every element $e_2\in S_2$, there exists a corresponding element $e_1\in S_1$ of which it is a cause. Applying  Corollary~\ref{corollary:HOaffectsCause2} to the first condition of Definition~\ref{def: ACL8} i.e., the irreducible affects relation $S_1$ affects $S_2$, we have that for every element $e_1\in S_1$, there exists a corresponding element $e_2\in S_2$ of which it is a cause. This was also the case for ACL1-6 as shown in Theorem~\ref{theorem: cycles}, so the statement of the present theorem then follows from the proof of Theorem~\ref{theorem: cycles}.

For ACL7, the first condition says that there is an irreducible affects relation $S_1$ affects $S_2$ in $\mathscr{A}$ and the second condition of Definition~\ref{def: ACL7} guarantees the existence of an affects chain from $s_2\subseteq S_2$ to $S_1$. The subtlety here is to note that if $s_2\subset S_2$, then $(S_2,S'_2)$ will be an incomplete node of $\mathscr{C}_{s_2}$. By the above proof for ACL8, we have concluded that the affects relations $\{S'_2$ affects $S_3$, $S'_3$ affects $S_4$, $\ldots$ ,$S'_{n-1}$ affects $S_n$, $S'_n$ affects $s_1\} \subseteq \mathscr{A}$ along with the third condition of ACL8 (which is similar for ACL7) either imply cyclicity of the causal structure or that for every element $e'_2\in S_2'$, there exists an element $e_1\in s_1\subseteq S_1$ such that $e_2'$ is a cause of $e_1$. If the node $(S_2,S'_2)$ is also incomplete as noted above, one can extend the same arguments using the third condition to conclude that either the causal structure is cyclic or for every element $e_2\in S_2$, there exists an element $e_1\in s_1\subseteq S_1$ such that $e_2$ is a cause of $e_1$. The same condition was obtained at the end of the previous paragraph, in the proof for ACL8, and shown to imply cyclicity. Therefore this establishes the theorem also for ACL7.
\end{proof}

\newpage

\end{document}